\documentclass[notitlepage,a4paper,12pt]{article}%
\usepackage{verbatim}
\usepackage{etoolbox}
\usepackage{multirow}
\usepackage{adjustbox}
\usepackage{float}
\usepackage{threeparttable}
\usepackage[LGR,T1]{fontenc}
\usepackage[utf8]{inputenc}%
\usepackage[greek.polutoniko,english]{babel} 
\usepackage[largesc,defaultsups
]{newtxtext}
\usepackage[varqu,varl]{zi4}
\usepackage[type1]{cabin}
\usepackage[vvarbb,amsthm]{newtxmath} 
\DeclareFontFamilySubstitution{LGR}{\rmdefault}{Tempora} 
\usepackage{superiors}
            \newcommand{\shortminus}{-}
\usepackage{anyfontsize}
\usepackage{mathfixs}
\usepackage{accents}
\usepackage[a4paper]{geometry}
\usepackage[onehalfspacing]{setspace}
\usepackage[header,title]{appendix}
\usepackage{graphicx}
\usepackage[inline]{enumitem}
        \newcommand{\colonequals}{\coloneq}
\usepackage{chngcntr}
\usepackage{afterpage}
\usepackage{etoolbox}
\usepackage{ifpdf}
\usepackage{xcolor}
\usepackage{bbold}
\usepackage{romanbar}
\usepackage[Algorithm]{algorithm}
\usepackage{algpseudocode}
\usepackage{relsize}
\usepackage{booktabs}
\usepackage{tabulary}
\usepackage{parskip}
\usepackage{xfrac}
\usepackage{xparse}
\NewDocumentCommand{\xroot}{om}{%
  \IfNoValueTF{#1}
    {\sqrt{\vphantom{#2}}#2}
    {\sqrt[#1]{\vphantom{#2}}#2}%
}
\usepackage{url}
\usepackage[authoryear,comma]{natbib}
\definecolor{MyBrown}{rgb}{0.3,0,0}
\definecolor{MyBlue}{rgb}{0,0,0.3}
\definecolor{MyRed}{rgb}{0.4,0,0.1}
\definecolor{MyGreen}{rgb}{0,0.4,0}
\usepackage[bookmarks=true,%
bookmarksnumbered=true,%
colorlinks=true,%
linkcolor=MyBlue,%
citecolor=MyRed,%
filecolor=MyBlue,%
urlcolor=MyGreen%
]{hyperref}
\newcommand{\periodafter}[1]{#1{.}}
\usepackage{titlesec}

\titleformat\paragraph[runin]
                        {\normalfont\normalsize\itshape}
                        {mypar}
                        {0pt}
                        {\periodafter}{}
\titlespacing\paragraph{\parindent}
                       {1ex minus 0.5ex}
                       {3pt}
\titleformat\subsubsection
{\normalfont\normalsize\bfseries}{\thesubsubsection}{1em}{}
\titlespacing*\subsubsection{0pt}{0pt}{0pt}
\definecolor{lightblue}{rgb}{.80,.95,1}
\providecommand{\U}[1]{\protect\rule{.1in}{.1in}}
\ifpdf
\fi
\patchcmd{\thmhead}{(#3)}{#3}{}{}

\newtheorem{assumption}{Assumption}
\newtheorem{corollary}{Corollary}
\newtheorem{handnum-corollary}{Corollary}
\newtheorem{definition}{Definition}
\newtheorem{lemma}{Lemma}
\newtheorem{handnum-lemma}{Lemma}
\newtheorem{proposition}{Proposition}
\newtheorem{example}{Example}
\newtheorem{handnum-proposition}{Proposition}

\newtheorem{remark}{Remark}

\newcommand{\ac}[1]{\textscale{0.9}{#1}}
\newcommand\abs[1]{\lvert{#1}\rvert}
\renewcommand{\qedsymbol}{\rule{1.0ex}{1.2ex}} 
\DeclareRobustCommand{\qed}{%
\ifmmode \mathqed
 \else
 \leavevmode\unskip\penalty9999 \hbox{}\nobreak
 \quad\hbox{\qedsymbol}
\fi
}
\makeatletter 
\renewenvironment{proof}[1][\proofname]{\par
\pushQED{\qed}%
\normalfont \topsep6\p@\@plus6\p@\relax
\trivlist
\item\relax
{\bfseries
#1\@addpunct{.}}\hspace\labelsep\ignorespaces
}{%
\popQED\endtrivlist\@endpefalse
}
\makeatother
\makeatletter
\newcommand*{\NoBreakPar}{\par\nobreak\@afterheading}
\makeatother

\newenvironment{figurenotes}[1][Note]{\begin{minipage}[t]{\linewidth}\footnotesize{\itshape#1: }}{\end{minipage}}

\DeclareMathOperator* \cov {cov}
\DeclareMathOperator* \var {var}

\DeclareMathOperator* \bal {bal} 

\usepackage[blocks]{authblk}

\usepackage{fancyhdr}

\fancypagestyle{fandl}{\fancyhf{}
        \cfoot{\thepage}}

\renewcommand\thesection{\arabic{section}}

\renewcommand\thesubsubsection{\thesection.\arabic{subsubsection}}


\let\myfootnote\footnote
\renewcommand{\footnote}[1]{\myfootnote{~#1}}

\begin{document}

\thispagestyle{fandl}

\renewcommand{\thefootnote}{\fnsymbol{footnote}}
\newcommand{\acknowledgements}{\hspace{0.1em} {Author Accepted Manuscript (not copyedited or formatted) for: Francesco Lancia, Alessia Russo, and Tim Worrall. "Intergenerational Insurance.'"  Forthcoming, \textit{Journal of Political Economy}, 2024.
Version of Record: \ac{DOI}: \href{https://dx.doi.org/10.1086/730206}{\urlstyle{rm}
  \nolinkurl{10.1086/730206}}. Original version circulated with the title ``Optimal Sustainable Intergenerational Insurance''.  We thank the editor and two anonymous referees along with Spiros Bougheas, Francesco Caselli, Gabrielle Demange, Mart\'{\i}n Gonzalez-Eiras, Sergey Foss, Alexander Karaivanov, Paul Klein, Dirk Krueger, Sarolta Lacz\'{o}, Costas Milas, Espen Moen, Iacopo Morchio, Nicola Pavoni, Jos\'e-V\'{\i}ctor R\'{\i}os-Rull, Karl Schlag, Kjetil Storesletten and Aleh Tsyvinski for helpful comments. The paper has also benefited from the comments of seminar participants at Cologne, the London School of Economics, New York University Abu Dhabi, Oslo and Warwick in addition to presentations at the \ac{NBER} Summer Institute on Macro Public Finance, the \ac{SED} Meeting in Edinburgh, the \ac{SAET} Conference in Faro, the \ac{CSEF-IGIER} Symposium on Economics and Institutions at Anacapri, the \ac{CEPR} European Summer Symposium in International Macroeconomics at Tarragona, the \ac{EEA-ESEM} Congress in Manchester, the Barcelona Graduate School of Economics Summer Forum and the Vienna Macroeconomics Workshop. Sergio Cappellini provided valuable research assistance. The second author gratefully acknowledges the support of \ac{STARS@UNIPD} Consolidator grant \ac{GENESIS} and the third author gratefully acknowledges the support of the \ac{UKRI} grant \ac{ES/L009633/1} and Leverhulme Trurst Research Fellowship \ac{RF-2023-415\textbackslash 7}. For the purpose of open access, the third author has applied a Creative Commons Attribution (\ac{CC-BY}) licence to any Author Accepted Manuscript version arising from this submission. Edited by Andrew Atkeson. }}
\title{
\textbf{Intergenerational Insurance}\footnote{\acknowledgements}}

\author[1]{Francesco Lancia}
\affil[1]{{\small \emph{Ca\textquotesingle\ Foscari University of Venice and the Centre for Economic Policy Research}}}
\author[2]{Alessia Russo}
\affil[2]{{\small \emph{University of Padua and the Centre for Economic Policy Research}}}
\author[3]{Tim Worrall}
\affil[3]{{\small \emph{University of Edinburgh}}}

\date{\small [Original Version, December 2020, Revised, December 2021, June 2023 and January 2024]}
\maketitle

\renewcommand{\abstractname}{}
\begin{abstract}
\noindent How should successive generations insure each other when the young can default on previously promised transfers to the old? This paper studies intergenerational insurance that maximizes the expected discounted utility of all generations subject to participation constraints for each generation. If complete insurance is unattainable, the optimal intergenerational insurance is history-dependent even when the environment is stationary.  The risk from a generational shock is spread into the future, with periodic \emph{resetting\/}. Interpreting intergenerational insurance in terms of debt, the fiscal reaction function is nonlinear and the risk premium on debt is lower than the risk premium with complete insurance.

\medskip



\textbf{JEL CODES}: D64; E21; H55.

\end{abstract}

\renewcommand*{\thefootnote}{\arabic{footnote}}
\setcounter{footnote}{0}
\pagebreak

\section{Introduction} \label{sec:intro}

Countries face economic shocks that result in unequal exposure to risk across generations. The Financial Crisis of 2008 and the \ac{COVID}-19 pandemic are two recent and notable examples.\footnote{\citet{Heathcoateetal20a} find that the Financial Crisis of 2008 had a negative impact on the older generation, while the young benefited from the fall in asset prices. \citet{Heathcoateetal20b} find that younger workers have been impacted to a greater extent by the response to the \ac{COVID}-19 pandemic because they disproportionately work in sectors that have been particularly adversely affected, such as retail and hospitality.} Faced with such shocks, it is desirable to share risk across generations. However, full risk sharing is not sustainable if it commits future generations to transfers they would not wish to make once they are born. The issue of the sustainability of intergenerational insurance is becoming increasingly relevant in many advanced economies as the relative standard of living of the younger generation has worsened in recent decades.\footnote{Part~A of the Supplemental Material reports changes in the relative standard of living of the young and the old for six \ac{OECD} countries using data from the Luxembourg Income Study Database.\label{fn:LIS}}
If this generational shift persists, future generations may be less willing to contribute to insurance arrangements than in the past. Therefore, a natural question to ask is how an optimal intergenerational insurance arrangement should be structured when there is limited enforcement of risk-sharing transfers. 

Despite its policy relevance, the literature on intergenerational insurance does not fully address this question. The normative approach in the literature investigates the optimal design of intergenerational insurance but assumes that transfers are mandatory, ignoring the issue of limited enforcement. Meanwhile, the positive approach highlights the political limits to intergenerational insurance while considering equilibrium allocations supported by a particular voting mechanism, which are not necessarily Pareto optimal.

In this paper, we examine optimal intergenerational insurance when subsequent generations can default on risk-sharing transfers promised to previous generations. We model the limited enforcement of transfers by assuming that transfers satisfy a participation constraint for each generation. This can be interpreted as requiring that the insurance arrangement be supported by each generation if put to a vote. An arrangement of risk-sharing transfers is \emph{sustainable\/} if it satisfies the participation constraint of every generation. \emph{Optimal\/} sustainable intergenerational insurance is the risk-sharing transfers that would be chosen by a benevolent social planner that maximizes the expected discounted utility of all generations subject to the participation constraints.

The model is simple. At each date, a new generation is born and lives for two periods. Each generation comprises a constant population of homogeneous agents with the population size normalized to one. Each agent receives an endowment of a single, nonstorable consumption good, both when young and old. Endowments are stochastic. Each generation is affected by an idiosyncratic shock (common to all agents within a generation) and an aggregate growth shock. We adopt the approach of \citet{Alvarez-Jermann01} and \citet{Krueger-Lustig10} and assume that preferences exhibit a constant coefficient of relative risk aversion (for simplicity, we concentrate on the case of logarithmic preferences) and that the idiosyncratic and growth shocks are independent and identically distributed. In this setting, the underlying economy is stationary. There are only two frictions. First, risk may not be allocated efficiently, even if the economy is dynamically efficient, because there is no market in which the young can share risk with previous generations \citep[see, for example,][]{Diamond77}. Second, the amount of risk that can be shared is limited because transfers between generations cannot be enforced. In particular, the old will not make a transfer to the young (since the old have no future). Conversely, the young may make a transfer to the old. However, the young will only do so if they receive promises for their old age that at least match their expected lifetime utility from autarky, and they anticipate that these promises will be honored by the next generation.

It is well known \citep[see, for example,][]{Aiyagari-Peled91} that if endowments are such that the young wish to defer consumption to old age at a zero net interest rate, then there are stationary transfers that improve upon autarky (Proposition~\ref{prop:st}). Under this condition, and assuming that the first-best transfers cannot be sustained, there is a trade-off between efficiency and providing incentives for the young to make transfers to the old. This trade-off is resolved by linking the utility the young are promised for their old age to the promise made to the young of the previous generation. The resulting optimal sustainable intergenerational insurance arrangement is history dependent, even though the economic environment is stationary.

To understand why there is history dependence, suppose that the first-best transfers would violate the participation constraint of the young in some endowment state. To ensure that the current transfer made by the young is voluntary, either the current transfer is reduced below the first-best level, or the promised transfers for their old age are increased. Both changes are costly since a lower current transfer reduces the amount of risk shared today while increasing the transfers promised to the current young for their old age tightens the participation constraints of the next generation and reduces the risk that can be shared tomorrow. Therefore, an optimal trade-off exists between reducing the current transfer and increasing the future promise. This trade-off depends both on the current endowment and the current promise. For example, consider some current endowment and a current promise such that the future promise for the same endowment state is higher than the current promise. If the same endowment state is repeated in the subsequent period, then the young in that period are called upon to make a larger transfer, which in turn requires a higher promise of future utility to them as well. Thus, the transfer depends not only on the current endowment but also on the past promise, and hence, the history of endowment shocks.

The optimal sustainable intergenerational insurance is found by solving a functional equation derived from the planner's maximization problem. The solution is characterized by policy functions for the consumption of the young (or equivalently, the transfer made to the old) and the future promised utility for their old age in each endowment state. Both policy functions depend on the current endowment and the current promise. For a given endowment, the consumption of the young is weakly decreasing in the current promise, while the future promise is weakly increasing in the current promise (Lemmas~\ref{lem:g} and~\ref{lem:c}). When the current endowment state is repeated, the policy function for the future promise has a unique fixed point, which (ignoring a boundary condition) equals the utility at the first-best outcome. Therefore, the future promise is higher than the current promise when it is less than the corresponding fixed point and lower than the current promise when the current promise is greater than the fixed point. When the promised utility is sufficiently low, there is some endowment state in which the participation constraint of the young does not bind. In that case, the future promise is \emph{reset\/} to the largest value that maximizes the planner's payoff.

When a participation constraint binds, the risk affecting one generation is spread to future generations. The resetting property shows, however, that the effect of a shock does not last forever. Moreover, it implies strong convergence to a unique invariant distribution (Proposition~\ref{prop:convergence}). The invariant distribution exhibits history dependence, and consumption fluctuates across states and over time, even in the long run. This starkly contrasts to the situation under either full enforcement of transfers or no risk. In the former case, the promised utility is constant over time, except possibly in the initial period (Proposition~\ref{prop:fbstm}). In the latter case, the promised utility is constant in the long run, although there may be an initial phase during which the promised utility falls (Proposition~\ref{prop:determ}). In both cases, the allocation is efficient in the long run. Thus, {\itshape both\/} risk and limited enforcement are necessary for history dependence and inefficiency in the long run.

Transfers to the old can be interpreted in terms of debt. Suppose that the planner issues one-period state-contingent bonds that trade at the state price determined by the corresponding intertemporal marginal rate of substitution and balances the budget by taxing or subsidizing the young. Given these bond prices and taxes, the young buy the correct quantity of state-contingent bonds to finance their optimal old-age consumption. It is then possible to use the model to study the dynamics of debt and to address the issues of debt valuation and sustainability, following the model-based approach introduced by \citet{Bohn95, Bohn98}.

When preferences are logarithmic, it is natural to measure debt relative to the endowment of the young. With debt measured in this way, there is a maximal debt limit and a debt policy function that determines the next-period debt as a function of the current debt and the next-period endowment share. This function is constant when debt is low but is nonlinear and strictly increasing when debt is above a critical threshold (Corollary~\ref{cor:debtf}). The debt policy function and the history of endowments determine the dynamics of debt. Debt rises or falls depending on the evolution of endowments but eventually resets to a minimum level, creating cycles of debt. The difference between debt and the revenue generated from issuing state-contingent bonds defines the fiscal reaction function that measures how the tax rate depends on debt. Absent enforcement frictions, the fiscal reaction function is linearly increasing in debt. However, with enforcement frictions, the fiscal reaction function is linear when debt is low but is nonlinear when debt is high. In particular, when debt is below the threshold, the amount of debt issued is independent of the current debt, while the price of state-contingent bonds decreases linearly in debt. Thus, bond revenue falls with debt, and the tax rate rises linearly. Above the threshold, two factors affect the fiscal reaction function. The price of state-contingent bonds decreases with debt, while bond issuance increases with debt according to the nonlinear debt policy function. The combined effect of these two factors results in a nonlinear fiscal reaction function.

The model also provides implications for asset pricing and the dependence of asset prices on debt (Proposition~\ref{prop:Ross}). Since the idiosyncratic and growth shocks are independent and identically distributed, the implied conditional yields are the sum of a growth-adjusted component and a constant given by the logarithm of the average growth rate. The price of state-contingent bonds decreases with debt, which implies that the conditional yields, including the risk-free rate, increase with debt. The discount factor of the planner and the average growth rate determine the yield on the long bond. However, the long-short spread may be positive or negative. The dynamics of debt imply that the long-short spread is positive when debt is low, and the young are poor because, in this case, debt will rise, leading to higher expected future yields. Likewise, the long-short spread is negative when debt is high, and the young are rich because debt will fall, leading to lower expected future yields.

The variability of yields and their decomposition into growth-adjusted and growth-dependent components is also significant for debt valuation. There is a linear decomposition of the risk premium on debt into a growth-adjusted component and a component depending on the aggregate risk (Proposition~\ref{prop:mrp}). The return on bonds increases with the endowment of the young next period, as does the marginal utility of consumption of the old next period. Thus, the return on bonds is positively correlated with the stochastic discount factor for a given debt, resulting in a risk premium on debt lower than the risk premium on aggregate risk. In the absence of enforcement frictions, this gap is zero. When there are enforcement frictions, debt is a hedge against the endowment risk, and this reduces the risk premium on debt. Consequently, for a fixed plan of future primary surpluses, higher debt can be sustained compared to a case where the future surpluses are discounted using the risk premium on aggregate risk. This gap between the risk premiums on aggregate risk and debt offers a potential resolution to the \textit{debt valuation puzzle\/} posed by \citet{Jiangetal19}, who find that the value of \ac{US} debt exceeds the present value of future primary surpluses when discounted by the risk premium on aggregate risk.\footnote{For an overview of the debt valuation and sustainability, see, for example, \citet{Reis22}, \citet{Willems-Zettelmeyer22} and \citet{Jiangetal23}.} Moreover, the risk premium on debt varies with debt. In particular, it rises or falls depending on whether the expected return on debt increases with debt at a faster or slower rate than the risk-free interest rate.

In an example with two endowment states, we provide a closed-form solution for the bound on the variability of the implied yields and show that the invariant distribution of debt is a transformation of a geometric distribution (Proposition~\ref{prop:canonic}). Numerically, the solution can be found using a shooting algorithm without the need to solve a functional equation. In this example, the risk premium increases with debt, leading to a reduction in the gap between the risk premium on aggregate risk and the risk premium on debt.

\paragraph{Literature} The paper builds on the literature on risk sharing in models with overlapping generations. In most of this literature, transfers are mandatory, and consideration is restricted to stationary transfers \citep[see, for example,][]{Shiller99,Rangel-Zeckhauser01}, in contrast to the voluntary and history-dependent transfers considered here. Our result on history dependence is foreshadowed in a mean-variance setting by \citet{Gordon-Varian88}, who establish that any time-consistent optimal intergenerational risk-sharing agreement is nonstationary. \citet{Ball-Mankiw07} analyze risk sharing when generations can trade contingent claims before they are born. They find that idiosyncratic shocks are spread equally across generations, and consumption follows a random walk, as in \citet{Hall78}. Such an allocation is not sustainable since it violates the participation constraint of some future generation almost surely. In contrast, we show that although the effects of a shock can be prolonged, they are unevenly spread
across future generations, and resetting ensures that they cannot last forever.

By interpreting the transfer to the old as debt, we complement the extensive literature on debt sustainability and the fiscal reaction function that began with \citet{Bohn95,Bohn98}. Our result on the nonlinearity of the fiscal reaction function echoes the discussion of \textit{fiscal fatigue\/}, which argues that the primary fiscal balance responds sluggishly to rising debt when debt is high because of the adverse implications of debt, such as the risk of default \citep[see, for example,][]{Mendoza-Ostry08, Ghoshetal13}. Despite the absence of default in our model, enforcement constraints generate nonlinearity in the fiscal reaction function. \citet{Bhandarietal17} also study optimal fiscal policy and debt dynamics but in a model with infinitely-lived and heterogeneous agents where markets are incomplete because of constraints on tax policy. \citet{Brunnermeieretal22a} provide a result similar to ours that the risk premium on debt is lower than the risk premium on aggregate risk. In their model, infinitely-lived agents must retain a fixed proportion of their idiosyncratic risk. Government debt serves as a hedge against idiosyncratic risk, and consequently, debt becomes a negative beta asset. The authors emphasize that debt can command a bubble premium, which may add to the safety of government debt. In contrast to \citet{Brunnermeieretal22a}, our model has no bubble component, and the extent of risk sharing is determined endogenously, depending on the history of endowment shocks.

Methodologically, the paper relates to the literature on risk sharing and limited enforcement frictions with infinitely-lived agents. Two polar cases have been examined: one with two infinitely-lived agents \citep[see, for example,][]{Thomas-Worrall88, Chari-Kehoe90, Kocherlakota96} and the other with a continuum of infinitely-lived agents \citep[see, for example,][]{Thomas-Worrall07a, Krueger-Perri11, Broer13}. The overlapping generations model considered here has a continuum of agents, but only two agents are alive at any point in time. The model is not nested in either of the two infinitely-lived agent models but fills an essential gap in the literature by analyzing optimal intergenerational insurance with limited enforcement frictions. Here, we establish strong convergence to the invariant distribution, whereas \citet{Krueger-Perri11} and \citet{Broer13} consider the solution only at an invariant distribution and \citet{Thomas-Worrall07a} discuss convergence only in a particular case.

\paragraph{Plan of Paper} Section~\ref{sec:model} sets out the model. Section~\ref{sec:bench} considers two benchmarks: one with full enforcement of transfers from the young to the old and the other without risk. Section~\ref{sec:opt} characterizes optimal sustainable intergenerational insurance and Section~\ref{sec:conv} establishes convergence to an invariant distribution on a countable ergodic set. Section~\ref{sec:debt} provides an interpretation of the optimum in terms of debt and derives the fiscal reaction function. Section~\ref{sec:risk} discusses the implications for asset pricing and Section~\ref{sec:debtval} considers the valuation of debt. Section~\ref{sec:example} presents an example with two endowment states. Section~\ref{sec:conc} concludes. The Appendix contains the proofs of the main results.\footnote{Additional proofs and further details can be found in the Supplemental Material.} 

\section{The Model}\label{sec:model}
        
Time is discrete and indexed by~$t=0,1,2,\ldots,\infty$. The model consists of a pure exchange economy with an overlapping generations demographic structure. At each time~$t$, a new generation is born and lives for two periods. The generation born at date~$t$ has a population of $N_t$ homogeneous agents. We assume that there is no population growth and normalize $N_t=1$, so it is as if each generation has a single agent.\footnote{The assumption that agents of the same generation are homogeneous makes it possible to focus on intergenerational risk sharing. However, it does mean that we ignore questions about inequality within generations and its evolution over time. Although we maintain the assumption of a constant population, the qualitative properties of the model are unchanged if there is a constant rate of population growth. Part~D of the Supplemental Material examines the impact of a demographic shock and shows how the effect of this shock can be amplified and prolonged.} Each agent is young in the first period of life and old in the second. The economy starts at $t=0$ with an initial old agent and an initial young agent. Since time is infinite, the initial old agent is the only agent that lives for just one period.

At each time~$t$, agents receive an endowment of a perishable consumption good. Endowments are finite and strictly positive. The endowment of the young and the old at time~$t$ are $e^y_t$ and $e^o_t$ with an aggregate endowment of $e_t=e^y_t+e^o_t$. The endowment \emph{share\/} of the young is $s_t\colonequals e^y_t/e_t$ (the endowment share of the old is $1-s_t$), and the gross \emph{growth\/} rate of the aggregate endowment is $\gamma_t\colonequals e_t/e_{t-1}$. There is both idiosyncratic (share of the generation's endowment) risk and aggregate (growth) risk. The sequences of random variables $(s_t\mathpunct{:} t\geq0)$ and $(\gamma_t\mathpunct{:}t\geq0)$ take values in finite sets $\mathcal{I}$ and $\mathcal{J}$, respectively, where $\abs{\mathcal{I}}=I\geq2$ and $\abs{\mathcal{J}}=J\geq1$. The pair $\rho_t\colonequals (s_t,\gamma_t)$ taking values in $\mathcal{P}\subseteq \mathcal{I}\times \mathcal{J}$ follows a finite-state, aperiodic, time-homogeneous Markov process with the probability of transiting from~$\rho_t$ to state~$\rho_{t+1}$ next period given by $\varpi(\rho_t, \rho_{t+1})$.

Denote the history of endowment shares and growth rates up to and including time~$t$ by $s^{t}:=(s_{0},s_{1},...,s_{t})  \in\mathcal{I}^{t}$ and $\gamma^{t}:=(\gamma_{0},\gamma_{1},...,\gamma_{t})  \in\mathcal{J}^{t}$ and let $\rho^{t}:=(  \rho_{0},\rho_{1},...,\rho_{t})  \in\mathcal{P}^{t}$. The distribution of $\rho_0$ is given by the function $\varpi(\rho_0)$ and the probability of reaching the history $\rho^{t}$ is $\varpi(\rho^{t})=\varpi(\rho^{t-1})\varpi(\rho_{t-1}, \rho_{t})$. Hence, the aggregate endowment at time~$t$ is the random variable $e_t=\prod_{k=0}^{t}\gamma_k$ with $\gamma_0=e_0$.

There is complete information. Endowments depend only on the current state, whereas consumption can, in principle, depend on the history of states. Denote the per-period consumption of the young by $C(\rho^{t})$ and the corresponding consumption share by $c(\rho^{t})=C(\rho^t)/e_t$. There is no technology to store the endowment from one period to the next, and hence, the aggregate endowment is consumed each period. Consequently, the per-period consumption of the old is $e_t-C(\rho^{t})$ and the corresponding consumption share is $1-c(\rho^{t})$. In autarky, agents consume only their own endowments, that is, the consumption share of the young is $s_t$, and the consumption share of the old is $1-s_t$ for all $t$ and $(\rho^{t-1},\rho_{t})$.

Each generation is born after that period's uncertainty is resolved when the growth rate of the economy and the endowment shares of the young and the old are known. Therefore, after birth, a generation only faces uncertainty in old age, and there is no insurance market in which the young can insure against their endowment risk. Let $\{C\} = \{C(\rho^t)\mathpunct{:} t\geq0, \rho^t \in \mathcal{P}^t\}$ denote a history-contingent consumption stream of the young. Then, the lifetime utility gain over autarky for a generation born after the history $\rho^t$ is:
\begin{gather*}
U\left(\{C\}\mathpunct{;} \rho^t\right) \colonequals u(C(\rho^t)) \shortminus u(e^y_t) + \beta\sum\nolimits_{\rho_{t+1}}\!\!\!\!\varpi
(\rho_t, \rho_{t+1})\left(u(e_{t+1}\shortminus C(\rho^t, \rho_{t+1}))\shortminus u(e^o_{t+1})\right),
\end{gather*}
where $u(\cdot)$ is the per-period utility function, common to both the young and the old, and $\beta\in(0,1]$ is the generational discount factor. We assume the per-period utility function is logarithmic, $u(\cdot)=\log(\cdot)$. Hence, the preferences of an agent can be expressed in terms of consumption and endowment shares. In particular, since $e^y_t=s_te_t$ and $C(\rho^t)=c(\rho^t)e_t$, it follows that $u(C(\rho^t)) - u(e^y_t)=\log(c(\rho^t))-\log(s_t)$ and $U(\{C\}\mathpunct{;} \rho^t) = U(\{c\}\mathpunct{;} \rho^t)$ where
\begin{gather*}
U\left(\{c\}\mathpunct{;} \rho^t\right) \colonequals \log(c(\rho^t)) \shortminus \log(s_t) + \beta\sum\nolimits_{\rho_{t+1}}\!\!\!\!\varpi
(\rho_t, \rho_{t+1})\left(\log(1\!\shortminus\! c(\rho^t, \rho_{t+1}))\shortminus\log(1\shortminus s_{t+1})\right).
\end{gather*}
We call the history-contingent stream of consumption shares $\{c\}= \{c(\rho^t)\mathpunct{:} t\geq0, \rho^t \in \mathcal{P}^t\}$ an \emph{intergenerational insurance rule\/} since it determines how consumption is allocated between the young and the old for any history $\rho^t$. Since storage is not possible and because the young are born after uncertainty is resolved, the only means of achieving intergenerational insurance is through transfers between the young and the old. We assume that there is a benevolent social planner who chooses an intergenerational insurance rule of history-contingent transfers to maximize a discounted sum of the expected utilities of all generations. Let the planner's expected discounted utility gain over autarky, conditional on the history~$\rho^t$, be
\begin{gather*}\label{eqn:Vc}
V\left(\{c\}\mathpunct{;} \rho^t\right) \colonequals \frac{\beta}{\delta}\left(\log(1\shortminus c(\rho^t))\shortminus\log(1\shortminus s_{t})\right) \!+\!\mathbb{E}_{t}\left[\sum\nolimits_{j=t}^{\infty}\delta^{t-j}U\left(\{c\}\mathpunct{;} \rho^{j}\right)\right]
\end{gather*}
where $\mathbb{E}_{t}$ is the expectation over future histories at time~$t$. The planner's discount factor is $\delta\in(0,1)$, and the weight on the utility of the initial old is $\beta/\delta$.\footnote{%
The assumption of geometric discounting for the planner is common \citep[see, for example,][]{Fahri-Werning07}. Using a weight of $\beta/\delta$ for the initial old preserves the same relative weights on the young and the old, including the initial old, in every period.}

To maximize the discounted sum of expected lifetime utilities, the planner must respect the constraint that transfers are voluntary.\footnote{The assumption that the transfer is voluntary can be interpreted as requiring that the intergenerational insurance rule is supported by each generation if put to a vote.} That is, the planner must respect the constraint that neither the old nor the young would be better off in autarky than adhering to the specified transfers for any history of shocks. For the old, this means they will not make a positive transfer to the young because there is no future benefit to offset such a transfer. Hence, the consumption of the young cannot exceed their endowment, or equivalently,
\begin{gather}\label{eqn:nonneg1}
c(\rho^t) \leq s_t \quad\text{for all $t\geq0$ and $\rho^t\in \mathcal{P}^t$.}
\end{gather}
The analogous participation constraint for the young requires that the conditional transfers promised for their old age sufficiently compensate for the transfer made when young so that they are no worse off than reneging on the transfer today and receiving the corresponding autarkic lifetime utility.  That is,
\begin{gather}
U\left(\{c\}; \rho^t\right) \geq 0 \quad \text{for all $t\geq0$ and $\rho^t\in \mathcal{P}^t$.} \label{IC}%
\end{gather}
For expositional simplicity, let the initial state $\rho_0$ be given.\footnote{The analysis is easily generalized to any given initial distribution $\varpi(\rho_0)$.} Hence, at $t=0$, the planner chooses $\{c\}$ to maximize: 
\begin{gather}
V\left(\{c\}; \rho_0\right)\mathpunct{,}\label{eqn:obj}
\end{gather}
subject to the constraint set $\Lambda\colonequals\{  \{c\} \mid \text{\eqref{eqn:nonneg1}  and  {\eqref{IC}}}\}$. Since utility is strictly concave, and the constraints in~\eqref{IC} are linear in utility, the planner's objective in equation~\eqref{eqn:obj} is concave and the constraint set $\Lambda$ is convex and compact.
\begin{definition}\label{def:sustain} An intergenerational insurance rule is \emph{sustainable\/} if $\{c\}\in\Lambda$.
\end{definition}

\begin{definition}\label{def:opt} An intergenerational insurance rule is \emph{optimal\/} if it is sustainable and it maximizes the objective in equation~\eqref{eqn:obj} subject to the constraint that the initial old receive a utility from their consumption share of at least $\bar{\omega}_0$:%
\begin{gather}
\log(1-c(\rho_{0}))\geq\bar{\omega}_0\mathpunct{.} \label{PK0}%
\end{gather}
\end{definition}
We introduce constraint~\eqref{PK0} with an exogenous initial target utility of $\bar{\omega}_0$ because it is useful when considering the evolution of the optimal sustainable intergenerational insurance rule in Section~\ref{sec:opt}.\footnote{The initial target utility may also depend on the initial state. Varying $\bar{\omega}_0$ traces out the Pareto frontiers that trade-off the utility of the old against the planner's valuation of the expected discounted utility of all future generations.} However, we will return to the case where the planner chooses the initial $\bar{\omega}_0$.

Since  $U(\{C\}\mathpunct{;} \rho^t) = U(\{c\}\mathpunct{;} \rho^t)$ and utility is logarithmic, the objectives and constraints are equivalent whether consumption is expressed in levels or shares. That is, the economy with stochastic growth is equivalent to an economy with a constant endowment and consumption expressed as shares of the aggregate endowment. The growth rate of the consumption levels is simply the growth rate of the consumption shares multiplied by the growth rate of the aggregate endowment.

\begin{remark}\label{rem:g}
For preferences that exhibit constant relative risk aversion, this equivalence property is well known to hold in models of idiosyncratic and aggregate risk with infinitely-lived agents \citep[see, for example,][]{Alvarez-Jermann01, Krueger-Lustig10}. An analogous extension can be shown to hold here by defining growth-adjusted transition probabilities and discount factors to satisfy the following:
\begin{gather*}
\hat{\varpi}(\rho_t, \rho_{t+1}) \colonequals \frac{\varpi(\rho_t,\rho_{t+1})(\gamma_{t+1})^{1\shortminus\alpha}}{\sum_{\rho_{t+1}}\varpi(\rho_t,\rho_{t+1})(\gamma_{t+1})^{1\shortminus\alpha}}
\enspace\mbox{and}\enspace \frac{\hat{\beta}(\rho_t)}{\beta}\! = \!\frac{\hat{\delta}(\rho_t)}{\delta}  \colonequals \sum_{\rho_{t+1}}\varpi(\rho_t,\rho_{t+1})(\gamma_{t+1})^{1\shortminus\alpha},
\end{gather*}
where $\alpha$ is the coefficient of relative risk aversion.
\end{remark}

In what follows, we assume that the shocks to endowment shares and growth rates are independent and are identically and independently distributed (hereafter, i.i.d.).
\begin{assumption}\label{ass:iid} \begin{enumerate*}[label=(\roman*)]
\item The state~$\rho$ is i.i.d.\ with the probability $\varpi(\rho)$.
\item The endowment share and the growth rate are independent, that is, $\varpi(\rho)=\pi(s)\varsigma(\gamma)$ where $\pi(s)$ and $\varsigma(\gamma)$ are the marginal distributions of the endowment shares and the growth rates.
\end{enumerate*}
\end{assumption}

By Part~(i) of Assumption~\ref{ass:iid}, the economy is stationary. We make this assumption to emphasize that the history dependence we derive below follows from the participation constraints rather than any feature of the economic environment itself.\footnote{The assumption of i.i.d.\ shocks is standard in OLG models where a generation may cover 20-30 years.} Since the terms $U(\{c\}\mathpunct{;} \rho^t)$ and $V(\{c\}\mathpunct{;} \rho^t)$ depend on the growth rates $\gamma_t$ and $\gamma_{t+1}$ only via the transition function $\varpi(\rho_t,\rho_{t+1})$, it follows that under Assumption~\ref{ass:iid} the consumption shares in any optimal sustainable intergenerational insurance rule depend only on the history of endowment shares $s^t$.

\begin{proposition}\label{prop:iid}
Under Assumption~\ref{ass:iid}, the consumption shares in any optimal sustainable intergenerational insurance rule depend only on the history $s^t$ and are independent of the history of growth shocks $\gamma^t$.
\end{proposition}

A similar result is well known from models with infinitely-lived agents \citep[see, again,][]{Alvarez-Jermann01, Krueger-Lustig10}.\footnote{Under Assumption~\ref{ass:iid} and preferences exhibiting constant relative risk aversion, the discount factors defined in Remark~\ref{rem:g} satisfy $\hat{\beta}/\beta= \hat{\delta}/\delta = \sum_{\gamma}\varsigma(\gamma)\gamma^{1-\alpha}$. If $\alpha\not=1$, then the planner's objective is finite provided $\delta\sum_{\gamma}\varsigma(\gamma)\gamma^{1-\alpha}<1$.} 

\paragraph{Preliminaries} Since there are $I\geq2$ states for the endowment share, order states such that $s(i)<s(i+1)$ for $i=1,\ldots, I-1$, so that, a higher state corresponds to a larger endowment share for the young. For convenience, we will refer to states $1,2,\ldots,I$ corresponding to the shares $s(1),s(2),\ldots, s(I)$ and to simplify notation will sometimes express variables as a function of~$i$ rather than~$s$.

Under Assumption~\ref{ass:iid}, the existence of a nonautarkic sustainable allocation can be addressed by considering small stationary transfers that depend only on the current endowment state. Denote the intertemporal marginal rate of substitution between the consumption share when young in state~$s$ and the consumption share when old in state~$r$ next period, evaluated at autarky, by $\hat{m}(s,r)\colonequals\beta s/(1-r)$ and let $\hat{q}(s,r)\colonequals \pi(r) \hat{m}(s,r)$. The terms $\hat{m}(s,r)$ and $\hat{q}(s,r)$ correspond to the stochastic discount factor and the state prices in an equilibrium model. Denote the $I \times I$~matrix of terms $\hat{q}(s,r)$ by $\hat{Q}$. A nonautarkic sustainable allocation exhausting the aggregate endowment and satisfying the participation constraints in~\eqref{eqn:nonneg1} and~\eqref{IC} exists whenever the Perron root of $\hat{Q}$ is greater than one \citep[see, for example,][]{Aiyagari-Peled91,Chattopadhyay-Gottardi99}. In this case, there exists a vector of strictly positive stationary transfers that improves the lifetime utility of the young in each state. Since the endowment states are i.i.d., the matrix $\hat{Q}$ has rank one, and the Perron root is its trace. We assume that the trace of $\hat{Q}$ is larger than the harmonic mean of the growth factors, $\bar{\gamma}\colonequals(\sum_\gamma \varsigma(\gamma)\gamma^{-1})^{-1}$.
\begin{assumption}\label{ass:sustain}%
$\sum\nolimits_{s\in\mathcal{I}}\hat{q}(s,s) > \bar{\gamma}$. 
\end{assumption}

If there is just one state with the young receiving a share $s$ of the aggregate endowment and no growth, then Assumption~\ref{ass:sustain} reduces to the standard Samuelson condition: $s>1/(1+\beta)$. In this case, it is well known that there are Pareto-improving transfers from the young to the old. Assumption~\ref{ass:sustain} is the generalization to the stochastic case and a natural assumption given that our focus is on transfers to the old.\footnote{%
A sufficient condition for Assumption~\ref{ass:sustain} to be satisfied is that the Frobenius lower bound, given by the minimum row sum of $\hat{Q}$, is greater than $\bar{\gamma}$. A row sum greater than $\bar{\gamma}$ implies that in autarky, the young would wish to save for their old age in each endowment state even if the net interest rate were zero.}
Given Assumption~\ref{ass:sustain}, it follows that the constraint set $\Lambda$ is nonempty.
\begin{proposition}
\label{prop:st} Under Assumption~\ref{ass:sustain}, there exists a nonautarkic and stationary sustainable intergenerational insurance rule.
\end{proposition}

Furthermore, we assume:
\begin{assumption}
\label{ass:LS} $s(1) \leq {\delta}/(\beta+\delta)$.
\end{assumption}
Assumption~\ref{ass:LS} provides a sufficient condition for the strong convergence result of Section~\ref{sec:conv}.
Since $\delta<1$, Assumption~\ref{ass:LS} implies that $s(1)<1/(1+\beta)$. That is, in the absence of growth, the state-wise Samuelson condition does not hold in every state, showing that our results do not depend on this property. In the terminology of \citet{Gale73}, the economy can be viewed as a mix of Samuelson and classic cases.

\section{Two Benchmarks}\label{sec:bench}

Before turning to the characterization of the optimal sustainable intergenerational insurance, it is helpful to consider two benchmark cases that illustrate the inefficiencies generated by the presence of limited enforcement and uncertainty. The first benchmark ignores the participation constraints of the young but not the participation constraints of the old. The second benchmark considers an economy without risk but requires that the planner respects the participation constraints of both the young and the old.

\paragraph{First Best}
Suppose the planner ignores the participation constraints of the young and let $\Lambda^{\ast}\colonequals\{  \{c\}\mid \text{\eqref{eqn:nonneg1}}\}$ denote the set of transfers without the constraints in~\eqref{IC}.\footnote{Hereafter, the asterisk designates the first-best outcome. Note that the first best could be defined by assuming that the planner ignores the participation constraints of both the young and the old. The reason for presenting the first best as we do is to show that this allocation is stationary. Hence, any history dependence of the optimal sustainable intergenerational insurance rule derives from the imposition of the participation constraints of the young.}

\begin{definition}An intergenerational insurance rule $\{c\}\in\Lambda^{\ast}$ is \emph{first best\/} if it maximizes the objective function~\eqref{eqn:obj} subject to constraint~\eqref{PK0}.
\end{definition}

It is easy to verify that at the first-best optimum:%
\begin{gather}
c^{\ast}(s^{t})=\min\left\{
\frac{\delta}{\beta+\delta}, s_t\right\} \quad\text{for all $t>0$ and $s^t\in \mathcal{S}^t$.}
\label{tau_fbm}%
\end{gather}
 Condition~\eqref{tau_fbm} shows that the consumption shares of the young are kept constant unless doing so would involve a transfer from the old to the young, in which case the consumption share is the autarky value.\footnote{Condition~\eqref{tau_fbm} is a special case of the familiar Arrow-Borch condition for optimal risk sharing modified to account for the constraint that transfers are only from the young to the old.} That is, at the first best, the consumption share is independent of the history $s^{t-1}$ and depends only on the current endowment share $s_t$ when the nonnegativity constraint on the transfer binds. Under Assumption~\ref{ass:LS}, there is always one state in which the participation constraint of the old holds with equality.

 It can be seen from condition~\eqref{tau_fbm} that for states in which transfers are positive, the first-best consumption share of the young is independent of~$s$. It is decreasing in $\beta$ since a higher $\beta$ puts more weight on the utility of the old who receive the transfer, and it is increasing in $\delta$ since a higher $\delta$ puts more weight on the utility of the young who make the transfer.

Let $\omega_{\min}(s) \colonequals \log(1-s)$ be the utility of the old at autarky and $\omega^{\ast}\colonequals \log(\beta/(\beta+\delta))$ be the utility of the old when the consumption share of the young is $\delta/(\beta+\delta)$. Then, $\omega^{\ast}(s)\colonequals \max\{\omega_{\min}(s), \omega^\ast\}$ is the utility of the old at the first-best solution when the endowment share of the young is~$s$. Since $s_0$ is the endowment share of the young at the initial date, it follows from Definition~\ref{def:opt} that constraint~\eqref{PK0} does not bind when $\bar{\omega}_0\leq\omega^{\ast}(s_0)$. In this case, the first-best consumption at $t=0$ is $c^{\ast}(s_{0})$, determined by condition~\eqref{tau_fbm} as in every other time~$t>0$. On the other hand, for $\bar{\omega}_0>\omega^{\ast}(s_0)$, constraint~\eqref{PK0} binds and $c^{\ast}(s_0) = 1- \exp(\bar{\omega}_0)$. In this case, the initial transfer to the old is correspondingly higher than implied by condition~\eqref{tau_fbm}.

Let $v^\ast(s)= \log(c^{\ast}(s))+(\beta/\delta)\log(1-c^\ast(s))$ denote the per-period payoff to the planner with the first-best allocation let $V^{\ast}(s_0,\omega)$ denote the expected discounted payoff to the planner when the initial endowment share is~$s_0$ and the initial utility of the old is~$\omega$. The maximum utility the old can get occurs if they consume all of the endowment, so that $\omega_{\max}=\log(1)=0$. Let $\Omega(s_0)=[\omega_{\min}(s_0), 0]$ be the set of possible utilities for the old at the initial state, $\bar{v}^{\ast}\colonequals\sum_{s}\pi(s)v^\ast(s)$ be the planner's expected per-period payoff at the first-best solution and $\bar{V}^\ast\colonequals\bar{v}^\ast/(1-\delta)$ be the corresponding continuation payoff. The first-best outcome is summarized in the following proposition.\footnote{The proof of Proposition~\ref{prop:fbstm} is omitted because it follows from standard arguments. Nonetheless, the properties of the function $V^\ast(s_0, \omega)$ are mirrored in Proposition~\ref{prop:determ} and Lemma~\ref{lemma:V_omega}, given below, which do respect the participation constraints of the young.}

\begin{proposition}
\label{prop:fbstm} \begin{enumerate*}[label=(\roman*)]
\item The consumption share $c^{\ast}(s^{t})$ is stationary and satisfies condition~\eqref{tau_fbm} for $t>0$. For $t=0$, $c^{\ast}(s_0)$ satisfies condition~\eqref{tau_fbm} for $\omega\leq \omega^\ast(s_0)$ and $c^\ast(s_0)=1-\exp(\omega)$ for $\omega>\omega^\ast(s_0)$. \item The value function  $V^{\ast}(s_0, \cdot)\mathpunct{:} \Omega(s_0) \to \mathbb{R}$ has $V^{\ast}(s_0, \omega)=v^\ast(s_0)+\delta\bar{V}^{\ast}$ for $\omega\leq \omega^\ast(s_0)$ and $V^{\ast}(s_0, \omega)=(\beta/\delta)\omega + \log(1-\exp(\omega)) + \delta \bar{V}^{\ast}$ for $\omega> \omega^\ast(s_0)$, where the derivative $V_{\omega}^{\ast}(s_0,\omega^\ast(s_0))=\min\{0,(\beta/\delta)-((1-s_0)/s_0)\}$ with $\lim_{\omega\to0}V_{\omega}^{\ast}(s_0,\omega)=-\infty$.
\end{enumerate*}
\end{proposition}

The value function $V^\ast(s_0,\omega)$ is decreasing and concave in $\omega$ (strictly decreasing and strictly concave in~$\omega$ for $\omega>\omega^\ast(s_0)$). The function ``extends to the left'' when the endowment share~$s_0$ is higher.\footnote{That is, for $s>r$ where $\omega_{\min}(s)<\omega_{\min}(r)$, $V^\ast(s,\omega)=V^\ast(r, \omega)$ for $\omega\in\Omega(r)$.} If $\omega^\ast(s_0)>\omega_{\min}(s_0)$ (or equivalently, $s_0>\delta/(\beta+\delta)$), then $V^\ast(s_0, \omega)$ is independent of $\omega$ for $\omega\leq \omega^\ast(s_0)$.  Hence, in the absence of constraint~\eqref{PK0}, the planner would choose $\omega(s_0)=\omega^\ast(s_0)$ because this gives the highest utility to the initial old while maximizing the payoff to the planner. In this case, the allocation given by condition~\eqref{tau_fbm} holds in every period. In contrast, when $\bar{\omega}_0>\omega^{\ast}(s_0)$, the consumption share of the young is lower than implied by condition~\eqref{tau_fbm}, but only in the initial period. There is immediate convergence to the stationary first-best distribution in one period.

Since the payoff to the planner depends on both~$s$ and~$\omega$, the stationary distribution is a pair $(s, \omega^{\ast}(s))$, the endowment share and the corresponding utility promised to the old. We note for future reference that this stationary distribution has $I$ values, one for each endowment state, with the probability of each pair given by $\pi(s)$.

\paragraph{Deterministic Economy}
We now consider a deterministic economy with a constant growth rate $\gamma$ and endowment share $s$. Unlike the previous benchmark, we assume that the planner respects the participation constraint of both the young and the old. Let $\hat{\upsilon}\colonequals\log(s)+\beta\log(1-s)$ be the lifetime endowment utility. Assumption~\ref{ass:sustain}, together with the strict concavity of the utility function, implies that there is a unique ${c}_{\min}<s$, which is the lowest \emph{stationary\/} consumption share of the young that satisfies the participation constraint with equality. The corresponding maximum utility of the old is $\omega_{\max}\colonequals \log(1-c_{\min})$.\footnote{The maximum utility of the old can be found by solving $\log(1-\exp(\omega_{\max})) +\beta \omega_{\max}=\hat{\upsilon}$. Equivalently, the minimum consumption is found by solving $\log({c}_{\min})+\beta \log(1-{c}_{\min})=\hat{\upsilon}$. 
} Analogously to condition~\eqref{tau_fbm}, the first-best consumption share is $c^{\ast}=\delta/(\beta+\delta)$ and the corresponding utility of the old is $\omega^{\ast}\colonequals \log(\beta/(\beta+\delta))$. If $\delta$ is above a critical value, then $c^{\ast}>c_{\min}$ (or equivalently, $\omega^{\ast}<\omega_{\max}$) and the first-best consumption share is sustainable. Otherwise, the first-best consumption share is not sustainable.

Denote the consumption share of the young at time~$t$ by ${c}_{t}$ and the corresponding utility of the old by $\omega_{t}=\log(1-c_t)$. Consider the maximization problem in~\eqref{eqn:obj} with the participation constraints of the young given by $\log({c}_{t})+\beta \log(1-{c}_{t+1})\geq \hat{\upsilon}$. For $\bar{\omega}_0\leq\omega^{\ast}$, constraint~\eqref{PK0} does not bind and it is optimal to set $c_{t}=\max\{c^{\ast},c_{\min}\}$ (or equivalently, $\omega_t=\min\{\omega^{\ast},\omega_{\max}\}$) for all $t\geq0$. On the other hand, consider the case where $\omega^\ast<\omega_{\max}$ and $\bar{\omega}_{0}>\omega^{\ast}$. Then, at $t=0$, $c_{0}$ must satisfy $\log(1-c_{0})\geq\bar{\omega}_0$, which requires that $c_{0}<c^{\ast}$. Clearly, it is desirable to set $c_{0}$ such that $\log(1-c_{0})=\bar{\omega}_0$ and $c_{1}=c^{\ast}$.
However, setting $c_{1}=c^{\ast}$ may violate the participation constraint of the young. In this case, $c_{1}$ has to be chosen to satisfy $\log(c_{0})+\beta \log(1-c_{1})=\hat{\upsilon}$, which implies that $c_{1}<c^{\ast}$. Repeating this argument for $t>1$ shows that given $c_{t}$, the consumption share of the young at time~$t+1$ either satisfies $\log(c_{t})+\beta \log(1-c_{t+1})=\hat{\upsilon}$ or $c_{t+1}=c^{\ast}$ if $\log({c}_{t})+\beta \log(1-{c}^{{\ast}})\geq\hat{\upsilon}$. Intuitively, if the consumption share of the young is low (or equivalently, the utility of the old is large), then the planner would like to raise the consumption share of the young to $c^{\ast}$ (or equivalently, reduce $\omega$ to $\omega^\ast$) as fast as possible to improve welfare. However, if the consumption share of the next-period young is raised too much, then the lifetime utility of the current young falls, and their participation constraint is violated. That is, in the presence of limited enforcement, the consumption share of the young has to be raised gradually. It is useful to express this rule in terms of a policy function $g\mathpunct{:}\Omega\to\Omega$ for the promised utility next period:
\begin{gather}\label{eq:detpol}
g(\omega)\colonequals%
\begin{cases}
\omega^{\ast} & \text{for $\omega\in[{\omega}_{\min},\omega^{c}]$,}\\
\frac{1}{\beta}\left(  \hat{\upsilon}-\log\left(  1-\exp(\omega)\right)  \right)
& \text{for $\omega\in(\omega^{c},{\omega}_{\max}]$,}
\end{cases}
\end{gather}
where $\Omega\colonequals[\omega_{\min}, \omega_{\max}]$, ${\omega}_{\min}=\log(1-s)$ and $\omega^{c}\colonequals \log(1-\exp(\hat{\upsilon}-\beta\omega^{\ast}))$. It follows from the strict concavity of the utility function that $\omega^{c}>\omega^{\ast}$. The function $g(\omega)$ is increasing and convex in $\omega$, as illustrated in Figure~\ref{fig:pol}. The dynamic evolution of $\omega_{t}$ is straightforwardly derived from $g(\omega)$: for $\omega_{t}\in[{\omega}_{\min},\omega^{c}]$, $\omega_{t+1}=\omega^{\ast}$ for all~$t$; for $\omega_{t}\in(\omega^{c},{\omega}_{\max}]$, $\omega_{t+1}$ declines monotonically. Since $\omega^{c}>\omega^{\ast}$, the process for $\omega_{t}$ converges to $\omega^\ast$, attaining its long-run value in finite time.

\begin{figure}[ht]
\begin{center}
\includegraphics{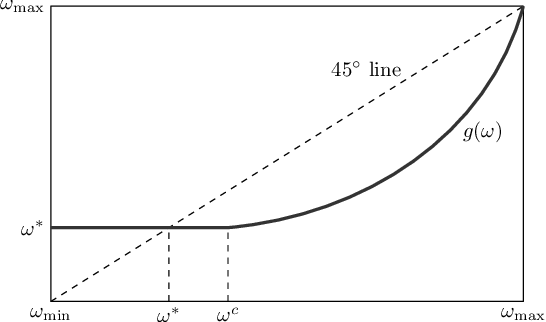}
\end{center}
\caption{Deterministic Policy Function.}\label{fig:pol}%
\begin{figurenotes}The solid line is the deterministic policy function $g\mathpunct{:}\Omega\to\Omega$ that determines the future promised utility as a function of the current promise. The case illustrated has $\omega_{\max}>\omega^{\ast}$. For any initial $\omega\in[\omega_{\min}, \omega_{\max})$, $\omega_t$ converges to $\omega^\ast$ in finite time. 
\end{figurenotes}
\end{figure}

Denote the per-period payoff to the planner with the first-best allocation in the absence of uncertainty by $v^{\ast}\colonequals \log(\delta/(\beta+\delta))+(\beta/\delta)\omega^{\ast}$ and the expected discounted payoff to the planner for $\omega\in\Omega$ by $V(\omega)$. The optimal solution for the deterministic case with sustainable $\omega^{\ast}$ is summarized in the following proposition.

\begin{proposition}
\label{prop:determ} \begin{enumerate*}[label=(\roman*)]
\item If $\omega\in[\omega_{\min},\omega^{\ast}]$, then the consumption share ${c}_{t}=\delta/(\beta+\delta)$ for $t\geq0$.
\item If $\omega\in(\omega^{\ast},{\omega}_{\max}]$, then $\omega_{t+1}$ satisfies equation~\eqref{eq:detpol}. There exists a finite $T$ such that $\omega_t$ is monotonically decreasing for $t<T$ and  $\omega_t=\omega^\ast$ for~$t\geq T$. Likewise, $c_t$ is monotonically increasing for~$t<T$ and ${c}_{t}=c^\ast$ for~$t\geq T$.
\item The value function $V\mathpunct{:}\Omega \to\mathbb{R}$ is equal to $V(\omega)=v^{\ast}/(1-\delta)$ for $\omega\in[\omega_{\min},\omega^{\ast}]$ and is strictly decreasing and strictly concave for $\omega\in(\omega^{\ast},{\omega}_{\max}]$ with $\lim_{\omega\to {\omega}_{\max}}V_{\omega}(\omega)=-\infty$.
\end{enumerate*}
\end{proposition}

The optimal solution is either stationary or converges monotonically to a stationary point within finite time, with $c_{T}=c^\ast$ for $T$ large enough. Hence, the long-run distribution of $\omega$ is degenerate, and for the case where $c^\ast>c_{\min}$, it has a single mass point at $\{\omega^{\ast}\}$.

In the following sections, we show that when the first-best allocation violates a participation constraint of the young, and there is endowment risk, the optimal sustainable intergenerational insurance is history dependent even in the long run, and the ergodic set of utilities has more than $I$ values. The benchmarks highlight that both limited enforcement of transfers and risk are necessary for this result.

\section{Optimal Sustainable Intergenerational Insurance}\label{sec:opt}

In this section, we characterize the optimal intergenerational insurance rule under uncertainty when the planner respects the participation constraints of both the young and the old. Recall that the shocks to growth rates and endowment shares are i.i.d.\ (Assumption~\ref{ass:iid}) and that the optimal sustainable consumption shares depend only on the history of endowment share $s^t$ (Proposition~\ref{prop:iid}). Proposition~\ref{prop:fbstm} describes the solution when the first-best outcome is sustainable. Therefore, in this section, we assume that the first-best allocation violates the participation constraint of the young in at least one state. Since the lifetime endowment utility of an agent is increasing in~$s$, we assume that:
\begin{assumption}\label{ass:nofb}
$ \log(c^\ast(I)) + \beta \sum_{r} \pi(r)\log(1-c^\ast(r))< \log(s(I)) + \beta \sum_{r} \pi(r)\log(1-r)$.
\end{assumption}

We reformulate the optimization problem described in Definition~\ref{def:opt} recursively using the utility $\omega$ promised to the old as a state variable. Let $\omega_r$ denote the state-contingent utility promised to the current young for their old age when the endowment share of the young next period is~$r$. Then, the planner's optimization problem is:
\begin{gather}
V(s,\omega) = \max_{\{c, \left(\omega_{r}\right)_{r\in \mathcal{I}}\}\in\Phi(s,\omega)} \enspace  \frac{\beta}{\delta} \log(1-c) + \log(c) + \delta \sum\nolimits_{r} \pi(r) V(r,\omega_{r}), \tag{P1} \label{eqn:p1}
\end{gather}
where $\Phi(s,\omega)$ is the constraint set given by the following inequalities:
\begin{align}
\log(1-c) &\geq \omega, \label{eqn:pk}\\
c &\leq s, \label{eqn:nn}\\
\omega_r &\leq \omega_{\max}(r) \quad \text{for each $r\in\mathcal{I}$,} \label{eqn:ub} \\
\omega_r &\geq \omega_{\min}(r) \quad \text{for each $r\in\mathcal{I}$,} \label{eqn:lb} \\
\log(c) + \beta \sum\nolimits_{r} \pi(r) \omega_{r} & \geq \log(s)+\beta \sum\nolimits_r\pi(r)\log(1-r).\label{eqn:icr}
\end{align}
The recursive formulation is similar to the promised-utility approach used in models with infinitely-lived agents \citep[see, for example,][]{Green87, Spear-Srivastava87, Thomas-Worrall88, Atkeson-Lucas92}. At each period, the planner chooses the consumption share of the young, $c$, and the state-contingent promise of utility, $\omega_r$. The state variable $\omega$ embodies information about the history of shocks. Constraint~\eqref{eqn:pk} is the promise-keeping constraint, which requires the current old to receive at least what they were promised previously. It is analogous to constraint~\eqref{PK0} that specifies a target utility for the initial old, but it now specifies a target utility in every period. Constraint~\eqref{eqn:nn} is the participation constraint of the old, which stipulates that the old do not transfer to the young. Constraints~\eqref{eqn:ub} and~\eqref{eqn:lb} require that the promise is feasible: $\omega_r\in\Omega(r)\colonequals[\omega_{\min}(r), \omega_{\max}(r)]$. Finally, constraint~\eqref{eqn:icr} requires that the consumption share of the young and the promises made to them for their old age at least match the expected lifetime utility they would receive in autarky.

It is easy to check that the constraint set $\Phi(s,\omega)$ is convex and compact.
Denote the state vector by $x\colonequals (s,\omega)$ and let $f(x)$ and $g_r(x)$ for $r\in\mathcal{I}$ be the optimal consumption share of the young and the state-contingent utility promised to the old next period. The compactness of the constraint set guarantees the existence of the optimal policies, and the strict concavity of the utility function guarantees uniqueness. The optimal allocation is solved recursively. Starting at date $t=0$ with a given state $s_0$ and given $\omega_{0}\in\Omega(s_0)$, solve the optimization problem~\ref{eqn:p1} to obtain the policy functions $f(s_0,\omega_0)$ and $g_r(s_0,\omega_0)$ for $r\in\mathcal{I}$. For the second period, solve the maximization problem again using the endowment share realized at date~$t=1$, say $\hat{r}$, together with the utility promise from the first period, $g_{\hat{r}}(s_0,\omega_0)$, in equation~\eqref{eqn:pk}. The process is then repeated for subsequent periods.

The function $V(s,\omega)$ cannot be found by standard contraction mapping arguments starting from an arbitrary value function because the value function associated with the autarkic allocation also satisfies the functional equation of problem~\ref{eqn:p1}. However, a similar iterative approach can be used to find the value function, starting from the first-best value functions $V^{\ast}(s,\omega)$ derived in Proposition~\ref{prop:fbstm}. Following the arguments of \citet{Thomas-Worrall94}, the limit of this iterative mapping is the optimal value function $V(s,\omega)$. Proposition~\ref{prop:fbstm} established that the first-best value function is nonincreasing, differentiable, and concave in~$\omega$, and the limit value function inherits these properties.
\begin{lemma}
\label{lemma:V_omega}
\begin{enumerate*}[label=(\roman*)]
\item The value function $V(s, \cdot) \mathpunct{:} \Omega(s) \to \mathbb{R}$ is nonincreasing, concave and continuously differentiable in $\omega$, with ${\omega}_{\min}(s)<{\omega}_{\max}(s)$.
\item For each $s\in\mathcal{I}$, there exists an ${\omega}^{0}(s)\in [{\omega}_{\min}(s), \omega^\ast(s)]$ such that $V(s,\omega)$ is strictly decreasing and strictly concave for $\omega>{\omega}^0(s)$. If $\omega^{\ast}(s)>\omega_{\min}(s)$, then $\omega^{0}(s)>\omega_{\min}(s)$ and for at least one such state $\omega^{0}(s)< \omega^{\ast}(s)$. For $\omega\in[\omega_{\min}(s),\omega^0(s)]$, $V_{\omega}(s,\omega)=0$. If $\omega^{\ast}(s)=\omega_{\min}(s)$, then $\omega^{0}(s)=\omega^{\ast}(s)$ and $V_{\omega}(s,{\omega}^{0}(s))\leq({\beta}/{\delta}) -((1-s)/s)\leq 0$. In either case, the limit, $\lim_{\omega\to{\omega}_{\max}(s)}V_{\omega}(s,\omega) = - ({\beta}/{\delta})\lambda_{\max}(s)$, where $\lambda_{\max}(s)\in\mathbb{R}_+ \cup \{\infty\}$.
\item The upper bounds satisfy $\omega_{\max}(s(i))<\omega_{\max}(s(i\shortminus 1))<0$. Similarly, $\omega^{0}(s(i))\leq \omega^{0}(s(i\shortminus 1))$ with strict inequality for at least one $i=2,\ldots,I$.
\end{enumerate*}
\end{lemma}
The strict concavity of the objective function and the convexity of the constraint set guarantee the concavity of $V(s,\omega)$ in $\omega$ with $\omega^{0}(s)=\sup\{\omega\mid V_\omega(s,\omega)=0\}$ if $V_\omega(s,\omega_{\min}(s))=0$ and $\omega^0(s) = \omega_{\min}(s)$ otherwise. Since the old will not transfer to the young voluntarily, ${\omega}_{\min}(s)=\log(1-s)$, the autarkic utility of the old. The upper endpoints $\omega_{\max}(s)$ are determined by the system of equations $\log(1-\exp(\omega_{\max}(s))) + \beta \sum_r\pi(r)\omega_{\max}(r) = \log(s)+\beta \sum_r\pi(r)\log(1-r)$. It can be checked that there is a unique nontrivial solution with $\omega_{\max}(s)$ decreasing with~$s$ and $\omega_{\min}(s)<\omega_{\max}(s)<0$. Analogous to $\omega_{\min}(s)$ and $\omega_{\max}(s)$, $\omega^{0}(s)$ is also decreasing in~$s$. Differentiability of $V(s, \omega)$ with respect to~$\omega$ follows because the constraint set satisfies a linear independence constraint qualification when $\omega\in[{\omega}_{\min}(s), {\omega}_{\max}(s))$. The left-hand derivative of $V(s,\omega)$ with respect to~$\omega$, evaluated at ${\omega}_{\max}(s)$, is finite if ${\omega}_{\max}(s)$ is part of the ergodic set and is infinite otherwise. 

\begin{remark}\label{rem:initialold}
Recall that $\bar{\omega}_0$ is the exogenous target utility given in constraint~\eqref{PK0}. Given the definition of $\omega^{0}(s)$, the planner chooses the initial utility of the old such that $\omega_0 = \max\{\omega^0(s_0), \bar{\omega}_0\}$. If the planner is not subject to constraint~\eqref{PK0} and can freely choose the initial utility, then the planner sets $\omega_0=\omega^{0}(s_0)$. Note that unlike $\bar{\omega}_0$, $\omega^0(s)$ is endogenous and depends on all of the model's primitives.
\end{remark}

\begin{remark}
The optimal sustainable intergenerational insurance is not renegotiation-proof because, in the case of default, it would be possible to promise a utility of $\omega^0(r)$, instead of $\omega_{\min}(r)$, without diminishing the planner's payoff. A renegotiation-proof outcome can be derived by replacing constraint~\eqref{eqn:icr} with $\log(c)  +\beta\sum_r\pi(r)\omega_{r} \geq \log(s)  +\beta\sum_r\pi(r)\omega^{0}(r)$. Since $\omega^0(r)$ is endogenous and appears in the constraint, a fixed-point argument similar to that used by \citet{Thomas-Worrall94} is required to find the solution. Although imposing this tighter constraint restricts risk sharing, the structure of the optimization problem is not affected. Therefore, we expect that the qualitative properties of the optimal solution are substantially unchanged.
\end{remark}

\paragraph{Optimal Policy Functions} We now turn to the properties of the policy functions $f(x)$ and $g_r(x)$. Given the differentiability of the value function, the first-order conditions for the programming problem~\ref{eqn:p1} are:
\begin{align}
f(x)  &= \min\left\{\frac{\delta(1+\mu(x))}{\beta(1+\lambda(x))+\delta(1+\mu(x))},s\right\}, \label{rmu1} \\
V_\omega(r, g_r(x)) & = - \frac{\beta}{\delta}\left(\mu(x)-\xi_r(x)+\eta_r(x)\right) \quad \text{for each $r\in\mathcal{I}$,} \label{fp1}
\end{align}
where $(\beta/\delta)\lambda(x)$ is the Lagrangian multiplier corresponding to the promise-keeping constraint~\eqref{eqn:pk}, 
$\beta\pi(r)\xi_r(x)$ are the multipliers corresponding to the upper bound on the promised utility~\eqref{eqn:ub}, $\beta\pi(r)\eta_r(x)$ are the multipliers corresponding to the lower bound on the promised utility~\eqref{eqn:lb}, and $\mu(x)$ is the multiplier corresponding to the participation constraints of the young~\eqref{eqn:icr}. Given the concavity of the programming problem, conditions~\eqref{rmu1} and~\eqref{fp1} are both necessary and sufficient. There is also an envelope condition:
\begin{gather}
V_\omega(x) = - \frac{\beta}{\delta}\lambda(x). \label{env1}
\end{gather}
Taken together, equations~\eqref{fp1} and~\eqref{env1} imply the following updating property:
\begin{gather}\label{eqn:update}
\lambda(x^\prime)=\mu(x)-\xi_r(x) +\eta_r(x),
\end{gather}
where $x^\prime = (r, g_r(x))$ is the next-period state variable.
To interpret equation~\eqref{eqn:update} suppose, for simplicity, that the boundary constraints on the promised utility do not bind. In this case, $\eta_r(x)=\xi_r(x)=0$ and the updating property reduces to $\lambda(x^\prime)=\mu(x)$. From equation~\eqref{fp1}, it follows that $\delta(1+\mu(x))$ is the relative weight placed on the utility of the young and $\beta(1+\lambda(x))$ is the relative weight placed on the utility of the old. Therefore, in this case, the updating property shows that the relative weight placed on the utility of the old corresponds to the tightness of the participation constraint they faced when they were young.

The following two Lemmas describe the properties of the policy functions.\footnote{To avoid the clumsy terminology of nondecreasing or weakly increasing, we describe a function as increasing if it is weakly increasing and highlight cases where a function is strictly increasing.}

\begin{lemma}\label{lem:g}
\begin{enumerate*}[label=(\roman*)]
\item The policy function $g_r(s,\cdot)\mathpunct{:}\Omega(s)\to[\omega^{0}(r), \omega_{\max}(r)]$ is continuous and increasing in~$\omega$ and strictly increasing for $g_r(s,\omega)\in(\omega^{0}(r), \omega_{\max}(r))$.
\item For each $r\in\mathcal{I}$ and $\omega\in(\omega_{\min}(s(i\shortminus 1)), \omega_{\max}(s(i)))$, $g_r(s(i), \omega)\geq g_r(s(i\shortminus 1), \omega)$ with strict inequality for at least one $i=2,\ldots,I$. For each $s\in\mathcal{I}$, $g_{r(i)}(s, \omega) \leq g_{r(i\shortminus 1)}(s,\omega)$ with strict inequality for at least one $i=2,\ldots,I$.
\item For endowment state~1, there is a critical value $\omega^c>\omega^{0}(1)$ such that $g_r(1,\omega)=\omega^{0}(r)$ for $\omega\in[\omega^{0}(1), \omega^c]$ and~$r\in\mathcal{I}$.
\item For each~$s\in\mathcal{I}$, there is a unique fixed point $\omega^f(s)=\min\{\omega^{\ast}(s), \omega_{\max}(s)\}$ of the mapping $g_s(s,\omega)$ with $g_s(s,\omega)>\omega$ for $\omega<\omega^f(s)$ and $g_s(s,\omega)<\omega$ for $\omega>\omega^f(s)$. For endowment state~$I$, $\omega^f(I)>\omega^{0}(I)$.
\end{enumerate*}
\end{lemma}

\begin{lemma}\label{lem:c}
\begin{enumerate*}[label=(\roman*)]
\item The policy function $f(s, \cdot)\mathpunct{:}\Omega(s) \to (0,s]$ where $f(s, \omega) =1-\exp(w)$ for $\omega\geq\omega^0(s)$ and $f(s, \omega) =1-\exp(\omega^0(s))$ for $\omega<\omega^0(s)$.
\item $c^0(s) \colonequals f(s, \omega^{0}(s))$ where $c^{0}(s(i))\geq c^{0}(s(i\shortminus 1))$ with strict inequality for at least one~$i=2,\ldots,I$.
\item At the fixed point $\omega^f(s)$, $f(s,\omega^f(s)) \leq c^{\ast}(s)$ with equality for $\omega^f(s)< \omega_{\max}(s)$.
\end{enumerate*}
\end{lemma}

The main properties of Lemmas~\ref{lem:g} and~\ref{lem:c} follow straightforwardly from the objective to share risk subject to the participation constraints. The policy function $g_r(s, \omega)$ is increasing in~$\omega$ (Lemma~\ref{lem:g}(i)), whereas $f(s,\omega)$ is decreasing in~$\omega$ (Lemma~\ref{lem:c}(i)). A higher promise to the current old means a lower consumption share for the current young and, for endowment states in which the participation constraint binds, this requires a higher future promise of utility for their old age as compensation. The consumption share of the young does not directly depend on~$s$ and depends only indirectly on~$s$ when $\omega=\omega^{0}(s)$ or $\omega=\omega_{\max}(s)$ (Lemma~\ref{lem:c}(ii)), whereas $g_r(\omega, s)$ is increasing in~$s$ and decreasing in~$r$ (Lemma~\ref{lem:g}(ii)). The policy function $g_r(\omega, s)$ is increasing in~$s$ because a higher endowment share of the young today is associated with a larger risk-sharing transfer, which, if the participation constraint is binding, has to be compensated by a higher promise for tomorrow. Likewise, the future promise is decreasing in~$r$ because a higher endowment share of the young tomorrow is associated with a higher consumption share when the participation constraint binds and, hence, a lower consumption share of the old tomorrow. Since the optimum is nontrivial and differs from the first best, there is at least one strict inequality in the relations of Lemma~\ref{lem:g}(ii), so that, $g_r(s(I), \omega)> g_r(s(1), \omega)$ and $g_{r(I)}(s,\omega)<g_{r(1)}(s, \omega)$.

Lemma~\ref{lem:g}(iii) shows that there is a range of $\omega$ above $\omega^{0}(1)$ such that the participation constraint of the young does not bind and hence, $g_r(1, \omega)=\omega^{0}(r)$ in this range. This is analogous to the deterministic case discussed in Section~\ref{sec:bench} where the policy function has an initial flat section (see, Figure~\ref{fig:pol}). More generally, when the participation constraint of the young does not bind, it follows from equation~\eqref{env1} that $g_r(x)=\omega^{0}(r)$ and $x^\prime=(r, \omega^{0}(r))$. In this case, we say that the promise is \emph{reset\/}. The promise is reset to the value that gives the most to the current old while maximizing the payoff to the planner. Lemma~\ref{lem:g}(iii) shows that resetting, in particular, occurs in state~1 for any $\omega\in[\omega^{0}(1), \omega^c]$.

Lemmas~\ref{lem:g}(iv) and~\ref{lem:c}({iii) describe what happens when the same endowment share repeats in successive periods. Suppose for simplicity that  $\eta_s(x)=\xi_s(x)=0$ and $f(x)>s$. From equations~\eqref{fp1} and~\eqref{env1}, $\mu(s, \omega^f(s))=\lambda(s, \omega^f(s))$ where $\omega^f(s)$ is the fixed point of $g_s(s,\omega)$. Using equation~\eqref{rmu1}, this implies that the consumption share is first best and hence, $\omega^f(s)=\omega^{\ast}(s)$. Furthermore, $g_s(s,\omega)>\omega$ for $\omega<\omega^f(s)$ and $g_s(s,\omega)<\omega$ for $\omega>\omega^f(s)$. That is, when the same endowment share repeats, the promise falls if the previous promise was above the first best and rises if the previous promise was below the first best. It follows that the policy function $g_s(s,\omega)>\omega$ cuts the 45$^\circ$ line once from above. To understand this, consider some $\omega>\omega^f(s)$ and suppose, to the contrary, that $g_s(s,\omega)\geq \omega$. In this case, equations~\eqref{fp1} and~\eqref{env1} imply that $\mu(s, \omega^f(s))>\lambda(s, \omega^f(s))$, which in turn implies $\omega<\omega^{\ast}(s)=\omega^f(s)$ from equation~\eqref{rmu1}, a contradiction. A similar argument shows that $g_s(s, \omega)>\omega$ for $\omega<\omega^f(s)$.\footnote{%
The argument can be extended to the case where the nonnegativity and upper bound constraints bind, and a complete proof of Lemma~\ref{lem:g} is provided in the Appendix.}

The implications of Lemmas~\ref{lem:g} and~\ref{lem:c} can be illustrated by considering a particular \emph{sample path\/} of the consumption share generated for a given history of endowment shares $s^T=(s_0, s_1,\ldots, s_T)$. The sample path of the consumption share is constructed iteratively from the policy functions $f(s,\omega)$ and $g_r(s,\omega)$ starting with $x_0=(s_0,\omega_0)$ as follows: $c_t=f^t(s^t,x_0)\colonequals f(s_t,g^t(s^t,x_0))$, where $g^t(s^t,x_0)  \colonequals g_{s_{t}}(s_{t-1},g^{t-1}(s^{t-1},x_0))$ and $g^0(s_0,x_0)=\omega_0$.

\begin{figure}[ht]
\begin{center}
\includegraphics{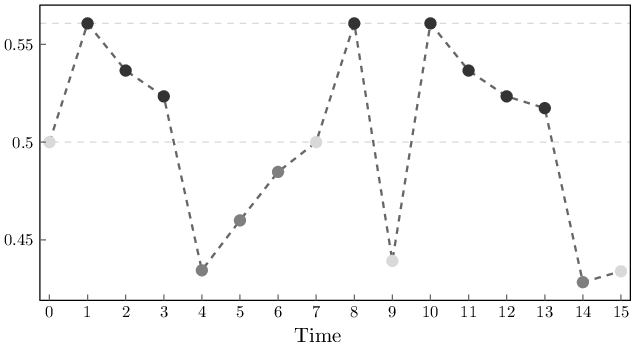}
\end{center}
\caption{Sample Path of the Young Consumption Share.} \label{fig:sp2}
\begin{figurenotes}
The illustration is for a case where $I=3$ and $\beta=\delta$ (where the first-best consumption share is $\sfrac{1}{2}$). The shade of the dots indicates the state~$s_t$: light gray for~$s_t=s(1)$, mid gray for~$s_t=s(2)$ and dark gray for~$s_t=s(3)$. The case illustrated has $s_0=s(1)$ and $\omega_0=\omega^{0}(1)=-\log(2)$.
\end{figurenotes}
\end{figure}

Figure~\ref{fig:sp2} depicts such a sample path in a three-state example and illustrates three important properties.\footnote{The example has $\beta=\delta=\exp(-\sfrac{1}{75})$ (corresponding to an interest rate of $\sfrac{1}{75}$), $s(1)=0.5$, $s(2)=0.625$ and $s(3)=0.8125$, with probabilities $\pi(1)=0.5$, $\pi(2)=0.25$ and $\pi(3)=0.25$.} First, the optimal sustainable consumption share fluctuates above and below the first-best level of $c^{\ast}(s)=0.5$.\footnote{By Lemma~\ref{lemma:V_omega}(ii), $\omega^{0}(s)\leq \omega^{\ast}(s)$. By Assumption~\ref{ass:LS}, $\omega^{\ast}(1)=\omega_{\min}(1)$. Hence, $\omega^{0}(1)=\omega^{\ast}(1)$. Since $g_1(s,\omega)$ is increasing in~$\omega$, the promise is above the first-best level (or equivalently, the consumption share is below the first-best level) in state~1. From Lemma~\ref{lem:g}(iii), $\omega^{0}(I)<\omega^{\ast}(I)$ and therefore, for low values of~$\omega$, the promise is below the first-best level (or equivalently, the consumption share is above the first-best level) in state~$I$.} Second, the path is history dependent. That is, the consumption share varies both with the current endowment state and the history of shocks. For example, the endowment share $s_t=s(3)$ occurs at $t=8$ and $t=13$, but the consumption share differs at the two dates. When state~3 occurs, the participation constraint of the young binds, and hence, a higher future utility must be promised to ensure that they are willing to share more of their current endowment. Subsequent realizations of state~3 exacerbate the situation because the young of the next generation must also deliver on past promises, meaning that the consumption share of the young falls when state 3 repeats. This property is evident in Figure~\ref{fig:sp2} where $c_t$ falls when state~3 repeats ($t=2,3$ and $t=11,12,13$). This implies that the consumption share is not necessarily monotonic in the endowment. For example, the consumption share at $t=4$, when the endowment share is $s_4=s(2)$, is lower than the consumption share at $t=9$, when the endowment share is $s_9=s(1)<s(2)$. This nonmonotonicity occurs because the promise made to the old for $t=4$ is higher than that made for $t=9$. Third, there are points in time when the consumption share returns to the same value in the same state. For example, this happens at $t=7$, which has the same state (state~1) and same consumption share as at $t=0$. In this case, there is \textit{resetting\/}. The path of the consumption share is the same following resetting if the same sequence of endowment shares occurs. Note that the definition of the resetting points is not unique. For example, there is resetting also at $t=1, 8, 10$, when state~3 occurs after state~1. Before resetting occurs, the effect of a shock persists. However, once resetting occurs, the history of shocks is forgotten, and the subsequent sample path is identical when the same sequence of states occurs. That is, the sample paths between resettings are probabilistically identical. This property is used in the next section to establish convergence to a unique invariant distribution. 

\section{Convergence to the Invariant Distribution}\label{sec:conv}

This section considers the long-run distribution of the pair $x=(s, \omega)$. It shows that there is a unique and countable ergodic set $E$ with cardinality $\abs{E}>I$ and strong convergence to the corresponding invariant distribution. Let $\Omega=\cup_{r\in \mathcal{I}}\Omega(r)$ and $\mathcal{X}=\mathcal{I}\times \Omega$. The future evolution of~$x$ is a Markov chain defined by the transition function $P(x,A\times B) \colonequals \Pr \{x_{t+1}\in A \times B \mid x_t = x\} = \sum\nolimits_{r\in A} \pi(r) \mathbb{1}_B g_r(x)$ where $A\subseteq \mathcal{I}$, $B\subseteq\Omega$ and $\mathbb{1}_B g_r(x)=1$ if $g_r(x) \in B$ and zero otherwise. The chain starts from $x_0=(s_0, \omega_0)$ with an initial promise $\omega_0=\max\{\omega^0(s_0), \bar{\omega}_0\}$. 

The monotonicity and resetting properties of Lemma~\ref{lem:g} imply that starting from any $x_t$, a sequence of $k$ consecutive state~1s (where the endowment share is~$s(1)$) leads to $x_{t+k} = (1, \omega^0(1))$ for a finite~$k$. This is because $g_{1}(1, \omega) < \omega$, so that repetition of state~1 leads to a decrease in $\omega$ and since $g_{1}(1, \omega)=\omega^{0}(1)$ for some $\omega>\omega^{0}(1)$, $\omega$ falls to $\omega^{0}(1)$ in finite time. In this case, we say that $x$ is \textit{reset\/} to $(1, \omega^{0}(1))$ at time~$t+k$. Since the probability of state~1 is $\pi(1)>0$, the probability of a history of $k$ consecutive state~1s is $\pi(1)^k>0$. An immediate consequence is that Condition~\textbf{M} of \citet[][page~348]{Stokey-Lucas89} is satisfied, and hence, there is strong convergence in the uniform metric to a unique invariant probability measure $\phi(X)$ for $X\in\mathcal{X}$.\footnote{%
Condition~{\bf M} is satisfied because there is a $k\geq1$ and an $\epsilon>0$ such that the $k$-step transition function $P^k(x, \{(1,{\omega}^{0}(1))\})>\epsilon$ for any $x\in\mathcal{X}$. In this case, $(1,\omega^0(1))$ is an atom of the Markov chain. \citet{Acikgoz18}, \citet{Fossetal18}, and \citet{Zhu17} use similar arguments to establish strong convergence in the Aiyagari precautionary savings model with heterogeneous agents.}

Since there is a positive probability that $x$ is reset to $(1,{\omega}^{0}(1))$ in finite time, the Markov chain for~$x$ is \textit{regenerative\/} and $(1,{\omega}^{0}(1))$ is a regeneration point \citep[see, for example,][]{Fossetal18}. For simplicity, suppose first that the process starts at $x_0=(1, \omega^{0}(1))$. Recall that $g^t(s^t,x_0)= g_{s_{t}}(s_{t-1},g^{t-1}(s^{t-1},x_0))$ where $g^0(1,x_0)=\omega^{0}(1)$. Let $r_x\colonequals \min \{k\geq 1 \mid (s, g^k((s^{k-1},s), x_0))=x\}$ denote the first time that the process is equal to~$x$ starting from $x_0$. Then, $r_{x_{0}}$ is the first regeneration time, the first time after the initial period at which $x_0$ reoccurs. Any sample path of promises can be divided into different blocks, with each block starting at a regeneration time. This can be seen in Figure~\ref{fig:sp2} where the first regeneration time occurs at~$t=7$. Although the blocks between regeneration points are not identical, the strong Markov property ensures that they are i.i.d. At each regeneration time, past shocks are forgotten, and the future evolution of~$x$ is probabilistically identical. The regeneration times are also i.i.d.\ and the expected regeneration time is $\varphi\colonequals \mathbb{E}_0[r_{x_{0}}]$, the same for any block. Moreover, $\varphi$ is finite since all positive probability paths must have a sequence of endowment states leading to~$x_0=(1, \omega^{0}(1))$ as described above.

Now consider a starting point $x_0=(i, \omega^0(i))$ for some initial state~$s_0=s(i)$. Given that $g_i(1, \omega^{0}(1))=\omega^{0}(i)$ by Lemma~\ref{lem:g}(iii), a positive probability path that leads back to $x_0$ is constructed by a sequence of consecutive state~1s, as outlined above, followed by state~$i$. Since the transition from state~1 to state~$i$ occurs with positive probability, $(i, \omega^{0}(i))$ is a regeneration point, and the blocks between these regeneration points are also probabilistically identical. As discussed in Remark~\ref{rem:initialold}, in the absence of constraint~\eqref{PK0}, the planner sets $\omega_0=\omega^{0}(i)$ and the process starts in the ergodic set. However, if constraint~\eqref{PK0} must be respected and $\bar{\omega}_0>\omega^{0}(i)$, then $x_0=(i, \bar{\omega}_0)$ and the process may start outside of the ergodic set. In this case, there is still a positive probability path back to a resetting point $(i, \omega^{0}(i))$. The only difference is that the first block in the regenerative process is different from subsequent blocks (which all start from $(i, \omega^{0}(i))$). However, this does not change the convergence properties of the process. 

Let $R_{x}\colonequals \Pr(r_x <\infty)$ be the probability of attaining the pair $x=(s,\omega)$ in finite time starting from $x_0$. If $R_x>0$, then $x$ is said to be \textit{accessible\/} from $x_0$. Since $x_0= (i, \omega^{0}(i))$ has a positive probability mass and the set of endowment states $\mathcal{I}$ is finite and time is discrete, the associated set $E\colonequals \{x \mid R_{x} >0\}$ is countable. Moreover, the set $E$ is an equivalence class because every $x\in E$ is accessible from $x_0$, and there is a path from every accessible $x$ back to $x_0$. Therefore, $E$ is an absorbing set, that is, $P(x, E)=1$ for all $x \in E$, and since no proper subset of $E$ has this property, it is \textit{ergodic\/} \citep[see, for example,][chapter 11]{Stokey-Lucas89}. Let $\varphi_{x}$ denote the \textit{expected return time\/} to~$x$ where $\varphi_{x_{0}}\equiv\varphi$. With $\varphi$ finite, it follows that $R_{x}=1$ and each $\varphi_{x}$ is finite, that is, each $x\in E$ is positive recurrent.

Since the ergodic set~$E$ is countable, standard results on the convergence of positive recurrent Markov chains apply. To state these results, let $P$ denote the transition matrix with elements $P(x, x^\prime)=\pi(r)\mathbb{1}_{\omega_{r}} g_r(x)$ where $x=(s,\omega)$ and $x^\prime=(r,g_r(x))$. Similarly, let $P^{k}(x, x^\prime)$ be the elements of the corresponding $k$-period transition matrix.

\begin{proposition}
\label{prop:convergence}
\begin{enumerate*}[label=(\roman*)]
\item There is pointwise convergence to a unique and nondegenerate invariant distribution $\phi=\phi P$ where for each $x\in E$, $\phi(x)=\lim\nolimits_{k\to\infty} P^{k}(\cdot, x)= \varphi_{x}^{-1}$.
\item The invariant distribution is the limit of the iteration $\phi_{t+1}(x^{\prime}) = \sum\nolimits_{x \in E}  P(x,x^\prime) \phi_t(x)$ for any given $\phi_0(x)$. \item The cardinality $\abs{E}>I$.
\end{enumerate*}
\end{proposition}

Parts (i) and (ii) of Proposition~\ref{prop:convergence} are standard and show convergence to a unique invariant distribution where the probability of each $x\in E$ is the inverse of the expected return time. The invariant distribution can be computed iteratively, given knowledge of the policy functions. In particular, for $s_0=s(i)$, the invariant distribution can be computed starting from an initial distribution $\phi_0(x)=1$ if $x=(i, \omega^{0}(i))$ and $\phi_0(x)=0$ otherwise.\footnote{%
The convergence results hold for any initial distribution $\phi_0(A)$ even if $A\not\subseteq E$ since eventually, once regeneration occurs, all subsequent promises belong to the ergodic set.} Part~(iii) shows that the cardinality of the ergodic set is greater than~$I$. That is, at the invariant distribution, there are multiple promised utilities associated with particular states. Hence, the history of endowment share affects the consumption allocation even in the long run. This result stands in contrast to the two benchmarks considered in Section~\ref{sec:bench}. If transfers are enforced, or if there is no risk, then convergence is to an ergodic set with a cardinality equal to the cardinality of the set of endowment states. 

Since Lemma~\ref{lem:g} shows that $g_r(s,\omega)$ is increasing in~$s$ and~$\omega$, $g_r(I, \omega^{\ast}(I))$ is the largest promise that can be reached in state~$r$ starting with $x_0=(i, \omega^{0}(i))$. If $g_r(I, \omega^{\ast}(I))<\omega_{\max}(r)$, then any $x=(r,\omega)$ with $\omega \in (g_r(I, \omega^{\ast}(I)), \omega_{\max}(r))$ is not accessible from~$x_0$. Therefore, such an~$x$ is transitory and is not part of the ergodic set. In Section~\ref{sec:example}, we compute the ergodic set and the invariant distribution in examples with $g_r(I, \omega^{\ast}(I))<\omega_{\max}(r)$.\footnote{%
The ergodic set and invariant distribution are difficult to characterize. In some cases, however, the invariant distribution is a transformation of a geometric distribution with a denumerable ergodic set, that is, $\abs{E}=\aleph_{0}$.}

\begin{remark}
The convergence result and all the results of Section~\ref{sec:opt} apply when preferences exhibit constant relative risk aversion. They also hold for any concave utility function if the aggregate endowment is constant. If the aggregate endowment is state-dependent, but there is no growth, then Lemmas~\ref{lem:g} and~\ref{lem:c}, and Proposition~\ref{prop:convergence} remain valid, except that the policy functions are not necessarily monotonic in the endowment state \citep[see,][for details]{LRW20}.  
\end{remark} 

\section{Debt}\label{sec:debt}

In this section, we reinterpret the optimal transfer to the old as debt. Suppose the planner issues one-period state-contingent bonds, which trade at the corresponding state prices. The planner uses the revenue generated by bond sales to fund the transfer to the old, balancing the budget by taxing or subsidizing the young. Given bond prices and taxes, the young buy the correct quantity of state-contingent bonds to finance their optimal old-age consumption. With this interpretation, the dynamics of debt and the fiscal reaction function can be examined. 

\paragraph{The Debt Policy Function} It is convenient to measure debt relative to the endowment share of the current young. Then, the optimal debt $d(x)$ satisfies $\omega=\log(1-s+sd(x))$, so $d(x)$ is increasing in~$\omega$.\footnote{For brevity, in what follows, we often refer to $d(x)$ simply as outstanding debt without the caveat that it is measured relative to the endowment share of the young.} Let $d^{0}(s) \colonequals d(s, \omega^{0}(s))\geq0$ denote the minimum debt at the optimal solution when the endowment share of the young is~$s$. Debt $d\in\mathcal{D}=[d_{\min}, d_{\max}]$ where the minimum debt $d_{\min}\colonequals\min_r d^{0}(r)$ and the maximum debt $d_{\max}$ is determined as the nontrivial solution of $\log(1-d_{\max}) +\beta \sum_{r}\pi(r)(\log(1-r+r d_{\max})-\log(1-r))=0$. We refer to $d_{\max}$ as the \textit{debt limit\/} and $d_{\max}-d$ as the \textit{fiscal space\/} \citep[see, for example,][]{Ghoshetal13}.\footnote{The debt limit is different from the maximum sustainable debt \citep[see, for example,][]{Collardetal15}. The maximal sustainable debt focuses on the limit that external investors are willing to lend to a government, taking into account the probability of default. Typically, it is calculated using a fixed rule for government taxes and expenditure and a constant interest rate.} It follows straightforwardly that $d_{\max}<1$, analogously to the result of Lemma~\ref{lemma:V_omega} that $\omega_{\max}(s)<0$.} The \textit{debt policy function\/} ${b_{r}}\mathpunct{:}\mathcal{{D}}\to \mathcal{{D}}$ determines the optimal debt next period when the current debt is $d$ and the endowment share of the young next period is~$r$. The properties of the debt policy functions are summarized in the following corollary to Lemmas~\ref{lem:g} and~\ref{lem:c}.

\begin{corollary}\label{cor:debtf}
\begin{enumerate*}[label=(\roman*)]
\item The debt policy function $b_{r}\mathpunct{:}\mathcal{{D}}\to \mathcal{{D}}$ is continuous in~$d$. For $d\leq d^c$,  $b_r(d) = d^{0}(r)$, and for $d> d^c$, $b_r(d)$ is strictly increasing in~$d$. The threshold $d^c$ satisfies
$
d^c = 1-\exp(-\beta\sum_{r}\pi(r)(\log(1-r+r d^{0}(r)) - \log(1-r))) \in(d_{\min}, d_{\max})
$ with $d_{\min}=0$ and $d_{\max}<1$.
\item For $d\in\mathcal{D}$, $b_{r(i)}(d) \geq b_{r(i\shortminus 1)}(d)$ with strict inequality for at least one~$i=2,\ldots,I$.
\item For each~$r\in\mathcal{I}$, there is a unique fixed point $d^f(r)$ of the mapping $b_r(d)$ where $d^f(r)=\min\{d^{\ast}(r),d_{\max}\}$ and $d^\ast(r)=1-c^{\ast}(r)/r$ is the first-best debt. For state~$I$, $d^f(r(I))>d^c$.
\end{enumerate*}
\end{corollary}

Corollary~\ref{cor:debtf} reveals the benefits of measuring debt relative to the endowment share of the young. First, the debt policy functions depend on the current debt $d$ but are independent of the current endowment share~$s$. Second, there is a common threshold $d^c$, below which the debt policy function is flat and above which it is strictly increasing. For $d\leq d^c$, the debt policy function $b_r(d)=d^{0}(r)$.  Lemmas~\ref{lem:g} and~\ref{lem:c} show why the debt policy function is independent of~$s$. When the participation constraint of the young binds, that is, when constraints~\eqref{eqn:pk} and~\eqref{eqn:icr} hold as equalities, the policy function for the promised utility $g_r(s,\omega)$ is an increasing function of $\log(s)-\log(1-\exp(\omega))$. With $\exp(\omega)=1-s+sd$, $\log(s)-\log(1-\exp(\omega))=-\log(1-d)$ and $g_r(s,\omega)$ is an increasing function of~$d$. Hence, the debt policy function depends on the current debt and endowment state next period.\footnote{For CRRA preferences with a coefficient of risk aversion greater than one, the same property applies with a different normalization of debt that depends on the coefficient of risk aversion.} 

Part~(i) of Corollary~\ref{cor:debtf} shows that the threshold $d^c$ is determined by setting $b_r(d)=d^0(r)$ for each~$r$. By Assumption~\ref{ass:LS}, $d_{\min}=0$ and by Assumption~\ref{ass:nofb}, $d^c<d^\ast(I)$. Part~(ii) shows that $b_r(d)$, and consequently, $rb_r(d)$, are increasing in~$r$. Since the consumption share of the old decreases with~$r$, the transfer to the old, $rb_r(d)$, is positively correlated with the marginal utility of consumption of the old. This positive correlation occurs because debt provides partial insurance. Note that the consumption share of the old decreases with~$r$ for a given debt~$d$, while it increases with~$d$ for a fixed~$r$. Therefore, in comparing two endowment states, the consumption share of the old may be higher when the young have a higher endowment share if the debt is sufficiently high. Part~(iii) follows directly from Lemma~\ref{lem:g}(iv) and the fixed point of the mapping $b_r(d)$ corresponds to the first-best debt.

\begin{figure}[ht]
\begin{center}
\includegraphics{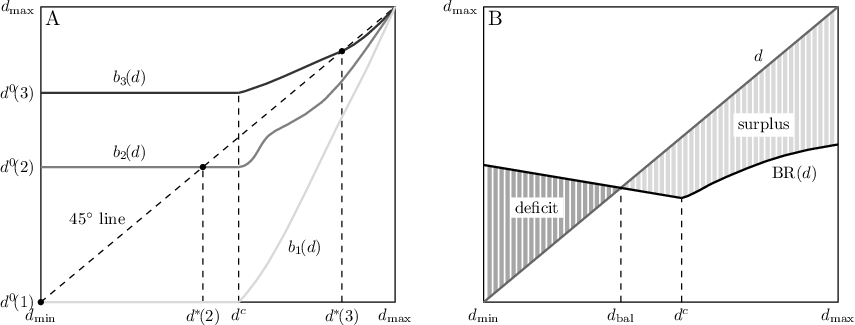}
\end{center}
\caption{Debt Dynamics (Panel A) \& Bond Revenue Function (Panel B).} \label{fig:Debt}
\begin{figurenotes}
The illustration is for the case $I=3$ corresponding to the example in Figure~\ref{fig:sp2}. Panel~A plots the debt policy functions $b_r\mathpunct{:}\mathcal{D}\to\mathcal{D}$ for $r=1,2,3$. The light-gray line is $b_{1}(d)$, the mid-gray line is $b_{2}(d)$, and the dark-gray line is $b_{3}(d)$. The level $d^{\ast}(3)$ is the largest sustainable debt, and $d_{\min}=d^{\ast}(1)=0$ is the lowest sustainable debt within the ergodic set. Panel~B plots the bond revenue function $\mbox{BR}\mathpunct{:}\mathcal{D}\to\mathcal{D}$. The  fiscal reaction function is the difference $d-\mbox{BR}(d)$. For $d<d_{\bal}$, the primary fiscal balance is in deficit, and for $d>d_{\bal}$, it is in surplus.
\end{figurenotes}
\end{figure}

\paragraph{The Dynamics of Debt} The dynamics of debt are derived from the debt policy functions described in Corollary~\ref{cor:debtf} and the history of endowment shares. Panel A of Figure~\ref{fig:Debt} plots the debt policy functions corresponding to the three-state example illustrated in Figure~\ref{fig:sp2}. For $d\leq d^c$, the debt policy function is independent of the current debt and depends only on the endowment share of the young next period. In particular, $d^0(1)=d^\ast(1)$ and $d^0(2)=d^\ast(2)$, so that the consumption share is first best in states~1 and~2, whereas in state~3, $d^0(3)<d^\ast(3)$ because the corresponding participation constraint binds, limiting the transfer from the young. For $d>d^c$, debt falls when the endowment share of the young next period is $r(1)$ or $r(2)$. If, for example, there are enough consecutive occurrences of the endowment state~1, then debt falls to zero. Since such sequences occur with positive probability, debt is reset to zero periodically. If, on the other hand, the endowment share of the young next period is $r(3)$, then the debt rises for $d<d^{\ast}(3)$ but falls for $d>d^{\ast}(3)$. Thus, any debt $d>d^{\ast}(3)$ is transitory and cannot occur in the long run.\footnote{In general, if $d^{\ast}(I)<d_{\max}$, then any $d\in[d^{\ast}(I), d_{\max})$ is transitory.}  In summary, the current debt encapsulates the history of endowment shares, and debt will rise or fall depending on the endowment share of the young next period. 

\paragraph{Fiscal Reaction Function} The fiscal reaction function shows how the tax rate depends on debt. Since the promised utility and debt are monotonically related, we abuse notation and rewrite the state space as $x=(s,d)$. With logarithmic preferences, the intertemporal marginal rate of substitution is $m(x, x^\prime)=\beta s(1-d)/(1-r(1-b_r(d)))$, where $x=(s,d)$ is the current state and $x^\prime=(r, b_r(d))$ is the next-period state. Since the endowment shares are i.i.d., the transition probability $\pi(x,x^\prime)=\pi(r)$ and given debt $d$, the current young can be thought as buying $rb_r(d)$ bonds contingent on a next-period endowment share of $r$ at the state price of $q(x, x^\prime)=\pi(r)m(x,x^\prime)$. This generates a bond revenue for the planner, measured relative to the endowment share of the young, of:
\begin{gather*}\label{eqn:br}
\mbox{BR}(d) \colonequals \left(\frac1s\right)\sum\nolimits_{r}q(x,x^\prime)r b_r(d)=
\beta\sum\nolimits_{r\in\mathcal{I}}\pi(r)\left(\frac{1-d}{1-r(1-b_r(d))}\right)r b_r(d).
\end{gather*}
Note that $\mbox{BR}(d)$ is independent of~$s$. The planner finances the current debt~$d$ by a combination of taxes (or subsidies) on the young and bond revenue $\mbox{BR}(d)$. Hence, the budget constraint of the planner is:
\begin{gather}\label{eqn:budget}
  \tau(d)=d-\mbox{BR}(d),
\end{gather}
  where $\tau(d)$ is the tax rate on the young, measured as a share of their endowment. We refer to $\tau(d)$ as the \textit{fiscal reaction function\/} and $s\tau(d)$ as the \textit{primary fiscal balance\/}.
  A positive value of $s\tau(d)$ corresponds to a primary fiscal surplus, whereas a negative value of $s\tau(d)$ corresponds to a primary fiscal deficit.

Panel~B of Figure~\ref{fig:Debt} plots the outstanding debt $d$ and the bond revenue $\mbox{BR}(d)$ with the fiscal reaction function $\tau(d)$ given by the difference between the two lines. The properties of $\mbox{BR}(d)$ are complex because $b_r(d)$ is increasing in~$d$, whereas the state price $q((s,d), (r,b_r(d)))$ is decreasing in both~$d$ and~$b_r$. By Proposition~\ref{prop:st}, there are transfers next-period for any debt $d<d_{\max}$, and hence, $\mbox{BR}(0)$ is strictly positive. Moreover, since $b_r(d)$ is constant for $d\leq d^c$, $\mbox{BR}(d)$ decreases linearly in this range. Hence, the fiscal reaction function $\tau(d)$ increases linearly in~$d$ for $d\leq d^c$. There is an intersection point $d_{\bal}$ where the bond revenue is equal to the current debt, $\tau(d_{\bal})=0$. For $d<d_{\bal}$, bond revenue exceeds the current debt, and the planner subsidizes the young, that is, there is a primary fiscal deficit. For $d>d_{\bal}$, bond revenue is insufficient to cover the current debt, and the planner taxes the young, that is, there is a primary fiscal surplus. For $d>d^c$, a rise in~$d$, that is, a reduction in the fiscal space, leads to more bond issuance but the price of bonds decreases. Thus, the net effect of a change in~$d$ on bond revenue is generally ambiguous. For the example illustrated in Panel~B, the fiscal reaction function $\tau(d)$ is increasing in~$d$ but initially at a slower rate for debt above the threshold level and then at a higher rate when debt is sufficiently large.

The situation depicted in Figure~\ref{fig:Debt} contrasts with the two benchmarks discussed in Section~\ref{sec:bench}. At the first best, the debt policy function is $b_r(d)=d^\ast(r)$, independent of~$d$. Hence, the debt policy functions in Panel~A of Figure~\ref{fig:Debt} are horizontal lines with fixed points at $d^\ast(r)$. There are no dynamics of debt except in the initial period, although debt varies with the endowment share. Ignoring the nonnegativity constraint on transfers, the first-best bond revenue function is linearly decreasing in debt, resulting in a fiscal reaction function that is linearly increasing.\footnote{It can be shown that $\mbox{BR}^\ast(d)=(a-1)(1-d)$ where $a=(1-\delta)+(\beta+\delta)\mathbb{E}_ss$ and $\mathbb{E}_ss$ is the expected endowment share. Hence, the fiscal reaction function is $\tau^\ast(d)=(1-a)+ad$. Since $\mathbb{E}_ss>\delta/(\beta+\delta)$, $a>1$.} In the deterministic case, the debt policy function is a transformation of the policy function in Figure~\ref{fig:pol} with a critical debt $d^c=(\exp(\omega^c)-(1-s))/s$. If the initial debt is above $d^c$, debt falls, and once it reaches or falls below $d^c$, the debt next period equals the first-best level $d^\ast$. The dynamics of debt are transitory, with debt reaching the fixed point $d^\ast$ in finite time. Along the transition path, debt falls, and the price of debt rises. These two offsetting effects mean it is possible that bond revenue rises or falls during the transition.

The two benchmarks show that enforcement frictions lead to the nonlinearity of the fiscal reaction function. By showing how this arises within an optimizing framework, the paper contributes to the literature that examines and provides evidence of this nonlinearity \citep[see, for example][among others]{Mendoza-Ostry08, Ghoshetal13}.

\section{Asset Pricing Implications}
\label{sec:risk}

In this section, we examine the asset pricing implications of the model.\footnote{%
We follow several authors in addressing asset pricing in overlapping generations models \citep[see,
for example,][]{Huberman84, Huffman86, Labadie86} and \citet{Garleanu-Panageas23} for a recent contribution.} In an overlapping generations model, the growth-adjusted stochastic discount factor is given by the intertemporal marginal rate of substitution $m(x, x^\prime) \colonequals \beta {u_c(1-c(x^\prime))}/{u_c(c(x))}$ where $x$ is the current state, $x^\prime$ is the state next period, $u_c(c(x))$ is the marginal utility of the current young and $u_c(1-c(x^\prime))$ is their marginal utility when old. This stochastic discount factor can be decomposed into two components:
\begin{gather}
m(x, x^\prime) = \underset{m_A(x,x^\prime)}{\underbrace{{\delta \left( \frac{u_c(c(x^\prime))}{u_c(c(x))}\right)}}}\underset{m_B(x,x^\prime)}
{\underbrace{{\left(\frac{\beta}{\delta}\frac{u_c(1-c(x^\prime))}{u_c(c(x^\prime))}\right)}}}.   \label{SDF}%
\end{gather}
The first component $m_A(x,x^\prime)$ represents risk sharing \emph{across\/} two adjacent generations of the young and the second component $m_B(x, x^\prime)$ represents risk sharing \emph{between\/} the young and the old at a given date. In a representative agent model, $m(x,x^\prime)=m_A(x,x^\prime)$ and the variability in the stochastic discount factor is determined by the variability of consumption, which in an endowment economy depends on the variability of the aggregate endowment. In contrast, in an overlapping generations model, if there is variability in the degree of risk sharing between the young and the old, then there is variability in $m_B(x, x^\prime)$, which interacts with the variability in $m_A(x,x^\prime)$ with consequent implications for asset pricing. In the optimal sustainable intergenerational insurance, the variability of $m_B(x,x^\prime)$ is determined by the first-order condition~\eqref{rmu1} and the updating rule~\eqref{eqn:update}. This variability depends on the current endowment share and the outstanding debt. To simplify the discussion, we confine attention to states in the ergodic set.\footnote{%
Limiting the analysis to the ergodic set is justified for two reasons. First, there is convergence to the ergodic set within finite time, as shown in Section~\ref{sec:conv}. Second, absent constraint~\eqref{PK0}, the planner sets the initial debt to $d_{\min}$, which lies in the ergodic set. Proposition~\ref{prop:convergence} shows that the ergodic set is countable. However, for simplicity and because it corresponds to our numerical procedures, we assume additionally that the ergodic set is finite. This is justifiable because it is possible to adapt the arguments to the denumerable case or even more general state spaces \citep[see, for example,][]{Hansen-Scheinkman09, Christensen17}.} We also assume that the bounds on debt do not bind. In this case, the first best exhibits complete insurance with the consumption share independent of the endowment state.\footnote{Although it is restrictive to assume that the bounds on debt are nonbinding, it simplifies the analysis, and we will note how results differ when the bounds are binding. }

Let $Q$ denote the matrix of state prices $q(x,x^{\prime})= \pi(r)m(x,x^{\prime})$ where $x=(s,d)$ and $x^\prime=(r, b_r(d))$, and let $\varrho$ and $\psi$ be the Perron root and corresponding eigenvector of $Q$. The Ross Recovery Theorem \citep{Ross15} shows that the $k$-period stochastic discount factor $m^k(x,x^\prime) = \varrho^k \psi(x)/\psi(x^\prime)$ where $\varrho$ and $\psi(x)$ can be interpreted as the discount factor and inverse marginal utility of a pseudorepresentative agent. Using the first-order condition~\eqref{rmu1} and the updating rule~\eqref{eqn:update}, $f(x^\prime)/(1-f(x^\prime)) = (\delta/\beta)(1+\mu(x^\prime))/(1+\mu(x))$  where $f(x)=s(1-d)$ is the consumption share of the young and $\mu(x)$ is the multiplier on the corresponding participation constraint. To ease notation, let $\nu(x)\colonequals 1+\mu(x)$ and $\nu_{\max}\colonequals\max_x \nu(x)$. Since we show below that $\varrho=\delta$, it follows from equation~\eqref{SDF} that $\psi(x) = f(x)/\nu(x)$.\footnote{The multiplicative decomposition of $\psi(x)$ into the components $f(x)$ and $1/\nu(x)$ is reminiscent of a number of other asset pricing models \citep[see, for example,][]{Bansal-Lehmann_97}.} The unit price of a $k$-period discount bond in state~$x$, $p^k(x)$, is given by the corresponding row sum of $Q^k$, the $k$-fold matrix power of $Q$. The corresponding yield is $y^{k}(x)\!\colonequals\!-(1/k)\log(p^{k}(x))$ and the yield on the long bond is $y^{\infty}(x) \colonequals \lim_{k\to\infty}y^k(x)$.

\citet{Martin-Ross19} shows that $\abs{y^k(x)-y^{\infty}(x)}\leq (1/k)\Upsilon$ for $\Upsilon\!\colonequals\!\log({\psi_{\max}}/{\psi_{\min}})$ where $\psi_{\max}$ and $\psi_{\min}$ are the maximum and minimum values of $\psi$. That is, $\Upsilon$ measures the range of the eigenvector and bounds the deviation of the yield from its long-run value. A low value of $\Upsilon$ means that the yield curve is relatively flat and that yields are not very sensitive to changes in debt.\footnote{The bound $\Upsilon$ provides a measure of the variability of the yields. Two alternative measures used to assess how risk is shared are the insurance coefficient \citep[see, for example,][]{Kaplan-Violante10} and the consumption equivalent welfare change \citep[see, for example,][]{Songetal15}. We discuss these alternatives in Part~C of the Supplemental Material and show that these two measures share similar comparative static properties with the bound $\Upsilon$.}

The matrix $Q$ is the growth-adjusted state price matrix. Let $q^k_{+}(x, x^\prime)$ and $m^k_{+}(x, x^\prime)$ denote the unadjusted state prices and marginal rate of substitution conditional on state~$x$ when $x^\prime$ is the state, and $\gamma$ is the growth factor $k$-periods ahead. Since the consumption shares are independent of the history of shocks to growth rates (Proposition~\ref{prop:iid}) and the shocks to growth rates are i.i.d., it can be checked that $q_+^k(x, x^\prime)=\varsigma(\gamma)\bar{\gamma}^{-k}(\bar{\gamma}/\gamma)q^k(x, x^\prime)$ and $m_+^k(x,x^\prime)=\bar{\gamma}^{-k}(\bar{\gamma}/\gamma)m^k(x,x^\prime)$ where $q^k(x, x^\prime)=\pi^k(x, x^\prime)m^k(x,x^\prime)$.\footnote{With stochastic growth, the Ross Recovery Theorem does not recover the true probability transition matrix. Instead, it recovers a transition matrix where probabilities are weighted by the relative growth factors \citep[see, for example,][for a discussion]{Borovickaetal16}.} Similarly, let $y_+^{k}(x)$ denote the yield on the $k$-period bond in the unadjusted case. Then, we can establish the following proposition.

\begin{proposition}\label{prop:Ross}
For each $x\in E$:
\begin{enumerate*}[label=(\roman*)]
\item $y_+^k(x)=y^k(x)+\log(\bar{\gamma})$ for each $k=1,2\ldots$.
\item In the limit, $y_+^{\infty}(x)=y^{\infty}+\log(\bar{\gamma})$ with $y^{\infty}=-\log(\delta)$.
\item $y^k(x)$ is increasing in~$d$ for each $s$ and~$k$.
\item The long-short spreads satisfy $y^\infty-y^1(1,d^{\ast}(1))>0>y^{\infty}-y^1(I,d^{\ast}(I))$.
\item The Martin-Ross measure $\Upsilon=\log(\nu_{\max})$ where $\nu_{\max}=\nu(I, d^\ast(I))$.
\end{enumerate*}
\end{proposition}

Part~(i) of Proposition~\ref{prop:Ross} shows that the difference between the yields in the growth-adjusted and unadjusted cases is simply the average growth rate as measured by $\log(\bar{\gamma})$, independent of the current state~$x$ or the bond maturity~$k$. This independence follows from Assumption~\ref{ass:iid} that the growth shocks are i.i.d., meaning that each generation faces the same growth risk. A similar result, that market risk premia are unaffected by market incompleteness, is established by \citet{Krueger-Lustig10} in a model with infinitely-lived agents and uninsurable idiosyncratic as well as aggregate risk. Part~(ii) follows from the result of \citet{Martin-Ross19} that the yield on the long bond is $y^{\infty} = -\log(\varrho)$, independent of~$x$, and that $\varrho=\delta$ when the upper bound and nonnegativity constraints do not bind.\footnote{If the upper bound constraint does not bind, then $\varrho\leq \delta$ and if the nonnegativity constraints do not bind, then $\varrho\geq\delta$.} To understand Part~(iii), note that the consumption share of the young is decreasing in~$d$ and that, since $b_r(d)$ is increasing in~$d$ from Corollary~\ref{cor:debtf}, the consumption share of the old next period is increasing in~$d$. Consequently, the stochastic discount factor $m(x, x^\prime)$ decreases in~$d$. Since the transition probabilities do not depend on~$d$, the price of the one-period discount bond is decreasing in~$d$, or equivalently, its yield is increasing in~$d$. Thus, an agent born into a generation with higher debt faces higher one-period yields. Since bond prices are linked recursively, this property holds for bonds of any maturity.

Part~(iv) of Proposition~\ref{prop:Ross} shows that the long-short spread $y^\infty-y^1(x)$ is positive when the young have a low endowment share and debt is low. In this case, it follows from Section~\ref{sec:debt} that debt is expected to rise in the future with a corresponding increase in yields. Conversely, the long-short spread is negative when the young have a high endowment share, and debt is high, in which case, both debt and yields are expected to fall in the future. Part~(v) shows that the bound $\Upsilon$ is determined by the multiplier on the participation constraint,  $\nu_{\max}$, corresponding to the fixed point of the debt policy function for the largest endowment share. That is, the bound on the variability of the yield curve is determined by the tightness of the participation constraint at the largest debt in the ergodic set.

To help understand the results of Proposition~\ref{prop:Ross}, consider the first-best and deterministic benchmarks of Section~\ref{sec:bench}. At the first best, the debt policy functions are constants, and the yield curve is flat with $y_+^k(x)=-\log(\delta)+\log(\bar{\gamma})$ and $\Upsilon=0$. Despite the flat yield curve, the risk premium on the aggregate risk is positive because the return on debt is high when the growth rate is high. Specifically, the expected return on a one-period bond is $\mathbb{E}_\gamma \gamma/\delta$, while the risk-free rate is $\bar{\gamma}/\delta$. Thus, the risk premium is $(\mathbb{E}_\gamma \gamma-\bar{\gamma})/\delta$, which is strictly positive when the growth shocks are nondegenerate. A Lucas tree or any other asset that pays a share of the aggregate endowment will carry this positive risk premium, so that the risk premium on aggregate risk corresponds to the risk premium on debt with complete insurance. In the deterministic case, the risk premium is zero. However, along the transition path, as debt falls, the yield $y^k(d)$ decreases to its long-run value of $y^\infty=-\log(\delta)+\log(\gamma)$, where $\gamma$ is the deterministic growth rate. Thus, $\Upsilon>0$ in the transition, even though there is no risk.\footnote{The ergodic set is degenerate at $d^\ast$  in the deterministic case. Once debt reaches this level, the yield curve is flat.}

\section{Debt Valuation}\label{sec:debtval}

The budget constraint in equation~\eqref{eqn:budget} can be iterated forward to show that current debt equals the present value of all future primary surpluses.\footnote{\citet{Jiangetal23} define \textit{fiscal capacity\/} as the present value of future surpluses. Since, in our model, debt is determined optimally, there is no mispricing or bubble component, and debt and fiscal capacity are equivalent in this sense. Other authors often use the term fiscal capacity more broadly to encompass both the debt limit and fiscal space.} As pointed out by \citet{Bohn95}, this present value depends on the risk premium on debt. In this section, we focus on the multiplicative risk premium on debt because it is the negative of the covariance between the stochastic discount factor and the return on debt and because this covariance is independent of the endowment share. When there is a growth shock~$\gamma$, the unadjusted return on debt is $R_+(x, x^\prime)=rb_r(d)\gamma e/(s\mbox{BR}(d)e)$, where $s\mbox{BR}(d)e$ is the value of bonds issued today. The multiplicative risk premium is $\mbox{MRP}_+(d) = (\bar{R}_+(x)-R_+^f(x))/R_+^f(x)$ where $\bar{R}_+(x)$ is the expected return on debt and $R_+^f(x)$ is the risk-free rate on interest in state~$x$. Denote the corresponding growth-adjusted values by $\mbox{MRP}(d)$, $\bar{R}(x)$ and $R^f(x)$. As shown in Section~\ref{sec:risk}, the risk premium on debt with complete insurance equals the risk premium on aggregate risk and we denote the common multiplicative risk premium by $\mbox{MRP}^{\ast}$. The following proposition shows that the multiplicative risk premium has a linear decomposition that depends on the growth-adjusted multiplicative risk premium and the multiplicative risk premium with complete insurance.

\begin{proposition}\label{prop:mrp}
The multiplicative risk premium $\mbox{MRP}_{+}(d) = \mbox{MRP}(d) + \alpha(d) \mbox{MRP}^\ast$, where $\alpha(d) = \bar{R}(x)/R^f(x)$. The components satisfy:
\begin{enumerate*}[label=(\roman*)]
    \item $\mbox{MRP}^\ast= (\mathbb{E}_\gamma \gamma-\bar{\gamma})/\bar{\gamma}\geq 0$;
    \item $\mbox{MRP}(d) < 0$; and
    \item $0<\alpha(d)<1$.
\end{enumerate*}
\end{proposition}

The decomposition of $\mbox{MRP}_{+}(d)$ into components depending on $\mbox{MRP}(d)$ and $\mbox{MRP}^\ast$ is analogous to the result of Proposition~\ref{prop:Ross} that the conditional yield is the sum of a growth-adjusted yield and a component corresponding to the average growth rate. In the same way as Proposition~\ref{prop:Ross}, this decomposition follows from Assumption~\ref{ass:iid} that the shocks to growth rates and endowment shares are independent of each other and i.i.d. Part~(i) of Proposition~\ref{prop:mrp} shows that $\mbox{MRP}^\ast$ is nonnegative. As discussed in Section~\ref{sec:risk}, $\mbox{MRP}^\ast$ is strictly positive when growth shocks are nondegenerate. To understand Part~(ii), note that the growth-adjusted return $R(x, x^\prime) = rb_r(d)/(s\mbox{BR}(d))$ is increasing in~$r$, from Part~(ii) of Corollary~\ref{cor:debtf}. Moreover,  the consumption share of the old is decreasing in~$r$, from Lemma~\ref{lem:g}, and hence, the stochastic discount factor $m(x, x^\prime)$ is increasing in~$r$. Consequently, the returns are high when the marginal utility of consumption of the old is high, resulting in a positive covariance term, and correspondingly, a negative growth-adjusted multiplicative risk premium. By comparison, with complete insurance, the stochastic discount factor is constant so that its covariance with the returns is zero, and hence, $\mbox{MRP}(d)=0$. As noted in equation~\eqref{SDF}, the stochastic discount factor comprises two components that measure risk sharing across two adjacent generations of the young and risk sharing between the young and the old. The first component $m_A(x,x^\prime)$ is decreasing in~$r$, whereas the second component $m_B(x,x^\prime)$ is increasing in~$r$. In a representative agent model, only $m_A(x,x^\prime)$ is present, and high debt returns are associated with a low marginal utility of consumption of the young, generating a positive risk premium. In contrast, $m_B(x,x^\prime)$ dominates in the overlapping generations model, making debt a negative beta asset.

Part~(iii) of Proposition~\ref{prop:mrp} shows that $\alpha(d)<1$, and hence, the gap
$\mbox{MRP}^\ast\shortminus\mbox{MRP}_{+}(d)>0$ for each $d$. That is, the multiplicative risk premium 
on debt is lower than the multiplicative risk premium on aggregate risk. Using \ac{US} data, \citet{Jiangetal19} show that the observed value of debt is higher than the present value of future primary surpluses when discounted using the risk premium on aggregate risk, a debt valuation puzzle. Convenience yields, seigniorage and other service flow values have been offered as potential explanations for this puzzle. Our results suggest an additional explanation. In the presence of enforcement frictions, risk sharing is partial and debt serves as a hedge against idiosyncratic risk, lowering the risk premium and raising the value of debt.\footnote{\citet{Jiangetal20} examine how to manufacture risk-free government debt. With the primary surplus disaggregated into tax and expenditure components, the risk premium on debt is a weighted average of the risk premiums on taxes and expenditure. Consequently, the risk premium on debt can be eliminated, but only at the cost of making taxes or expenditures less cyclical. Since we do not distinguish between taxes and expenditure, the risk premium on the primary surplus equals the risk premium on debt, and making debt risk-free may not be feasible or desirable.}

Part~(iii) of Proposition~\ref{prop:mrp} also shows that the gap $\mbox{MRP}^\ast\shortminus\mbox{MRP}_{+}(d)$ depends on~$d$, evolving according to the dynamics of debt outlined in Section~\ref{sec:debt}. For $d\leq d^c$, this gap is independent of~$d$. For $d>d^c$, the effect of debt on the size of the gap is ambiguous. From Proposition~\ref{prop:Ross}, the risk-free interest rate increases with debt. Therefore, the gap rises or falls depending on whether the expected return on debt increases with debt at a faster or slower rate than the risk-free interest rate. Although the overall effect is ambiguous, Section~\ref{sec:example} provides an example in which $\mbox{MRP}^\ast\shortminus\mbox{MRP}_{+}(d)$ decreases with~$d$ for $d>d^c$.

\section{Two-State Example}
\label{sec:example}

Finding the optimal sustainable intergenerational insurance is complex because it involves solving the functional equation of problem~\ref{eqn:p1}. In this section, we present an example with $I=2$ that can be solved using a shooting algorithm.\footnote{Part~E of the Supplemental Material provides details of the shooting algorithm.} For this case, we solve for the invariant distribution and derive a closed-form solution for the Martin-Ross measure.

Suppose there are two possible endowment shares for the young: $s(1)=\kappa-\epsilon(1-\pi)/\pi$ and $s(2)=\kappa+\epsilon$, where $\pi=\pi(1)$, $\kappa\geq\sfrac12$ and $\epsilon>0$. The young are poor in state~1 and rich in state~2. An increase in $\epsilon$ is a mean-preserving spread of the risk. By Assumptions~\ref{ass:LS} and~\ref{ass:nofb}, $d^\ast(2)>d^0(2)>d^c>d^0(1)=d^{\ast}(1)=0$. By Corollary~\ref{cor:debtf}, the debt policy functions satisfy $b_2(d)>b_1(d)$. We make two
additional assumptions.
\begin{assumption}\label{ass:canonic}
\begin{enumerate*}[label=(\roman*)] \item $d^{\ast}(2)<d_{\max}$; and \item $b_1(d^{\ast}(2))<d^c$. \end{enumerate*}
\end{assumption}
Part (i) of Assumption~\ref{ass:canonic} implies that the debt limit never binds. By Part~(ii), debt is below $d^c$ whenever state~1 occurs. In such a case, the history of endowment states is forgotten once state~1 occurs and the dynamics of debt depend only on the number of consecutive state 2s in the most recent history, starting from the resetting level $d^0(s)$. The longer is the sequence of state~2s, the higher is debt. The set of parameter values that satisfy Assumption~\ref{ass:canonic}, as well as Assumptions~\ref{ass:sustain}-\ref{ass:nofb}, is nonempty with the following belonging to this set.
\begin{example} \label{exmpl:canon} $\delta=\beta=\exp(-\sfrac{1}{75})$, $\pi=\sfrac{1}{2}$, $\kappa=\sfrac{3}{5}$, and $\epsilon=\sfrac{1}{10}$.
\end{example}%
To simplify notation, let $d^{(n)}(s)$ be the debt in state~$s=1,2$ after $n$ consecutive state~2s, where  $d^{(0)}(s)=d^0(s)$ are the resetting levels and $\lim_{n\to\infty}d^{(n)}(2)=d^\ast(2)$. Under Assumption~\ref{ass:canonic}, the invariant distribution of debt is a transformation of a geometric distribution and the bound $\Upsilon$ has a closed-form solution.

\begin{proposition}\label{prop:canonic}
Under Assumption~\ref{ass:canonic}:
\begin{enumerate*}[label=(\roman*)]
\item The ergodic set $E=\{(s,d^{(n)}(s))_{n\geq0, s=1,2}\}$ with a probability mass function
$\phi(s,{d}^{(n)}(s))=\phi(s,{d}^0(s))(1-\pi)^{n}$ for $n\geq1$ where $\phi(1,{d}^0(1))=\pi^2$ and $\phi(2,{d}^0(2))=\pi(1-\pi)$.
\item $\Upsilon=\log(\delta/\beta)-\log(\chi^{-1}-1)$ where
\end{enumerate*}
\begin{gather*}
\chi =\textstyle\left(\frac{\delta}{\beta}\right)^{\frac{1-\pi}{\pi}}\left( \frac{\beta +\delta }{\delta }\right) ^{\frac{1+\beta (1-\pi) }{\beta\pi }}\left( \kappa +\epsilon \right) ^{\frac{1}{\beta \pi }}\left(1-\kappa -\epsilon \right) ^{\frac{1-\pi }{\pi }}\left( 1-\kappa +\epsilon \frac{1-\pi }{\pi }\right).
\end{gather*}
\end{proposition}
As stated in Part~(i) of Proposition~\ref{prop:canonic}, the invariant distribution has a probability mass of $\phi(1,{d}^0(1))=\pi^2$ and $\phi(2,{d}^0(2))=\pi(1-\pi)$ at the regeneration states and zero probability mass at states $(s,b_s(d^{\ast}(2)))$. Panel~A of Figure~\ref{fig:numerical} plots the invariant distribution for the parameter values of Example~\ref{exmpl:canon}. Low debt levels occur only in state~1, while high levels occur only in state~2. 

Part~(ii) of Proposition~\ref{prop:canonic} provides a closed-form solution for the bound $\Upsilon$. By Proposition~\ref{prop:Ross}, the bound is strictly positive and determined by the tightness of the participation constraint of the young when $x=(2,d^{\ast}(2))$. Using this closed-form solution, it is easily checked that $\Upsilon$ decreases with the discount factors $\beta$ or $\delta$, that is, as either the agent or the planner becomes more patient. Moreover, $\Upsilon$ decreases with the average endowment share to the young, $\kappa$, and increases with risk, $\epsilon$.\footnote{Part~C of the Supplemental Material presents some comparative static properties of $\Upsilon$ for parameter values that violate Assumption~\ref{ass:canonic}.}

\begin{figure}[ht]
\begin{center}
\includegraphics{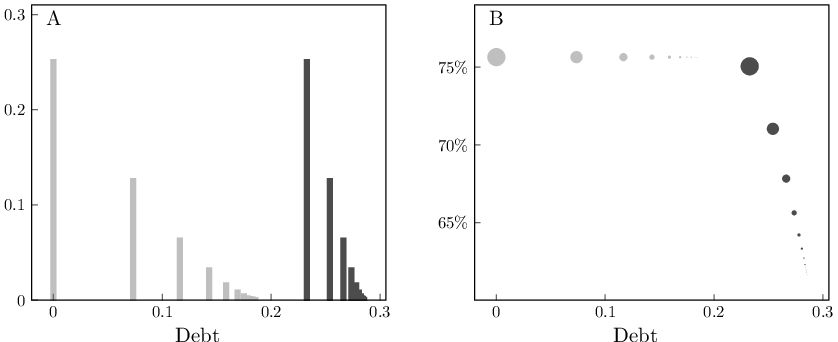}
\end{center}
\caption{Invariant Distribution (Panel~A) \& Multiplicative Risk Premium (Panel~B).} \label{fig:numerical}
\begin{figurenotes}
Both panels use the parameter values of Example~\ref{exmpl:canon}. Panel~A plots the invariant distribution $\phi(s,{d}^{(n)}(s))$ for $s=1$ (light-gray bars) and $s=2$ (dark-gray bars). Panel B plots the distortion of the multiplicative risk premium relative to its value under complete insurance, $\Delta^\ast(\mbox{MRP}(d))=(\mbox{MRP}^{\ast}\shortminus\mbox{MRP}_{+}(d))/\mbox{MRP}^{\ast}$, for the values of~$d$ in the ergodic set. Light-gray dots correspond to $s=1$ and the dark-gray dots to $s=2$. The size of each dot indicates the frequency of occurrence.
\end{figurenotes}
\end{figure}

Panel~B of Figure~\ref{fig:numerical} illustrates the impact of debt on the risk premium in a version of Example~\ref{exmpl:canon} with stochastic growth. In this example, the arithmetic mean growth rate is set to 4\% and the corresponding multiplicative risk premium is approximately 5\%. Proposition~\ref{prop:mrp} shows that $\mbox{MRP}^\ast>\mbox{MRP}_+(d)$ and Panel~B illustrates that the gap is constant when debt is low, but decreases with debt when debt is high. As noted in Section~\ref{sec:debtval}, the multiplicative risk premium may increase or decrease with debt for $d>d^c$, depending on the relative magnitude of the effect of debt on its return and the marginal utility of consumption of the old. In this example, the effect on the return dominates causing the risk premium to rise with debt. Since the risk premium on aggregate risk is independent of debt, a rise in debt narrows the gap between the risk premiums on aggregate risk and debt.

\section{Conclusion}\label{sec:conc}

The paper has developed a theory of intergenerational insurance in a stochastic overlapping generations model with limited enforcement of risk-sharing transfers. Despite the stationarity of the underlying economic environment, the generational risk is spread across future generations in ways that cause transfers to be history dependent. There is periodic resetting, and the history of shocks is forgotten when this occurs. By interpreting intergenerational insurance in terms of debt, we provide a theory of the dynamics of debt that offers a new perspective on the fiscal reaction function and the sustainability and valuation of debt. With complete insurance, the fiscal reaction function is linear, and the risk premium on debt equals the risk premium on aggregate risk. When there are enforcement frictions, intergenerational insurance is incomplete, the fiscal reaction function is nonlinear, and the risk premium on debt is below the risk premium on aggregate risk.

The results suggest several potential directions for future research. First, the qualitative predictions about the dynamics of debt could be compared with historical data for advanced economies, for example, with a specific focus on the baby boom and subsequent generations. Second, the model has no heterogeneity within a generation. Enriching the demographic structure of the model, either by having more than two overlapping generations or allowing for heterogeneity within the same generation, would make it possible to address the interdependence between intergenerational and intragenerational insurance. Third, to study the interplay between self-insurance and intergenerational insurance, a technology that can transform endowments across dates could be added. Finally, incorporating a stochastic demand for public good provision would allow the study of the risk premia associated with the various components of the primary surplus. 

\eject

\section*{Appendix (Proofs of Main Results)}\label{app:A}

This Appendix contains the proofs of the main results. Omitted proofs can be found in Part~B of the Supplemental Material.

\begin{proof}[Proof of Lemma~\ref{lem:g}]
\NoBreakPar
\begin{enumerate}[itemsep=1pt,align=left,leftmargin=0pt,labelsep=0pt,itemindent=2\parindent,labelwidth=2\parindent,listparindent=\parindent,label=(\roman*)]
\item Since the constraint set $\Phi(s,\omega)$ is convex and the objective function is strictly concave, the policy function $g_r(\omega, s)$ is single-valued and continuous in $\omega$. Let $h_s(\omega)\colonequals-(\delta/\beta)V_{\omega}(s, \omega)$ where $h_s\mathpunct{:}\Omega(s)\to[\lambda_{\min}(s), \lambda_{\max}(s)]$ with $\lambda_{\min}(s) =\max\{0,(\delta/\beta)((1-s)/s)-1\}$.  Let $h_s^{-1}\mathpunct{:}[\lambda_{\min}(s), \lambda_{\max}(s)]\to\Omega(s)$ be its inverse. By the concavity of the frontier $V(s,\omega)$ in~$\omega$, $h_s^{-1}(\lambda)$ is strictly increasing in $\lambda$ for $\lambda>\lambda_{\min}(s)$. Suppose first that $\omega\geq\omega^{0}(s)$. Hence, from~\eqref{eqn:pk}, $f(s,\omega)=1-\exp(\omega)$. Since $g_r(s,\omega)=\max\{\omega_{\min}(r), \min\{\omega_{\max}(r), h_r^{-1}(\mu(s,\omega))\}\}$, substituting into~\eqref{eqn:icr}, there is a unique value (possibly zero) of $\mu$ that satisfies the constraint. If $\mu(s,\omega)=0$, then $g_r(s,\omega)=\omega^{0}(r)$ for each~$r$. If $\mu(s,\omega)>0$, then $\mu(s,\omega)$ is strictly increasing in~$\omega$ since $f(s,\omega)$ is strictly decreasing in~$\omega$ and $h_r^{-1}(\mu)$ is increasing in $\mu$. Thus, $g_r(s,\omega)$ is strictly increasing in~$\omega$ for $g_r(s,\omega)\in(\omega^{0}(r), \omega_{\max}(r))$. If $\omega< \omega^{0}(s)$, then $\lambda(s,\omega)=0$ and hence, since $f(s,\omega)$ is independent of~$\omega$, $g_r(s,\omega)$ is also independent of~$\omega$.
\item Consider states $s(i)>s(i\shortminus 1)$, $i=2,\ldots,I$. For brevity, write $g_r(i, \omega)$ for $g_r(s(i), \omega)$ and $g_i(s,\omega)$ for $g_{r(i)}(s,\omega)$ etc. We first show that $\mu(i,\omega)\geq\mu(i\shortminus 1, \omega)$ for $\omega\in[\omega_{\min}(i\shortminus 1), \omega_{\max}(i)]$ with a strict inequality unless $\mu(i,\omega)=\mu(i\shortminus 1, \omega)=0$. Suppose to the contrary that $\mu(i\shortminus 1, \omega)\geq \mu(i,\omega)>0$. It follows from~\eqref{rmu1} that $g_r(i\shortminus 1, \omega)\geq g_r(i,\omega)$. Using~\eqref{eqn:icr} and $\hat{\upsilon}(s)=\log(s)+\beta\sum_r\pi(r)\log(1-r)$, gives
    \begin{gather*}\label{eqn:muchange}
    \log(f(i\shortminus 1, \omega))-\log(f(i,\omega)) = \left(\hat{\upsilon}(i\shortminus 1)-\hat{\upsilon}(i)\right)+  \beta \sum\nolimits_{r} \pi(r)\left(g_r(i,\omega)-g_r(i\shortminus 1, \omega)\right).
    \end{gather*}
    Since $\hat{\upsilon}(i\shortminus 1)-\hat{\upsilon}(i)<0$ and $g_r(i,\omega)-g_r(i\shortminus 1, \omega)\leq0$, $f(i,\omega)>f(i\shortminus 1, \omega)$ and $\log(1-f(i\shortminus 1, \omega))>\log(1-f(i,\omega))\geq\omega$. Hence, $\lambda(i\shortminus 1, \omega)=0\leq \lambda(i,\omega)$. However, since $\lambda(i,\omega)\geq \lambda(i\shortminus 1, \omega)$ and $\mu(i\shortminus 1, \omega)\geq \mu(i,\omega)$, it follows from~\eqref{rmu1} that $f(i\shortminus 1, \omega)\geq f(i, \omega)$, a contradiction. Hence, if $\mu(i\shortminus 1, \omega)=\mu(i,\omega)=0$, then $g_r(i\shortminus 1, \omega)=g_r(i,\omega)=\omega^{0}(r)$ independent of~$s$. If, however, $\mu(i\shortminus 1, \omega)>0$, then it follows from~\eqref{rmu1} that $g_r(i\shortminus 1, \omega)< g_r(i, \omega)$ for $\omega\in [\omega_{\min}(i\shortminus 1), \omega_{\max}(i)]$. By Assumption~\ref{ass:LS}, $\mu(1, \omega^{0}(1))=0$ and by Assumption~\ref{ass:nofb}, $\mu(I, \omega^{0}(I))>0$. Since $\mu(s,\omega)$ is increasing in~$\omega$, $\mu(I, \omega^{0}(I))>0$ and $\mu(I, \omega)>\mu(1, \omega)$ for $\omega\in(\omega^{0}(1), \omega_{\max}(I))$. Hence, from~\eqref{fp1}, $V_\omega(r, g_r(I, \omega))<V_\omega(r, g_r(1, \omega))$ and therefore, from the strict concavity of $V(r, \omega)$ in~$\omega$ for $\omega>\omega^{0}(1)\geq \omega^{0}(r)$, it follows that $g_r(I, \omega)>g_r(1, \omega)$.
    \\
    Next, if $g_{i}(x)\leq\omega^{0}(i\shortminus 1)$ or $g_{i\shortminus 1}(x)\geq\omega_{\max}(i)$, then $g_{i\shortminus 1}(x)\geq g_i(x)$.  Therefore, suppose $g_{i}(x), g_{i\shortminus 1}(x)\in(\omega^{0}(i\shortminus 1), \omega_{\max}(i))$. We first show that $V_{\omega}(i\shortminus 1, \omega) \geq V_\omega(i,\omega)$ for $\omega\in(\omega^{0}(i\shortminus 1), \omega_{\max}(i))$. For $\omega>\omega^{0}(i\shortminus 1)$, it follows that $\lambda(i\shortminus 1,\omega)>0$ and since $\omega^{0}(i\shortminus 1)\geq \omega^{0}(i)$, $\lambda(i, \omega)> 0$. Therefore, $f(i,\omega)=f(i\shortminus 1, \omega)$. In this case, it follows from above that $\mu(i,\omega)\geq\mu(i\shortminus 1, \omega)$ with equality only if $\mu(i,\omega)=\mu(i\shortminus 1, \omega)=0$. Hence, it follows from~\eqref{rmu1} that $\lambda(i\shortminus 1, \omega)\leq \lambda(i, \omega)$ with strict inequality if $\mu(i,\omega)>0$. Using~\eqref{env1}, it follows that $V_{\omega}(i\shortminus 1, \omega) \geq V_\omega(i,\omega)$ with strict inequality if $\mu(i,\omega)>0$. For $g_{i}(x), g_{i\shortminus 1}(x) >\omega^{0}(i\shortminus 1)$, $\eta_{i}(x)=\eta_{i\shortminus 1}(x)=0$ and for $g_{i}(x), g_{i\shortminus 1}(x) < \omega_{\max}(i)$, $\xi_i(x)=\xi_{i\shortminus 1}(x)=0$. Hence, it follows from~\eqref{fp1} that $V_\omega(i,g_{i}(s,\omega))=V_\omega(i\shortminus 1, g_{i\shortminus 1}(s,\omega))$. Since $V_{\omega}(i\shortminus 1, \omega) \geq V_\omega(i,\omega)$, it follows from the concavity of $V(\cdot, \omega)$ in~$\omega$ that $g_{i\shortminus 1}(s,\omega)\geq g_{i}(s,\omega)$. The inequality is strict if $V_{\omega}(i\shortminus 1, \omega) > V_\omega(i,\omega)$ by the strict concavity of $V(\cdot, \omega)$ in~$\omega$. Since $\mu(I, \omega)>\mu(1, \omega)$ for $\omega\in(\omega^{0}(1), \omega_{\max}(I))$, $V_\omega(1,\omega)>V_\omega(I, \omega)$ and hence, $g_{1}(s, \omega)>g_{I}(s,\omega)$.

\item Since $\mu(1, \omega^{0}(1))=0$ and $f(1,\omega^{0}(1))=s(1)$, it follows that $g_r(1, \omega^{0}(1)) =\omega^{0}(r)$ for each~$r$. Since $\omega^{0}(r)>\omega_{\min}(r)$ for at least some~$r$, it follows that~\eqref{eqn:icr} is strictly slack and there is some $\omega^c>\omega^{0}(1)$ such that~\eqref{eqn:icr} is nonbinding with $g_r(1, \omega) =\omega^{0}(r)$ for each~$r$ and $\omega\in[\omega^{0}(1), \omega^c]$.

\item It follows from~\eqref{rmu1} that for $\omega=\omega^{\ast}(s)>\omega_{\min}(s)$, $\mu(s,\omega)=\lambda(s,\omega)$. In this case, $V_\omega(s,\omega^{\ast}(s))=V_\omega(r, g_r(s,\omega^{\ast}(s)))$ for $g_r(s,\omega^{\ast}(s))\in(\omega^{0}(r), \omega_{\max}(r))$, and, in particular, $g_s(s, \omega^{\ast}(s))=\omega^{\ast}(s)$, so that $\omega^{\ast}(s)$ is a fixed point of the mapping $g_s(s,\omega)$. Equally, for $\omega<\omega^{\ast}(s)$, it follows from~\eqref{rmu1} that $\mu(s,\omega)>\lambda(s,\omega)$, so that from the concavity of the frontier, $g_s(s,\omega)>\omega^{\ast}(s)$. Likewise, for $\omega>\omega^{\ast}(s)$, it follows from~\eqref{rmu1} that $\mu(s,\omega)<\lambda(s,\omega)$, so that from the concavity of the frontier, $g_s(s,\omega)<\omega^{\ast}(s)$. If $\omega^{\ast}(s)=\omega_{\min}(s)$, then $f(s,\omega)=s$ and $\mu(s,\omega^{\ast}(s))=0$ by Assumption~\ref{ass:sustain}. Hence, $g_s(s,\omega^{\ast}(s))=\omega^{\ast}(s)$. Since $\omega^{0}(s)$ is decreasing in~$s$, it follows by Assumption~\ref{ass:nofb} that $\omega^{0}(I)<\omega^f(I)\leq\omega^{\ast}$.
    \qedhere
\end{enumerate}
\end{proof}

\begin{proof}[Proof of Lemma~\ref{lem:c}]\NoBreakPar
\begin{enumerate}[itemsep=1pt,align=left,leftmargin=0pt,labelsep=0pt,itemindent=2\parindent,labelwidth=2\parindent,listparindent=\parindent,label=(\roman*)]
\item For $\omega>\omega^{0}(s)$, $\lambda(s,\omega)>0$ and therefore, it follows from~\eqref{eqn:pk} that $f(s, \omega)=1-\exp(w)$. For $\omega=\omega^{0}(s)$, either $\lambda(s,\omega^{0}(s))>0$ or $\lambda(s,\omega^{0}(s))=0$. In either case, it follows from~\eqref{eqn:pk} or the definition of $\omega^{0}(s)$ that $f(s, \omega^{0}(s))=1-\exp(w^{0}(s))$. For $\omega<\omega^{0}(s)$, it follows that $\lambda(s,\omega)=0$. From~\eqref{rmu1}, let $\phi(s,\mu)=\min\{\delta(1+\mu)/(\beta+\delta(1+\mu)),s\}$  where $\phi(s,\mu)$ is increasing in $\mu$ with $\phi(s, 0)=c^{\ast}(s)$. Recall that $h_r^{-1}(\mu)$, defined in the proof of Lemma~\ref{lem:g}, satisfies $V_{\omega}(r,h_r^{-1}(\mu))=-(\beta/\delta)\mu$ where $g_r(s,\omega)=\max\{\omega_{\min}(r), \min\{\omega_{\max}(r), h_r^{-1}(\mu(s,\omega))\}\}$. Since $h_r^{-1}(\mu)$ is increasing in~$\mu$, it follows from~\eqref{eqn:icr} that when $f(s, \omega^{0}(s))=1-\exp(w^{0}(s))$, there is a unique value of $\mu$, say $\mu^{0}(s)$, that solves the constraint. Furthermore, $\omega^{0}(s)=\log(1-\phi(s, \mu^{0}(s)))$.

\item Since $\hat{\upsilon}(i)>\hat{\upsilon}(i\shortminus 1)$, it follows from Part~(i) that $\mu^{0}(i)\geq \mu^{0}(i\shortminus 1)$ with strict inequality if $\mu^{0}(i)>0$. Therefore, since $\phi(s,\mu)$ is strictly increasing in $\mu$ and independent of~$s$ for $\mu>0$, $c^{0}(i)\geq c^{0}(i\shortminus 1)$ with strict inequality if $\mu^{0}(i)>0$. By Assumption~\ref{ass:nofb}, $\mu^{0}(I)>0$ and by Assumption~\ref{ass:LS}, $\mu^{0}(1)=0$. Hence, $c^{0}(I)>c^{0}(1)$.
\item Lemma~\ref{lem:g} establishes that at the fixed point, $\omega^f(s)=\min\{\omega_{\max}(s), \omega^{\ast}(s)\}$. Hence, $f(s,\omega^f(s)) \leq c^{\ast}(s)$ with equality for $\omega^f(s)< \omega_{\max}(s)$. \qedhere
\end{enumerate}
\end{proof}

\begin{proof}[Proof of Proposition~\ref{prop:convergence}]
Using the properties of $g_r(x)$ from Lemma~\ref{lem:g} and the argument in the text, it follows that there is an $k\geq1$ and $\epsilon>0$ such that $P^k(x, \{x_{0}\})>\epsilon$ for each $x\in\mathcal{X}$ and any $x_0$. Hence, Condition~{\bf M} of \citet[][page 348]{Stokey-Lucas89} is satisfied. Therefore, Theorem~11.12 of~\citet{Stokey-Lucas89} applies and there is strong convergence. Nondegeneracy with $\abs{E}>I$ follows from Assumption~\ref{ass:nofb}. The finiteness of the return times follows from Lemma~\ref{lem:g}(iii) and the finiteness of $\mathcal{I}$. The relationship between the probability mass and the expected return times and the pointwise convergence is standard \citep[see, for example, Theorems~10.2.3 and~13.1.2,][]{Meyn-Tweedie09}.
\end{proof}

\begin{proof}[Proof of Proposition~\ref{prop:Ross}]\NoBreakPar
\begin{enumerate}[itemsep=1pt,align=left,leftmargin=0pt,labelsep=0pt,itemindent=2\parindent,labelwidth=2\parindent,listparindent=\parindent,label=(\roman*)]
\item Since $q_+^k(x, x^\prime)=\varsigma(\gamma)\bar{\gamma}^{-k}(\bar{\gamma}/\gamma)q^k(x, x^\prime)$, summing over $x^\prime$ and $\gamma$, the unadjusted bond prices are $p^k_{+}(x)=\bar{\gamma}^{-k}p^k(x)$. Hence, the yields satisfy $y^k_{+}(x) = y^k(x)+\log(\bar{\gamma})$.
\item It is a standard result \citep[see, for example,][]{Martin-Ross19} that $\lim_{k\to\infty}y^k(x)=\mathbb{E}_\phi[\log(m(x,x^\prime))]= \log(\varrho)$, where $\mathbb{E}_\phi$ is the expectation taken over the invariant distribution of $x$ and $\varrho$ is the Perron root of the matrix $Q$. Taking logs in equation~\eqref{SDF}, $\log(m(x,x^\prime))= \log(\beta) +\log(c(x)) -\log(1-c(x^\prime))$. Using equations~\eqref{rmu1} and~\eqref{eqn:update}, assuming the nonnegativity constraints and the upper bound constraint do not bind, gives $\log(c(x^\prime))-\log(1-c(x^\prime)) = -\log(\beta/\delta) +\log(\nu(x^\prime))-\log(\nu(x))$, where $\nu(x)=1+\mu(x)$. Therefore, $\log(m(x,x^\prime))=\log(\delta) + \log(c(x)) - \log(c(x^\prime))+\log(\nu(x^\prime))-\log(\nu(x))$. Taking expectations at the invariant distribution, $\mathbb{E}_\phi[\log(m(x,x^\prime))]= \log(\delta)$. Hence, $\varrho=\delta$ and $\lim_{k\to\infty}y^k_+(x)=\log(\delta)+\log(\bar{\gamma})$.

\item Recall that $m(x,x^\prime)= m((s,d), (r, b_r(d)))=\beta s(1-d)/(1-r(1-b_r(d)))$. Since $b_r(d)$ is increasing in~$d$ by Corollary~\ref{cor:debtf}, it follows that $m(x,x^\prime)$ is decreasing in~$d$. The price of a one-period discount bond in state $(s,d)$ is $p^1(s,d) = \sum_{r}\pi(r) m((s,d), (r, b_r(d)))$, which is also decreasing in $d$. Making the induction hypothesis that the price of a $k$-period discount bond is decreasing in $d$, $p^{k+1}(s,d) = \sum_{r}\pi(r)m((s,d), (r, b_r(d)))p^{k}(r,b_r(d))$. Since $p^k(s,d)$ and $m((s,d), (r,b_r(d)))$ are positive and decreasing in~$d$, and $b_r(d)$ is increasing in~$d$, it follows that $p^{k+1}(s, d)$ is decreasing~$d$. Hence, the conditional yield $y^k(s,d)=-(1/k)\log(p^k(s,d))$ is increasing in~$d$ for each $s$ and~$k$.

\item From Corollary~\ref{cor:debtf}, the fixed points of the mappings of $b_r(d)$ are $d^\ast(r)$ when the upper bound constraint does not bind, and the consumption share is at the first best at these fixed points. Hence, $m((s, d^\ast(s)), (s, d^\ast(s))) = \delta$. By Lemma~\ref{lem:g}, the consumption share of the old decreases with~$r$. Hence, $m((1, d^\ast(1)), (r, b_r(d^\ast(1)))) \geq \delta$ with a strict inequality for some~$r$. Taking expectations, the bond price $p^1(1, d^\ast(1)) > \delta$ and consequently, the yield $y^1(1, d^\ast(1))<-\log(\delta)$. Since $y^\infty=-\log(\delta)$, $y^\infty-y^1(1,d^\ast(1))>0$. Likewise, it can be checked that $m((I, d^\ast(I)), (r, b_r(d^\ast(I)))) \leq \delta$ with a strict inequality for some~$r$, which shows that $y^\infty-y^1(I,d^\ast(I))<0$.

\item By definition $\Upsilon=\log(\psi_{\max}/\psi_{\min})$ where $\psi_{\max}$ and $\psi_{\min}$ are the maximum and minimum values of the eigenvector of the matrix $Q$. Using~\eqref{rmu1} and~\eqref{eqn:update} and assuming the nonnegativity and upper bound constraints do not bind, $m_B(x,x^\prime) = \nu(x^\prime)/\nu(x)$. Since $m_A(x, x^\prime)=\delta f(x)/f(x^\prime)$, the eigenvector $\psi(x)=f(x)/\nu(x)$. Since $f(x^\prime)=\delta\nu(x)/(\beta\nu(x)+\delta\nu(x^\prime))$, it follows that $\psi(x^\prime) = \delta/(\beta\nu(x)+\delta\nu(x^\prime))$. The maximum value of $\psi(x^\prime)$ occurs when $\nu(x)=\nu(x^\prime)=1$, in which case $\psi_{\max}=\delta/(\beta+\delta)$. The minimum value occurs when $\nu(x)=\nu(x^\prime)=\nu_{\max}$, in which case $\psi_{\min}=\delta/((\beta+\delta)\nu_{\max})$. Hence, $\Upsilon=\log(\psi_{\max}/\psi_{\min})=\log(\nu_{\max})$. It is easily checked that $\nu(s,d)$ is increasing in~$d$ with $\nu(s, d^{0}(s))$ increasing in~$s$, so that for $(s,d)\in E$, $\nu_{\max}=\nu(I, d^\ast(I))$. \qedhere
\end{enumerate}
\end{proof}

\begin{proof}[Proof of Proposition~\ref{prop:mrp}]
Since $R(x, x^\prime)= rb_r(d)/(s\mbox{BR}(d))$, the expected return $\bar{R}(x)= \sum_r\pi(r)rb_r(d)/(s\mbox{BR}(d))$. The risk-free rate is $R^f(x) = (\sum_r q(x,x^\prime))^{-1}$ where $q(x, x^\prime)=\pi(r) \beta s(1-d)/(1-r(1-b_r(d))$. Therefore, $\bar{R}(x)/R^f(x)$ is independent of~$s$. Since the risk-adjusted return on any asset is equal to the risk-free return, it follows that $\mbox{MRP}(d)=-\cov(m(x,x^\prime), R(x,x^\prime))$ where $m(x,x^\prime)=q(x,x^\prime)/\pi(r)$. From Corollary~\ref{cor:debtf}, $b_r(d)$ is increasing in~$r$ and hence, $R(x,x^\prime)$ is increasing with~$r$. From Lemma~\ref{lem:g}, old consumption ($1-r(1-b_r(d))$) falls with~$r$ and hence, $m(x,x^\prime)$ is increasing with~$r$. By Assumption~\ref{ass:nofb}, risk sharing is incomplete, and hence, the covariance term is positive and $\mbox{MRP}(d)<0$. That is, $\bar{R}(x)/R^f(x) <1$. With growth shocks, $R_+(x,x^\prime)=R(x, x^\prime)\gamma$ and $q_+(x,x^\prime)=\varsigma(\gamma)q(x,x^\prime)/\gamma$. Hence, $\bar{R}_+(x)=\bar{R}(x)(\mathbb{E}_\gamma \gamma)$, $R^f_+(x)=R^f(x)\bar{\gamma}$, and
\begin{gather*}\label{eqn:mrp}
\mbox{MRP}_+(d) = \frac{\bar{R}_+(x)-R^f_+(x)}{R^f_+(x)}= \left(\frac{\bar{R}(x)}{R^f(x)}-1\right)+\left(\frac{\bar{R}(x)}{R^f(x)}\right)\left(\frac{\mathbb{E}_\gamma \gamma}{\bar{\gamma}}-1\right).
\end{gather*}
Let $R^\ast_+(x, x^\prime)$ denote the returns with complete insurance. It is easy to check that
\begin{gather*}
R^\ast_+(x,x^\prime) =
\frac{\left(r - \frac{\delta}{\beta+\delta}\right)\gamma}{\delta\left(\sum_r \pi_r r - \frac{\delta}{\beta+\delta}\right)}.
\end{gather*}
The corresponding expected return is $\bar{R}^\ast_+(x)= (\mathbb{E}_\gamma \gamma)/\delta$. Likewise, the state price is $q^\ast_+(x,x^\prime)=\delta\varsigma(\gamma)\pi(r)/\gamma$, so that the risk-free return is $R^{f\ast}_+=\bar{\gamma}/\delta$.
Hence, the corresponding multiplicative risk premium is $\mbox{MRP}^\ast= (\mathbb{E}_\gamma \gamma-\bar{\gamma})/\bar{\gamma}$.
Since the arithmetic mean is larger than the harmonic mean, $\mbox{MRP}^\ast>0$. Substituting into the equation above gives
$\mbox{MRP}_{+}(d) = \mbox{MRP}(d) + \alpha(d) \mbox{MRP}^\ast$, where $\alpha(d) = \bar{R}(x)/R^f(x)$, as required.
\end{proof}

\begin{proof}[Proof of Proposition \ref{prop:canonic}]\NoBreakPar
\begin{enumerate}[itemsep=1pt,align=left,leftmargin=0pt,labelsep=0pt,itemindent=2\parindent,labelwidth=2\parindent,listparindent=\parindent,label=(\roman*)]
\item Since the probability of endowment state~1 is $\pi$ and debt is reset to the regeneration levels $d^0(s)$ after endowment state~1 has occurred, the probability that the state $(s,d^0(s))$ occurs is $\phi(s, d^0(s))=\pi(s)\pi$, irrespective of the date or history. Therefore, $T$ periods after such a resetting, the distribution function is:
\begin{gather*}
\phi_{T}(s,{d}^{(n)}(s))=\phi(s,{d}^0(s))(1-\pi)^{n}\quad
\text{for}\quad n=0,1,2,\ldots,T-1,
\end{gather*}
with $\phi_{T}(s,{d}^{(T)}(s))=\pi(s)(1-\pi)^T$. Taking the limit as $T\to\infty$ gives the invariant distribution $\phi$ described in the text.

\item By Proposition~\ref{prop:Ross}, $\Upsilon=\log(\nu_{\max})$. The value of $\nu_{\max}$ can be found from the fixed point of the mapping $b_2(d)$, which occurs at $d=d^\ast(2)$. From the first-order condition~\eqref{rmu1}, $\log(\nu_{\max})=\log(\delta/\beta)+\log((s(1)(1-b_1(d^\ast(2))))^{-1}-1)$. Since the participation constraint is binding when $d=d^\ast(2)$ and $b_2(d^\ast(2))=d^\ast(2)$, $b_1(d^\ast(2))$ can be found by solving:
\begin{gather*}\label{eq:PCbind}
\log (1-d^{\ast }(2))+\beta \left( \pi \log \left( 1-s(1)+s(1)b_1(d^{\ast
}(2))\right) +(1-\pi) \log \left(1-s(2)+s(2)d^{\ast }(2))\right)\right) \\
\qquad =\beta \left(\pi \log \left(1-s(1)\right) +(1-\pi)\log \left(1-s(2)\right)\right).
\end{gather*}
Since $s(1)=\kappa-\epsilon(1-\pi)/\pi$ and $s(2)=\kappa+\epsilon$, setting $\chi=1-s(1)(1-b_1(d^\ast(2)))$ and using $d^\ast(2)=1-\delta/(s(2)(\beta+\delta))$ gives the result in the text. \qedhere
\end{enumerate}
\end{proof}

\section*{Data Availability and Replication Files}

The code for replicating the figures in this article and the Supplemental Material, together with information about the Luxembourg Income Study Database, can be found in \citet{LRW24} in the Harvard Dataverse, \url{doi.org/10.7910/DVN/XDBGVY}.

\eject


\ifx\undefined\bysame
\newcommand{\bysame}{\hskip.3em \leavevmode\rule[.5ex]{3em}{.3pt}\hskip0.5em}
\fi

\clearpage
\newgeometry{margin=1in,top=1.35in,bottom=1.2in,headheight=0.85in, heightrounded}
\raggedbottom
\pagenumbering{arabic}
\renewcommand*{\thepage}{S\arabic{page}}
\urlstyle{tt}
\renewcommand{\thefootnote}{\fnsymbol{footnote}}

\section*{Supplemental Material\footnote[1]{Replication files are available at \url{http://dx.doi.org/10.7910/DVN/XDBGVY} \citep{LRW24sm}.}}\label{app:B}

\renewcommand{\thefootnote}{\arabic{footnote}}
\setcounter{footnote}{0} 

This supplemental material consists of six parts. Part~\ref{app:SF} provides evidence on the relative income of the young and the old, referred to in footnote~\ref{fn:LIS} in the Introduction, for six \ac{OECD} countries. Part~\ref{app:SA} provides proofs of Propositions~\ref{prop:st} and~\ref{prop:determ} from Sections~\ref{sec:model} and~\ref{sec:bench} together with the proofs of Lemma~\ref{lemma:V_omega} from Section~\ref{sec:opt} and Corollary~\ref{cor:debtf} from Section~\ref{sec:debt}. Parts~\ref{app:SE}-\ref{app:SC} relate to the two-state example of Section~\ref{sec:example}. Part~\ref{app:SE} examines two alternative welfare measures, the insurance coefficient, and the consumption-equivalent welfare change, and compares some of their comparative static properties with those of the Martin-Ross measure. Part~\ref{app:SDemo} considers the impact of a one-period unexpected increase in population size and shows how the effect of this demographic shock is both amplified and prolonged compared to the complete insurance benchmark. Part~\ref{app:SC} presents the shooting algorithm that can be used to solve the two-state example under Assumption~\ref{ass:canonic}. Part~\ref{app:SD} describes the pseudocode for the numerical algorithms used in the paper.

\appendix
\counterwithin{figure}{section}
\counterwithin{equation}{section}


\section{Change in Relative Income of Young and Old}\label{app:SF}

Significant changes have occurred in the relative average disposal income of the young and the old over recent decades for several \ac{OECD} countries. In this part of the supplemental material, we document these changes for Germany, Italy, Norway, Spain, the United Kingdom and the United States using data from the \citet{LIS24}. The Luxembourg Income Study Database provides harmonized microdata collected from more than 50 countries. Data is grouped into eleven waves spanning 1976-2020. There is also pre-wave data available for some countries for the period up to 1975. Data is available on individual and household incomes. 

The six countries we consider have at least one dataset per wave. Apart from Norway and Spain, the other countries have data from the pre-1975 period. Some counties have annual datasets in each wave, and others have just one dataset per wave. To examine the relative income of the young generation, we compute the average net (of taxes and transfers) equivalized disposable income of households where the head of the household is aged 25-34 relative to the average disposal income for all age groups. Likewise, the relative income of the old generation is measured by the ratio of the average equivalized disposable income of households where the head of the household is aged 65-74 relative to the average disposable income for all age groups. With more than one dataset per wave, we create country-wave data by averaging the yearly averages. 

\begin{figure}[tb]
\begin{center}
\includegraphics{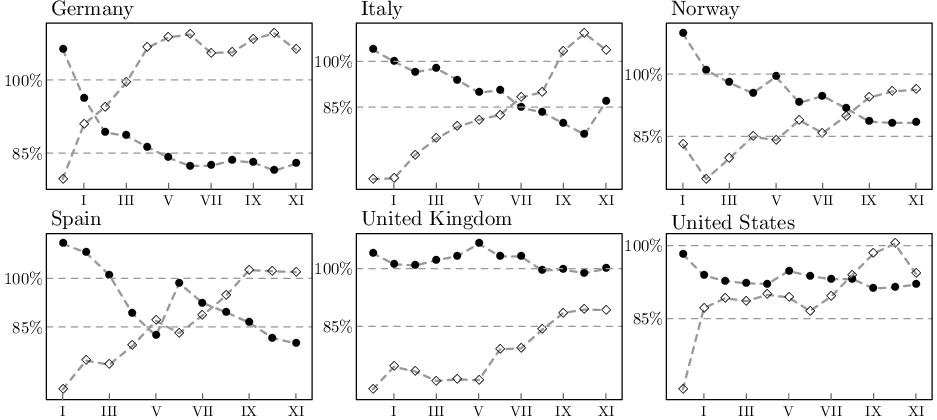}
\end{center}
\caption{Relative Income of Young and Old for six \ac{OECD} Countries across Waves of \ac{LIS} Data.}
\label{fig:LIS}
\begin{figurenotes}The relative income of the young is indicated by the $--\bullet --$ line and the $--\diamond --$ line is the relative income of the old. Income is the net average equivalized disposable income. The young correspond to households where the head of the household is aged 25-34 and the old correspond to households where the head of the household is aged 65-74. Income measures are relative to the equivalent average for all age groups. The 100\% and 85\% values are indicated for the purpose of comparison. The waves \Romannum{1}-\Romannum{11} correspond to the date ranges 1976-80, 1981-85, 1986-90, 1991-95, 1996-2000, 2001-04, 2005-07, 2008-10, 2011-13, 2014-16 and 2017-20. The historical data from pre-1975 is plotted as the first point for Germany, Italy, the United Kingdom and the United States. 
\end{figurenotes}
\end{figure}

Figure~\ref{fig:LIS} plots the relative income measures for the young and old generations for each of the six countries across the data waves~\Romannum{1}-\Romannum{11}. The pre-wave historical data is also plotted, where it is available. A value above 100\% indicates that the average of the income group (young or old) is above the average for the whole population. For example, in the United Kingdom, the young have a higher average than the average for the whole population, whereas in Germany, the old have been relatively better off for most of the period. In each country, there has been an improvement in
the average disposable income of the old compared to the average disposable income of the young over
the sample period. For example, the average disposable income
of the young in the United States has fallen from just below the national average at the start of the period to
just above 90\% of the national average at the end. Over the same period, the old have fared much better, with their average disposable income rising from approximately 70\% of the national average to just under 95\% of the national average. Moreover, the old overtook the young for the first time around the financial crisis of 2008. A similar pattern can be seen in Italy, Norway and Spain. A narrowing of the gap between the young and the old can also be observed in the United Kingdom. Germany is somewhat different, with the old overtaking the young as early as the 1980s. 

\section{Omitted Proofs}\label{app:SA}

\begin{proof}[Proof of Proposition~\ref{prop:st}] Here, we consider the case without growth shocks, so that $\bar{\gamma}=1$. The lifetime endowment utility of an agent born in a state with endowment share~$s$ is:%
\begin{gather*}
\hat{\upsilon}(s)\colonequals \log(s)+\beta\sum\nolimits_{r}\pi(r)\log(1-r)\text{.}%
\end{gather*}
Consider a small transfer $\mathrm{d}\tau(s)$ in state~$s$ from the young
to the old. The problem of existence of a sustainable allocation can be
answered by finding a vector of positive transfer $\mathrm{d}\tau$ such that there is a weak improvement over the lifetime endowment utility in all states and a strict improvement in at least one state. The change in the lifetime endowment utility induced by a vector $\mathrm{d}\tau$ is
nonnegative if%
\begin{gather}
-s^{-1}\mathrm{d}\tau(s)+\beta\sum\nolimits_{r}\pi(r)(1-r)^{-1}\mathrm{d}\tau(r)\geq0\text{.} \label{pc}%
\end{gather}
Rearranging~\eqref{pc} in terms of the marginal rates of substitution $\hat{m}(s,r)$, we have:%
\begin{gather*}
-\mathrm{d}\tau(s)+\sum\nolimits_{r}\pi(r)\hat{m}(s,r)\mathrm{d}\tau(r)\geq0\text{.}%
\end{gather*}
The
problem of existence can then be addressed by finding a vector $\mathrm{d}\tau>0$ that
solves:%
\begin{gather}
\left(  \hat{Q}-I\right)  \mathrm{d}\tau\geq0 \label{M},%
\end{gather}
where $I$ is the identity matrix and $\hat{Q}$ is the matrix of $\hat{q}(s,r)=\pi(r)\hat{m}(s,r)$. Equation~\eqref{M} has a well-known solution. Using the
Perron-Frobenius theorem, there exists a strictly positive solution for
$\mathrm{d}\tau$, provided that the Perron root, that is, the largest eigenvalue of
$\hat{Q}$, is greater than one. This is satisfied by Assumption~\ref{ass:sustain}, which guarantees the existence of positive transfers from the young to the old that improve the utility of each generation.
\end{proof}

\begin{proof}[Proof of Proposition~\ref{prop:determ}]
Let $\omega^c_0 = \omega^\ast$ and define the critical utility $\omega^c_1$ by:
\begin{gather*}
\log(1-\exp(\omega_1^c)) + \beta \omega^{\ast} = \hat{\upsilon}\colonequals \log(s)+\beta \log(1-s).
\end{gather*}
Define $\omega^c_n$ recursively by:
\begin{gather*}
\log(1-\exp(\omega_n^c))) + \beta \omega^c_{n-1} = \hat{\upsilon} \quad \mbox{for $n=2,3,\ldots,\infty$.}
\end{gather*}
From the strict concavity of the logarithmic utility function, $\omega^c_n>\omega^c_{n-1}$ and $\lim_{n\to\infty}\omega^c_n=\omega_{\max}=\log(s)$. Let $v^\ast =\log(1-\exp(\omega^\ast))+({\beta}/{\delta})\omega^\ast$. With some abuse of notation, let $V_n(\omega)$ denote the value function when $\omega\in(\omega^c_{n-1}, \omega^c_n]$. Hence,
\begin{gather*}
V_n(\omega) = \log(1-\exp(\omega)) + \frac{\beta}{\delta}\omega + \delta V_{n-1}\left(\frac{1}{\beta}\left(\hat{\upsilon}-\log(1-\exp(\omega))\right)\right).
\end{gather*}
For $\omega\leq\omega^\ast$, $\omega^\prime=\omega^\ast$. Therefore, $V(\omega)=v^\ast/(1-\delta)$ for $\omega\in[\omega_{\min}, \omega^\ast]$. For $\omega\in(\omega^\ast, \omega^c_1]$,
\begin{gather*}
V_1(\omega) = \log(1-\exp(\omega)) + \frac{\beta}{\delta}\omega + \frac{\delta}{1-\delta}v^\ast.
\end{gather*}
Differentiating the function $V_1(\omega)$ gives:
\begin{gather*}
\frac{\mathrm{d} V_1(\omega)}{\mathrm{d} \omega} = \frac{\beta}{\delta} - \frac{\exp(\omega)}{1-\exp(\omega)}.
\end{gather*}
Let $h(\omega) \colonequals \exp(\omega)/(1-\exp(\omega))$. Since $\omega>\omega^\ast$, $h(\omega) > {\beta}/{\delta}$ and ${\mathrm{d} V_1(\omega)}/{\mathrm{d} \omega}<0$. Note that $h(\omega^\ast)={\beta}/{\delta}$ and therefore, in the limit as $\omega\to\omega^\ast$, ${\mathrm{d} V_1(\omega)}/{\mathrm{d} \omega}=0$. Furthermore,
the function $V_1(\omega)$ is strictly concave because $h(\omega)$ is increasing in~$\omega$. Using this result, we can proceed by induction and assume $V_{n-1}(\omega)$ is
decreasing and strictly concave. Then, it is straightforward to establish that $V_n(\omega)$ is decreasing and strictly concave. Continuity follows since $\lim_{\omega\to \omega^{c}_{n}}V_{n+1}(\omega)=V_{n}(\omega^c_n)$. To establish
differentiability, we need to demonstrate that:
\begin{gather*}
\lim_{\omega\to \omega^{c}_{n}}\frac{\mathrm{d} V_{n+1}(\omega)}{\mathrm{d} \omega}=\frac{\mathrm{d} V_n(\omega^c_n)}{\mathrm{d} \omega}.
\end{gather*}
To show this, note that for $\omega \in (\omega^c_{n}, \omega^c_{n+1})$:
\begin{gather*}
\frac{\mathrm{d} V_{n+1}(\omega)}{\mathrm{d} \omega}  =  \frac{\beta}{\delta} - h(\omega)\left(1-\frac{\delta}{\beta}\frac{\mathrm{d} V_n(\omega^\prime)}{\mathrm{d} \omega}\right).
\end{gather*}
Starting with $n=1$, we have:
\begin{gather*}
\lim_{\omega\to \omega^{c}_{1}}\frac{\mathrm{d} V_{2}(\omega)}{\mathrm{d} \omega} =  \frac{\beta}{\delta} - h(\omega^c_1)\left(1-\frac{\delta}{\beta}\lim_{\omega\to \omega^{c}_{0}}\frac{\mathrm{d} V_{1}(\omega)}{\mathrm{d} \omega}\right).
\end{gather*}
Since $\lim_{\omega\to \omega^{c}_{0}}{\mathrm{d} V_{1}(\omega)}/{\mathrm{d} \omega}=0$, we have:
\begin{gather*}
\lim_{\omega\to \omega^{c}_{1}}\frac{\mathrm{d} V_{2}(\omega)}{\mathrm{d} \omega} = \frac{\beta}{\delta}-h(\omega^c_1)=\frac{\mathrm{d} V_{1}(\omega^c_1)}{\mathrm{d} \omega}.
\end{gather*}
Therefore, make the recursive assumption that $\lim_{\omega\to \omega^{c}_{n-1}}{\mathrm{d} V_{n}(\omega)}/{\mathrm{d} \omega} =
{\mathrm{d} V_{n-1}(\omega^c_{n-1})}/{\mathrm{d} \omega}$. In general, we have:
\begin{align*}
\lim_{\omega\to \omega^{c}_{n}}\frac{\mathrm{d} V_{n+1}(\omega)}{\mathrm{d} \omega} &=  \frac{\beta}{\delta} - h(\omega^c_n)\left(1-\frac{\delta}{\beta}
\lim_{\omega\to \omega^{c}_{n-1}}\frac{\mathrm{d} V_{n}(\omega)}{\mathrm{d} \omega}\right)\\
\frac{\mathrm{d} V_{n}(\omega^c_n)}{\mathrm{d} \omega} &= \frac{\beta}{\delta} - h(\omega^c_n)\left(1-\frac{\delta}{\beta} \frac{\mathrm{d} V_{n-1}(\omega^c_{n-1})}{\mathrm{d} \omega}\right).
\end{align*}
By the recursive assumption, these two equations are equal. Hence, we conclude that $V(\omega)$ is differentiable. In particular, repeated substitution gives:
\begin{gather*}
\frac{\mathrm{d} V_{n}(\omega^c_n)}{\mathrm{d} \omega}= \frac{\beta}{\delta}-\left(\frac{\delta}{\beta}\right)^{n-1}\prod\nolimits_{j=1}^{n}h(\omega^c_j).
\end{gather*}
Since $h(\omega^c_j)\in[\beta/\delta, s/(1-s))$, taking the limit as $n\to\infty$, or equivalently, $\omega\to{\omega_{\max}}$, gives $\lim_{\omega\to{\omega}_{\max}}{\mathrm{d} V(\omega)}/{\mathrm{d} \omega} = - \infty$. Equation~\eqref{eq:detpol} follows from above given that $\omega^\prime=\omega^\ast$ or satisfies $\log(1-\exp(\omega))+\beta \omega^\prime=\hat{\upsilon}$ if $\log(1-\exp(\omega))+\beta \omega^\ast<\hat{\upsilon}$.
\end{proof}

\begin{proof}[Proof of Lemma~\ref{lemma:V_omega}]\NoBreakPar
\begin{enumerate}[itemsep=1pt,align=left,leftmargin=0pt,labelsep=0pt,itemindent=2\parindent,labelwidth=2\parindent,listparindent=\parindent,label=(\roman*)]
\item Given the participation constraint of the old, $\omega\geq\omega_{\min}(s)=\log(1-s)$. The largest feasible $\omega$, $\omega_{\max}$, can be found by solving the set of equations $\log(1-\exp(\omega_{\max}(s))) +\beta \sum_r\pi(r)\omega_{\max}(r)=\hat{\upsilon}(s)$. There is a trivial solution where $\omega_{\max}(s)=\omega_{\min}(s)$ but by Proposition~\ref{prop:st} there is also a nontrivial solution with $\omega_{\max}(s)>\omega_{\min}(s)$ for each~$s$. Since the utility function is concave, this nontrivial solution is unique. Let $\Delta\varpi\colonequals \sum_r\pi(r)(\omega_{\max}(r)-\omega_{\min}(r))$. Then, $\omega_{\max}(s)$ can be found by solving the set of equations $\log(1-\exp(\omega_{\max}(s))) - \log(1-\exp(\omega_{\min}(s)))+\beta \Delta\varpi=0$.  Since $\Delta\varpi>0$ by Proposition~\ref{prop:st}, it follows that $\omega_{\max}(s)>\omega_{\min}(s)$ for all~$s$. A reduction in~$\omega$, enlarges the constraint set $\Phi(s,\omega)$ of Problem~\ref{eqn:p1} and hence, $V(s,\omega)$ is nonincreasing in~$\omega$. To show that $V(s,\omega)$ is concave in~$\omega$, consider the mapping $T$ where%
\begin{gather*}
(TJ)(s,\omega)=  \max_{\{c, \left(\omega_{r}\right)_{r\in \mathcal{I}}\}\in\Phi(s,\omega)} \enspace \frac{\beta}{\delta} \log(1-c) + \log(c) + \delta \sum\nolimits_{r} \pi(r) J(r,\omega_{r})  \text{.}%
\end{gather*}
Consider $J=V^{\ast}$, the first-best frontier. Proposition~\ref{prop:fbstm} established that $V^\ast(s,\omega)$ is concave in~$\omega$. It follows from the definitions of $T$ and $V^\ast$ that $TV^{\ast}(s,\omega)\leq V^{\ast}(s,\omega)$ because
$V^{\ast}(s,\omega)\leq v^{\ast}(s) + \delta \bar{V}^{\ast}$ and the mapping $T$ adds the participation constraints~\eqref{eqn:icr}. That is, $T^{n}V^{\ast}(s,\omega)\leq T^{n-1}V^{\ast}(s,\omega)$ for $n=1$. Now, make the induction
hypothesis that $T^{n}V^{\ast}(s,\omega)\leq T^{n-1}V^{\ast}(s,\omega)$ for
$n\geq2$ and apply the mapping $T$ to the two functions $T^{n}V^{\ast}(s,\omega)$ and $T^{n-1}V^{\ast
}(s,\omega)$. It is straightforward to show that $T^{n+1}V^{\ast}(s,\omega)\leq T^{n}V^{\ast}(s,\omega)$, because the constraint set is the same in both cases but, by the induction hypothesis, the objective is no
greater in the former case. Hence, the sequence
$T^{n}V^{\ast}(s,\omega)$ is nonincreasing and converges. Let $V^{\infty}(s,\omega)\colonequals\lim_{n\to\infty}T^nV^{\ast}(s,\omega)$ be the pointwise limit of the mapping $T$. We have that $V^{\infty}$ and $V$ are both fixed points of $T$. Since the mapping is monotonic, $T^n(V^\ast)\geq T^n(V)=V$. Hence, $V^\infty\geq V$ but, since $V$ is the maximum, we have that $V^\infty=V$. Starting from $V^\ast$, the objective function in the mapping $T$ is concave because $V^\ast$ and the utility function are concave. The constraint set $\Phi(s,\omega)$ is convex. Hence, $TV^\ast(s,\omega)$ is concave. By induction, $T^{n}V^{\ast}(s,\omega)$ and the limit function $V$ are also concave. Differentiability follows because the linear independence constraint qualification is satisfied on the interior of the domain. There are $I+1$ choice variables and $2I+3$ constraints. Since $V$ is concave, it is differentiable if the multipliers associated with the constraints are unique. The multipliers are unique if the linear independence constraint qualification is satisfied, that is, if the gradients of the binding constraints are linearly independent at the solution. Since $\omega_{\max}(s)>\omega_{\min}(s)$, the corresponding upper and lower bound constraints in~\eqref{eqn:ub} and~\eqref{eqn:lb} are not both active. Not all lower bound constraints in~\eqref{eqn:lb} are active because this would involve no transfers next period, which we know from Proposition~\ref{prop:st} is not optimal. For $\omega<\omega_{\max}(s)$, not all upper bound constraints in~\eqref{eqn:ub} are active because this would imply that the participation constraint of the young~\eqref{eqn:icr} does not bind, in which case some $\omega_r$ can be lowered below $\omega_{\max}(r)$ to raise the planner's payoff. If $\omega>\omega_{\min}(s)$, then the nonnegativity constraint~\eqref{eqn:nn} is not active. If $\omega=\omega_{\min}(s)$ and constraint~\eqref{eqn:nn} binds, then the young participation constraint is not active. Hence, in either case, there are at most $I+1$ active constraints. Since the marginal utility is strictly increasing,  $\beta>0$, and $\pi(s)>0$ for each~$s$, it can be checked that the matrix of active constraints has full rank. Hence, the multipliers are unique and $V(s,\omega)$ is differentiable in~$\omega$ on $(\omega_{\min}(s), \omega_{\max}(s))$. Since $V(s,\omega)$ is concave and differentiable in~$\omega$, it is also continuously differentiable in~$\omega$.

\item If constraint~\eqref{eqn:pk} binds, then the frontier is strictly downward sloping. Strict concavity of $V$ when $V$ is strictly downward sloping follows since $TV$ is strictly concave when the frontier is strictly downward sloping because of the strict concavity of the logarithmic utility function and the concavity of $V$. If $\omega^{0}(s)>\omega^{\ast}(s)$, then it would be possible to lower $\omega^{0}(s)$, increase consumption of the current young keeping all future promises the same without violating any constraints and increase the planner's utility. Assumption~\ref{ass:nofb} guarantees that $\omega^{0}(s)<\omega^{\ast}(s)$ for at least one state~$s$. If $\omega_{\min}(s)=\omega^{\ast}(s)$, then $\omega^{0}(s)=\omega^{\ast}(s)$ and hence, constraint~\eqref{eqn:icr} does not bind. Therefore, $\mu(s,\omega_{\min}(s))=0$. In this case, from equation~\eqref{fp1}, either $V_\omega(r, g_r(s,\omega))=0$ for $\omega=\omega_{\min}(s)$ or one of the constraints~\eqref{eqn:ub} or~\eqref{eqn:lb} binds and $g_r(s, \omega)$ is independent of~$\omega$. Therefore, in either case, $V_{\omega}(s,{\omega}^{0}(s))=({\beta}/{\delta}) -(\exp(\omega_{\min})/(1-\exp(\omega_{\min})))=({\beta}/{\delta}) -((1-s)/s)\leq 0$. If $\omega_{\min}(s)<\omega^{\ast}(s)$, then it cannot be that $\omega^{0}(s)=\omega_{\min}(s)$ because this implies $c(s,\omega^{0}(s))=s$, which, in turn, implies $\omega^{0}(s)\geq \omega^{\ast}(s)$ from equation~\eqref{eqn:pk}, a contradiction. The multiplier $\lambda(s,\omega)\geq0$ and since $V(s,\omega)$ is concave in~$\omega$, $\lambda(s,\omega)$ is increasing in~$\omega$. Let $\lambda_{\max}(s)\colonequals\lim_{\omega\to\omega_{\max}}\lambda(s,\omega)$, then   $\lim_{\omega\to{\omega}_{\max}}V_{\omega}(s,\omega) = -({\beta}/{\delta})\lambda_{\max}(s)$, where $\lambda_{\max}(s)\in\mathbb{R}_+ \cup \{\infty\}$.
\item Since $\omega_{\min}(s)=\log(1-s)$, it follows that  $\omega_{\min}(s(i))<\omega_{\min}(s(i\shortminus 1))$ for $i=2,\ldots,I$. It follows from the proof of Part~(i) that $\log(1-\exp(\omega_{\max}(s))) - \log(1-\exp(\omega_{\min}(s)))$ is independent of~$s$. Therefore, $\omega_{\max}(s(i\shortminus 1))>\omega_{\max}(s(i))$. Equally, $\Delta\varpi$ is bounded above since the endowment share is less than one in each state. Hence, $\omega_{\max}(s)<0$, otherwise the constraint of the young cannot be satisfied. Since $\omega^{0}(s)=\log(1-c^{0}(s))$, it follows from Lemma~\ref{lem:c}(ii) that $\omega^{0}(s(i))\leq\omega^{0}(s(i\shortminus 1))$  with strict inequality for at least one~$i$, and hence, $\omega^{0}(I)<\omega^{0}(1)$. \qedhere
\end{enumerate}
\end{proof}

\begin{proof}[Proof of Corollary~\ref{cor:debtf}]\NoBreakPar
\begin{enumerate}[itemsep=1pt,align=left,leftmargin=0pt,labelsep=0pt,itemindent=2\parindent,labelwidth=2\parindent,listparindent=\parindent,label=(\roman*)]
\item Since consumption $c=s(1-d)$ for a fixed debt $d$, the participation constraint~\eqref{eqn:icr} can be rewritten as $\log(1-d)+ \beta \sum\nolimits_{r} \pi(r) (\omega_{r}-\log(1-r)) \geq 0$. Hence, the multiplier on this constraint depends only on~$d$. If the constraint does not bind, then $g_r(s, \omega)=\omega^0(r)$ from Lemma~\ref{lem:g} and consequently, $b_r(d)=d(r, \omega^0(r))=d^0(r)$. There is a critical value of debt, $d^c$, below which the constraint does not bind. Since $b_r(d)=d^0(r)$ for $d\leq d^c$, it follows that $
d^c = 1-\exp(-\beta\sum_{r}\pi(r)(\log(1-r+r d^{0}(r)) - \log(1-r))) \in(d_{\min}, d_{\max})
$.
By Assumption~\ref{ass:LS}, $d^0(1)=0$ and hence, $d_{\min}=0$. Since $d_{\max}$ is the largest nontrivial solution of $\log(1-d_{\max}) +\beta \sum_{r}\pi(r)(\log(1-r+r d_{\max})-\log(1-r))=0$, it follows that $d_{\max}<1$. Let $\tilde{h}_s(d)\colonequals-(\delta/\beta)V_{\omega}(s, \log(1-s+sd))$. Then, for $d\in[d^c, d_{\max})$, $b_r(d)=\tilde{h}_r^{-1}(\mu(d))$. Since $\mu(d)$ is strictly increasing in~$d$ and $\tilde{h}_r^{-1}(\mu)$ is strictly increasing in~$\mu$, it follows that $b_r(d)$ is strictly increasing in~$d$ for $d>d^c$.    
\item Consider the value function
\begin{gather*}
\hat{V}(s,d)=\max_{(b_r)_{r\in \mathcal{I}}} (\beta/\delta)\log(1-s+sd)+\log(s(1-d)) + \delta \sum\nolimits_{r}\pi(r)\hat{V}(r,b_r)
\end{gather*}
subject to the constraint
\begin{gather*}
\log(1-d) + \beta\sum\nolimits_{r}\pi(r)\left(\log(1-r+rb_r)-\log(1-r)\right) \geq 0.
\end{gather*}
It follows from this maximization that $\hat{V}_d(s^\prime, d)>\hat{V}_d(s,d)$ for $s^\prime>s$. We want to show $b_{r^{\prime}}(d)\geq b_r(d)$ for $r^\prime>r$. Suppose to the contrary that $b_{r^{\prime}}(d)< b_r(d)$. It follows from the first-order condition that $\hat{V}_d(r^\prime, b_{r^{\prime}}(d)) \leq \hat{V}_d(r, b_r(d))$ with equality only if the multiplier $\mu(d)$ is not binding. But since $\hat{V}_d(r, b_{r^{\prime}}(d)) < \hat{V}_d(r^\prime, b_{r^{\prime}}(d))$ it follows that $\hat{V}_d(r, b_{r^{\prime}}(d))<\hat{V}_d(r, b_r(d))$, which from the concavity of $\hat{V}(r, d)$ in~$d$, implies $b_{r^{\prime}}(d)>b_r(d)$, a contradiction. 
\item Lemma~\ref{lem:g}(iv) and Assumption~\ref{ass:nofb} guarantee that $d^f(r(I))>d^c$. \qedhere
\end{enumerate}
\end{proof}

\section{Alternative Measures and Comparative Statics}\label{app:SE}

Recall that the Martin-Ross bound $\Upsilon$ is the logarithm of the ratio of the maximum and the minimum values of the leading eigenvector of the state-price matrix. The paper uses the Martin-Ross bound as a measure of risk because of its relationship to the variability in the yield curve. In this part of the supplemental material, we use the two-state case of Section~\ref{sec:example} to examine the comparative statics of the Martin-Ross measure and compare it to two alternative measures of risk sharing used in the literature. For the purposes of comparison with the alternative measures, we consider state prices as a function of the endowment state and the promised utility, that is, $x=(s,\omega)$, as in Sections~\ref{sec:opt} and~\ref{sec:conv} of the paper. The comparative statics considered are related to the properties of the endowment process and the discount factors. In this comparative static exercise, we change the value of the parameter of interest holding all other parameters at the values given in Example~\ref{exmpl:canon} ($\delta=\beta=\exp(-\sfrac{1}{75})$, $\pi=\sfrac{1}{2}$, $\kappa=\sfrac{3}{5}$, and $\epsilon=\sfrac{1}{10}$).\footnote{%
For the comparative statics with respect to $\kappa$ and $\epsilon$, the invariant distribution remains geometric as the parameter changes. However, the invariant distribution is not geometric when the discount factors are low. In this case, we can no longer  use the shooting algorithm described in Section~\ref{sec:example} and instead implement
an algorithm based on a value function iteration method (see, Part~\ref{app:SD} of this supplemental material for a description).
}

\begin{figure}[tb]
\begin{center}
\includegraphics{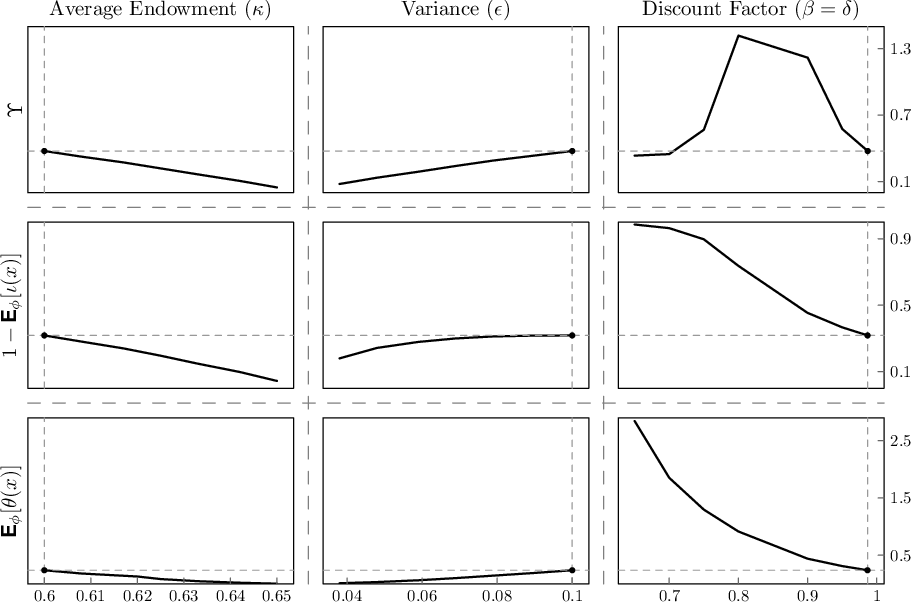}
\end{center}
\caption{Comparative Statics on the Martin-Ross Bound $\Upsilon$, the Expected Insurance Coefficient and the Expected Consumption Equivalent Welfare Change.}
\label{fig:comp_welfare}
\begin{figurenotes}
The top row illustrates the bound $\Upsilon$. The middle row illustrates one minus the average insurance coefficient $1-\mathbb{E}_{\phi}[\iota(x)]$. The bottom row illustrates the  average consumption equivalent welfare change $\mathbb{E}_{\phi}[\theta(x)]$. The columns correspond to changes in the average endowment share of the young $\kappa$, the variance of the endowment share of the young $\epsilon$ and the common discount factor $\beta=\delta$. The dot in each panel indicates the parameter values given in Example~\ref{exmpl:canon}. 
\end{figurenotes}
\end{figure}

We consider changes in the values of three parameters. The average endowment share of the young ($\kappa$), the variance of the endowment share of the young ($\epsilon$), and the discount factor $\delta$, holding $\beta=\delta$. The top row of Figure~\ref{fig:comp_welfare} plots the bound $\Upsilon$ against $\kappa$, $\epsilon$ and $\delta$. A larger $\kappa$ corresponds to a larger average endowment share of the young, while a smaller $\epsilon$ corresponds to reduced endowment risk. Increasing $\kappa$, or reducing $\epsilon$, raises risk sharing as measured by a reduction in $\Upsilon$. For $\kappa$ above a critical value, or $\epsilon$ below a critical value, the first best is sustainable at the invariant distribution, in which case $\Upsilon=0$.\footnote{%
The critical values are $\kappa\approx0.6565$ and $\sigma\approx0.0243$.} For the range of parameters considered, the measure $\Upsilon$ is relatively insensitive to changes in $\kappa$ and $\epsilon$. The effect of changes in the discount factor on $\Upsilon$ is nonmonotonic when we consider sufficiently low values of $\delta$ for which Assumption \ref{ass:canonic} does not hold. When the discount factor is high, the invariant distribution has geometric probabilities as described in Part~(i) of Proposition~\ref{prop:canonic}. As the discount factor falls, either the current transfer falls, or the future promise increases to satisfy the participation constraint of the young in state~2. This change reduces $\omega^{0}(2)$. With $\omega^{\ast}(2)$ fixed, it effectively enlarges the ergodic set, increasing risk, as reflected by the rise of $\Upsilon$. As the discount factor falls further, the upper bound constraint becomes binding and it is no longer true that the future promise is reset to $\omega^0(1)$ as soon as state~1 occurs. Reversion to $\omega^0(1)$ occurs less frequently, and the invariant distribution has a positive probability mass at $\omega_{\max}(2)$. Although $\omega_0(2)$ falls with the discount factor, the range ${\omega}_{\max}(2)-\omega_0(2)$ decreases, implying that the bound $\Upsilon$ falls.

Two alternative measures used in the literature to assess how risk is shared are the \textit{insurance coefficient\/} \citep[see, for example,][]{Kaplan-Violante10sm} and the \textit{consumption equivalent welfare change\/} \citep[see, for example,][]{Songetal15sm}. Both measures are conditional on the given state~$x=(s,\omega)$. The insurance coefficient $\iota(x)$ is the fraction of the variance of the endowment shock that does not translate into a corresponding change in consumption. With i.i.d.\ shocks,
\begin{gather*}
1-\iota(x)=\frac{\cov\left(\log\left(f(r,g_r(x))\right), \log\left(r\right)\right)}
{\var\left(\log\left(r\right)\right)},
\end{gather*}
where $r$ is the endowment shock next period, $g_r(x)$ is the next period promise and $f(r,g_r(x))$ is the next-period consumption. The consumption equivalent welfare change relative to the first best is measured by solving the following equation in terms of~$\theta$:
\begin{gather*}
\begin{split}
V(s,\omega)=& \log\left((1-\theta)c^\ast(s)\right)+\frac{\beta}{\delta}\log\left((1-\theta)(1-c^\ast(s))\right)
\\ 
 & \enspace +\frac{\delta}{1-\delta}\sum\nolimits_{r}\pi(r)\left(\log\left((1-\theta)c^\ast(r)\right)
+\frac{\beta}{\delta}\log\left((1-\theta)(1-c^\ast(r))\right)\right),
\end{split}
\end{gather*}
where $c^\ast(s)$ is the first-best consumption share in the stationary state (after the initial period). The solution $\theta(x)$ measures the proportion by which the first-best consumption share needs to be reduced to match the optimal solution for each~$x$. To compare the measures $1-\iota(x)$ and $\theta(x)$ with $\Upsilon$, we compute their expected value at the stationary state using the invariant distribution $\phi(x)$. 

The average insurance coefficient measure $1-\mathbb{E}_{\phi}[\iota(x)]$ is plotted in the middle row of Figure~\ref{fig:comp_welfare} against $\kappa$, $\epsilon$ and $\delta$. The consumption equivalent welfare change measure $\mathbb{E}_{\phi}[\theta(x)]$ is plotted in the bottom row of Figure~\ref{fig:comp_welfare} against $\kappa$, $\epsilon$ and $\delta$. The comparative static properties of these measures are similar to those of the bound $\Upsilon$ for changes in $\kappa$ and $\epsilon$. There is a difference for low discount factors. For lower values of the discount factor, the Martin-Ross measure $\Upsilon$ falls with the discount factor because, in this case, the range of feasible promises narrows. 

\section{Demographic Shock}\label{app:SDemo} 

In this part of the supplemental material, we consider a one-period unexpected increase in population size, an \ac{MIT} shock, and examine the impulse response function. We consider the two-state case of Section~\ref{sec:example} and use the parameter values of Example~\ref{exmpl:canon}. We show that compared to the complete insurance benchmark, the effect of this demographic shock on consumption is amplified and prolonged. 

Consider a one-off, unexpected but permanent increase in population size at date~$T+1$. Suppose the cohort size is $N_t=1$ for $t\leq T$ and $N_t=1+\varepsilon>1$ for $t> T$. Assume that the endowment is proportional to the population size, so the total endowment is $e_t=e^y_t N_t+e^o_t N_{t-1}$. Furthermore, assume there is no growth in the total endowment, and that the planner weighs the utility of the young and the old by adjusting for the cohort size. Normalize the endowment so that $e_t=1$ for $t\leq T$, $e_{T+1}=e^y_{T+1} (1+\varepsilon) + e^o_{T+1}=1+\varepsilon e^y_{T+1}$, and $e_t=1+\varepsilon$ for $t\geq T+2$. The consumption share of the young is $c_t=N_tC_t/e_t$ where $C_t$ is the corresponding consumption level. At $t\not=T+1$, $c_t=C_t$, while $c_{T+1}=(1+\varepsilon)C_{T+1}/(1+\varepsilon e^y_{T+1})$. Likewise, the endowment share of the young is $s_t=N_te^y_t/e_t$, where $s_t=e^y_t$ for $t\not=T+1$ and $s_{T+1}=(1+\varepsilon)e^y_{T+1}/(1+\varepsilon e^y_{T+1})$.

First, consider the complete insurance benchmark where the nonnegativity constraint on transfers does not bind. In this case, analogous to equation~\eqref{tau_fbm} of the paper, the consumption share of the young is $c^{\ast}_t = \delta N_{t}/(\beta N_{t-1}+\delta N_t)$. Hence, $c^{\ast}=\delta/(\beta+\delta)$ for all $t\not=T+1$ and, at the time of the shock, $c^{\ast}_{T+1}=\delta(1+\varepsilon)/(\beta +\delta(1+\varepsilon))>\delta/(\beta+\delta)$. We need to check that the promise-keeping constraint is satisfied. Since the shock was not anticipated, the old at~$T+1$ were promised a consumption level of $\beta/(\beta+\delta)$. At~$T+1$, after the shock is realized, the consumption level of the old is $(1-c^{\ast}_{T+1})e_{T+1}$. Since the nonnegativity constraint does not bind, $e^y_t\geq \delta/(\beta+\delta)$, and hence, $e_{T+1}=1+\varepsilon e^y_{T+1} \geq (\beta+\delta(1+\varepsilon))/(\beta+\delta)$. Therefore, $(1-c^{\ast}_{T+1})e_{T+1}\geq \beta/(\beta+\delta)$ and the promise-keeping constraint is satisfied. 

Next consider the outcome when the planner respects the participation constraints of each generation. Let the state at~$T$ be $x_T=(s_T, \omega_T)$, where $s_T$ is the endowment share of the young and $\omega_T$ is the inherited promise. The promise to the old at~$T+1$, determined at~$T$, is $g_{s_{T+1}}(x_T)$, conditional on the endowment share at~$T+1$. The planner's problem at all~$t\not=T+1$ is the same as without the shock, and the optimal policy functions are $g_r(x)$ and $f(x)$, as described in Lemmas~\ref{lem:g} and~\ref{lem:c} of the paper. The planner's problem at~$T+1$ is slightly different. In particular, the planner chooses the consumption share of the young, $c$, and the future promises, $\omega_r$, to maximize
\begin{gather}
(1+\varepsilon) \log(c) + \left(\frac{\beta}{\delta}\right)\log(1-c) +\delta (1+\varepsilon) \sum\nolimits_r\pi(r)V(r, \omega_r), \tag{PD} \label{eqn:pd}
\end{gather}
subject to the participation constraint of the young:
\begin{gather}
\log(c)+ \beta  \sum\nolimits_r\pi(r) \omega_r \geq \log(s)+  \beta\sum\nolimits_r\pi(r)\log(1-r), \label{eqn:PCD}
\end{gather}
and the promise-keeping constraint:
\begin{gather}
\log\left((1-c)e\right) \geq g_s(x), \label{eqn:PKD}
\end{gather}
where $e=1+\varepsilon e^y_{T+1}$ is the aggregate income, and $s=(1+\varepsilon) e^y_{T+1}/e$ is the endowment share of the young. Note that the continuation value for the planner $V(r, \omega_r)$ is the value function derived in Lemma~\ref{lemma:V_omega} of the paper because, from~$T+2$ onward, the optimization problem has a constant population. Denote the solution to the maximization problem at~$T+1$ by the policy functions $\tilde{f}(s, g_s(x))$ and $\tilde{g}_r(s, g_s(x))$. When $\varepsilon=0$, these two policy functions coincide with the policy functions $f(s,g_s(x))$ and $g_r(s, g_s(x))$. For $\varepsilon>0$, more resources are available, and the objective in the optimization problem~\ref{eqn:pd} puts more weight on both the consumption share of the young and future generations. With more resources available, the promise-keeping constraint~\eqref{eqn:PKD} is relaxed with each young agent able to retain a slightly higher consumption level while maintaining the promise to the old. Hence, the participation constraint~\eqref{eqn:PCD} is relaxed, enabling a reduction in the promise to the old next period. That is $\tilde{f}(s, g_s(x)) \geq f(s, g_s(x))$ and $\tilde{g}_r(s,g_s(x))\leq g_r(s,g_s(x))$ with strict inequalities in some states. 

To construct the impulse response function, consider all possible sample paths following the shock. A sample path from~$T$ to~$T+t$ is a sequence of endowment states from~$T$ to~$T+t$. Denote this path by $s^{t\mid T}=(s_T, s_{T+1}, \ldots, s_{T+t})$. Following the shock at~$T+1$, the consumption share of the young is constructed iteratively from the policy functions $f$, $\tilde{f}$, $g$, and $\tilde{g}$. In particular, $c_{T+t}=\tilde{f}^t(s^{t\mid T}, x_T)\colonequals f(s_{T+t}, \tilde{g}^t(s^{t\mid T}, x_T)$, where $\tilde{g}^t(s^{t\mid T}, x_T)\colonequals g_{s_{T+t}}(s_{T+t-1}, \tilde{g}^{t-1}(s^{t-1\mid T}, x_T)$, and at~$T+1$, $\tilde{f}^1(s^{1\mid T}, x_T)=\tilde{f}(s_{T+1}, g_{s_{T+1}}(x_T))$. The probability of the sample path $s^{t\mid T}$ is $\pi(s^{t\mid T})=\prod_{\tau=1}^{t}\pi(s_{T+\tau})$. Hence, the expected consumption share of the young at~$T+t$ conditional on~$x_T$ is $\bar{c}_{T\!+\!t}(x_T)\colonequals\sum_{\mathcal{I}^t} \pi(s^{t\mid T})\tilde{f}^t(s^{t\mid T}, x_T)$. To compute the unconditional expected impulse response, we take the expectation over $x_T$ at the invariant distribution~$\phi$. Since each of the policy functions is monotone in~$\omega$, $\tilde{f}^t(s^{t\mid T}, x_T)\geq f^t(s^{t\mid T}, x_T)$ for each~$t$. That is, along each sample path the consumption following the shock is no lower than without the shock.  

Figure~\ref{fig:demo} uses the parameter values of Example~\ref{exmpl:canon} with a 1\% increase in the population size at~$T+1$ ($\varepsilon=0.01$). It plots the change in the expected consumption share of the young relative to the first-best outcome. At~$T$, the plot shows that the average consumption share is around 0.04\% higher than the first-best consumption share. This positive impact on the consumption share occurs because the participation constraint in state~2 binds at the optimum. As explained above, at~$T+1$, $c^\ast_{T+1}$ responds positively to the shock. With $\beta=\delta$ and $\varepsilon=0.01$, $c^{\ast}_{T+1}\approx 0.5025$. A 1\% shock to the population size increases the consumption share of the young by approximately 0.4975\%, but the increase occurs only in the period of the shock. The share reverts to the long-run value in the period following the shock. Figure~\ref{fig:demo} shows that the average consumption increases relative to the first-best outcome at~$T+1$. That is, the effect of the shock is amplified. 

Moreover, the effect of the shock is prolonged over several periods. This occurs because a positive shock to the population size relaxes the current promise-keeping constraint, leading to increased consumption share for the current young and reduced future promises. As a result, the consumption share of the young is increased in the next period, thereby propagating the impact of the temporary demographic shock across several periods. For sample paths where there is resetting, the average consumption share returns to its long-run value. However, there are other sample paths, occurring with a low probability, where there is no resetting and the average consumption remains above its average at the invariant distribution. Since such sample paths occur with vanishingly small probability, the effect of the shock goes to zero with probability one in the long run.

\begin{figure}[tb]
\begin{center}
\includegraphics{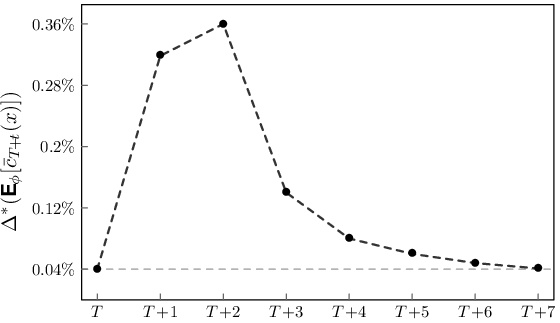}
\caption{Impulse Response Functions for a Demographic Shock that Increases Population Size.} \label{fig:demo}
\end{center}
\begin{figurenotes}
The figure plots the average consumption shares of the young in response to a demographic shock relative to the first-best consumption share: $\Delta^{\ast} (\mathbb{E}_{\phi}[\bar{c}_{T+t}(x)])$ where $\Delta^{\ast} (z_{t})=(z_{t}-c^{\ast}_{t})/c^{\ast}_{t}$. The figure uses the parameter values of Example~\ref{exmpl:canon} and a shock that increases the population size by 1\% at~$T+1$. The positive value at~$T+1$ indicates that the shock is amplified relative to the first best. The effect of the shock is prolonged but eventually the deviation of the average consumption from its long-run average returns to zero. In the first-best case, the effect of the shock lasts for only one period. 
\end{figurenotes}
\end{figure} 


\section{Shooting Algorithm}\label{app:SC}

In the two-state economy in Section~\ref{sec:example}, under Assumption~\ref{ass:canonic}, the optimal consumption share depends only on the number of previous state~2s. Let $c^{(n)}(s)$ denote the consumption share after $n$ consecutive state 2s and let $\mu^{(n)}$ denote the corresponding multiplier on the state~2 participation constraint. It follows from the first-order condition~\eqref{rmu1} and the updating rule~\eqref{eqn:update} that:
\begin{gather*}
c^{(n)}(1) = \frac{\delta}{\beta \nu^{(n-1)}+\delta}\quad\mbox{and}\quad c^{(n)}(2)=\frac{\delta \nu^{(n)}}{\beta \nu^{(n-1)}+\delta\nu^{(n)}}\quad \text{for $n=0,1,2\ldots,\infty$,}
\end{gather*}
where $\nu^{(n)}=1+\mu^{(n)}$ and $\nu^{(-1)}=1$. Let $\nu^{(\infty)}=\lim_{n\to\infty}\nu^{(n)}$. The participation constraint of the young binds in state~2. Hence,
\begin{gather}\label{eqn:D1}
\log\left(\tfrac{\delta\nu^{(n)}}{\beta \nu^{(n-1)}+\delta\nu^{(n)}}\right) +\beta\left(\pi\log\left(\tfrac{\beta\nu^{(n)}}{\beta\nu^{(n)}+\delta}\right)+(1-\pi)\log\left(\tfrac{\beta\nu^{(n)}}{\beta\nu^{(n)}+\delta\nu^{(n+1)}}\right)\right) = \upsilon(2).
\end{gather}
Since equation~\eqref{eqn:D1} holds in the limit,
\begin{gather}\label{eqn:D2}
\log\left(\tfrac{\delta}{\beta+\delta}\right) +\beta\left(\pi\log\left(\tfrac{\beta\nu^{(\infty)}}{\beta\nu^{(\infty)}+\delta}\right)+(1-\pi)\log\left(\tfrac{\beta}{\beta+\delta}\right)\right)
= \upsilon(2).
\end{gather}
Since $\upsilon(2) = \log(s(2))+\beta(\pi \log(1-s(1))+(1-\pi)\log(1-s(2)))$, equation~\eqref{eqn:D2} can be solved to give:
\begin{gather}
\nu ^{(\infty )}=\frac{\delta }{\beta }\left(\!-\!1\!+\!\left(\left( \tfrac{%
\delta }{\beta }\right)^{\frac{1\!-\!\pi}{\pi }} \left( \tfrac{%
\beta\!+\!\delta }{\delta }\right)^{\tfrac{1\!+\!\beta (1\!-\!\pi) }{\beta \pi }%
}\left( {s(2)}\right)^{\tfrac{1}{\beta \pi }}\left(
{1\!-\!s(2)}\right)^{\tfrac{1\!-\!\pi }{\pi }}{(1\!-\!s(1))}%
\right)^{\!-\!1}\right)^{\!-\!1}\text{.}
\label{nuinf}
\end{gather}
Using equations~\eqref{eqn:D1} and~\eqref{eqn:D2}, gives a second-order difference equation for $\nu^{(n)}$:
\begin{gather}\label{eqn:2orddiffv}
\textstyle
\nu^{(n\!+\!1)} \!=\!\textstyle  \tfrac{\beta}{\delta}\nu^{(n)}\!\left(\!-\!1 \!+\! \left(\tfrac{\beta\nu^{(n)}}{\beta\nu^{(n)}\!+\!\delta}\right)^{\!\frac{\pi}{1\shortminus\pi}}\!
\left(\tfrac{\beta\nu^{(\infty)}\!+\!\delta}{\beta\nu^{(\infty)}}\right)^{\!\frac{\pi}{1\shortminus\pi}}\!
\left(\tfrac{\beta\!+\!\delta}{\delta}\right)^{\!\tfrac{1}{\beta(1\shortminus\pi)}}\!\left(\tfrac{\beta\!+\!\delta}{\beta}\right)\!
\left(1\!+\!\tfrac{\beta}{\delta}\tfrac{\nu^{(n\!-\!1)}}{\nu^{(n)}}\right)^{\!-\!\tfrac{1}{\beta(1\shortminus\pi)}}\!\right)\!\text{.}
\end{gather}
It can be shown that the second-order difference equation in~\eqref{eqn:2orddiffv} has a unique saddle path solution. Since $\nu^{(-1)}=1$, the solution can be found by a \emph{forward shooting\/} algorithm to search for $\nu^{(0)}$ such that the absolute difference between $\nu^{(\infty)}$ (given in~\eqref{nuinf}) and $\nu^{(N+1)}$ (given in~\eqref{eqn:2orddiffv}) is sufficiently close to zero for $N$ sufficiently large.

\section{Pseudocode for Numerical Algorithms}\label{app:SD}

Algorithms are implemented in \ac{MATLAB}\textsuperscript{\textscale{0.75}{\textregistered}}. At each iteration, the optimization uses the nonlinear programming solver command {\tt fsolve} in Algorithm 1 and command {\tt fmincon} in Algorithm 2. Value function interpolation uses the spline method of the {\tt interp1} command. In a typical example, the value function converges within 15 iterations. Algorithm~3 shows how the invariant distribution is computed using the results of Proposition~\ref{prop:convergence} and Section~\ref{sec:conv}. Algorithms are expressed in terms of the promise~$\omega$. The policy functions in terms of debt $d$ are obtained by using the transform $\omega=\log(1-s+sd)$.    

\begin{algorithm*}[!htp]
\renewcommand{\thealgorithm}{1:}
\caption{Shooting Algorithm}
\begin{algorithmic}
\Procedure{}{}\Comment{Find $\nu^{(0)}=1+\mu^{(0)}$ in two state economy (Section~\ref{sec:example})}
\State $\mbox{target} \gets \nu^{(\infty)}$ \Comment{Use equation~\eqref{nuinf} in Part~\ref{app:SC}}
    \State $\mbox{tolerance} \gets \epsilon>0$\Comment{$\epsilon=10^{-10}$}
\Repeat
\State $\mbox{initialization} \gets \nu_{0}^{(0)}>0$
\State Compute $\nu_{0}^{(N)}$ for $N=20$ \Comment{Use equation~\eqref{eqn:2orddiffv} in Part~\ref{app:SC}}
\State $d\gets d(\nu_{0}^{(N)},\nu^{(\infty)})$ \Comment{$d(\nu_{0}^{(N)},\nu^{(\infty)})=\abs{\nu_{0}^{(N)}-\nu^{(\infty)}}$}
\Until{$d<\epsilon$}
\State $\nu^{(0)} \gets \nu_{0}^{(0)}$
\EndProcedure
\end{algorithmic}
\end{algorithm*}

\begin{algorithm*}[!htp]
\renewcommand{\thealgorithm}{2:}
\caption{Find Value and Policy Functions}
\begin{algorithmic}
\Procedure{}{}\Comment{Find solution to functional equation~\ref{eqn:p1}}
    \State $\Omega \gets [\omega_{\min}, \omega_{\max}]$ \Comment{$\omega_{\min}$ and $\omega_{\max}$ computed}
    \State $\mbox{gridpoints}\gets gp$ \Comment{Discretize $\Omega$: $gp=200$ Chebyshev interpolation points}
    \State $\mbox{tolerance} \gets \epsilon>0$\Comment{$\epsilon=10^{-6}$}
    \State $J\gets V^\ast$ \Comment{$V^\ast$ is first best}
\Repeat
\State Compute $TJ$ from $J$ \Comment{Use equation~\ref{eqn:p1} and interpolate}
\State $d\gets d(TJ,J)$ \Comment{$d(TJ,J)=\max_\omega\abs{TJ(\omega)-J(\omega)}$}
\State $J\gets TJ$
\Until{$d<\epsilon$}
\State $V \gets J$
\State Compute $g_r(s,\omega)$ and $f(s,\omega)$ \Comment{Using the function $V$ just computed.}
\EndProcedure
\end{algorithmic}
\end{algorithm*}

\begin{algorithm*}[!htp]
\renewcommand{\thealgorithm}{3:}
\caption{Computing the Invariant Distribution}
\begin{algorithmic}
\Procedure{}{}\Comment{Find invariant distribution for $x=(s,\omega)\in \mathcal{X}$}
\State $\mbox{bins} \gets N$ \Comment{Discretize $\Omega$: $N=10^{3}$ bins}
    \State $\mbox{initialization} \gets a_{0}$\Comment{$a_{0}$ is an arbitrary $I\times N$ probability vector.}
\State Compute $a=P a_{0}$ \Comment{Use the $I\times N$ transition matrix $P=P(x,x^\prime)$}
 \State $\mbox{tolerance} \gets \epsilon>0$\Comment{$\epsilon=10^{-8}$}
\Repeat
\State Compute $a=P a$ \Comment{$P((s,n),(r, n^\prime))=\pi(r)\mathbb{1}_{n^{\prime}}g_r(x)$ for $\omega$ in bin $n$.}
\State \Comment{$\mathbb{1}_{n^{\prime}}g_r(x)=1$ $\mathbin{\mathit{iff}}$ $g_r(x)$ is in bin $n^\prime$}
\State $d\gets d(P a,a)$ \Comment{$d(P a,a)=\max_{x}\abs{P(x,x^\prime)a(x)-a(x)}$}
\State $a\gets P a$
\Until{$d<\epsilon$}
\State $\phi \gets a$ \Comment{$\phi=P\phi$ is the invariant distribution}
\EndProcedure
\end{algorithmic}
\end{algorithm*} 

\newpage
\ifx\undefined\bysame
\newcommand{\bysame}{\hskip.3em \leavevmode\rule[.5ex]{3em}{.3pt}\hskip0.5em}
\fi

\end{document}